\documentclass[10pt]{article}
\usepackage{amsmath}
\usepackage{amssymb}
\usepackage{amsthm}
\usepackage{anysize}

\usepackage[round]{natbib}
\usepackage{morefloats}
\usepackage{graphicx}
\usepackage{lscape}
\usepackage{multicol}
\usepackage{multirow}
\usepackage{scalefnt}
\usepackage{color}
\usepackage{subfig}
\marginsize{25mm}{25mm}{25mm}{25mm}

\usepackage{amsmath}%
\usepackage{amsfonts}%
\usepackage{amssymb}%
\usepackage{graphicx}
\usepackage{colonequals,nicefrac}
\usepackage{cleveref}
\usepackage{multicol}
\usepackage{multirow}
\usepackage{bbm}
\usepackage{natbib}
\usepackage{placeins}

\newtheorem{theorem}{Theorem}[section]
\theoremstyle{plain}

\newtheorem{assumption}{Assumption}[section]

\newtheorem{definition}{Definition}[section]

\newtheorem{lemma}{Lemma}[section]

\numberwithin{equation}{section}

\newcommand{\cF}{\mathcal{F}}

\newcommand{\cH}{\mathcal{H}} 

\newcommand{\cV}{\mathcal{V}} 

\newcommand{\fV}{\mathfrak{V}} 

\newcommand{\PP}{\mathbb{P}}
\newcommand{\QQ}{\mathbb{Q}}

\newcommand{\RR}{\mathbb{R}}

\newcommand{\CC}{\mathbb{C}}
\newcommand{\NN}{\mathbb{N}}

\newcommand{\bZ}{\mathbf{Z}}
\newcommand{\bW}{\mathbf{W}}
\newcommand{\bV}{\mathbf{V}}
\newcommand{\ba}{\mathbf{a}}
\newcommand{\bb}{\mathbf{b}}

\newcommand{\bsH}{\boldsymbol{H}}
\newcommand{\bsG}{\boldsymbol{G}}

\newcommand{\tV}{\mathtt{V}} 

\newcommand{\de}{\mathrm{d}}

\renewcommand{\Re}{\mathrm{Re}}
\renewcommand{\Im}{\mathrm{Im}}
\newcommand{\imag}{\mathtt{i}}

\newcommand{\Excond}[3]{\mathbb{E}^{#1}\left[\left.#2\right|#3\right]}  

\begin{document}

\title{{\huge  A Penny Saved is a Penny Earned:\\  Less Expensive Zero Coupon Bonds}}

\setcounter{page}{1}

\author{
\textrm{Alessandro Gnoatto}\thanks{Department of Economics, University of Verona. Via Cantarane 24, 37129 Verona, Italy. Email alessandro.gnoatto@univr.it}
\and\textrm{Martino Grasselli}\thanks{Dipartimento di Matematica, Universit\`a degli Studi di Padova (Italy) and L\'{e}onard de Vinci P\^{o}le Universitaire, Research Center, Finance Group, 92 916 Paris La D\'{e}fense, France. Email: grassell@math.unipd.it.}
\and
\textrm{Eckhard Platen}\thanks{University of Technology, Sydney, Finance Discipline Group and School of Mathematical and Physical Sciences, PO Box 123, Broadway, NSW, 2007, Australia, and University of Cape Town, Department of Actuarial Science. Email eckhard.platen@uts.edu.au.}} 

\maketitle
\begin{abstract}
In this paper we show how to hedge a zero coupon bond with a smaller amount of initial capital than required by the classical risk neutral paradigm, whose (trivial) hedging strategy does not  suggest to invest in the risky assets. Long dated zero coupon bonds we derive, invest first primarily in risky securities and when approaching more and more the maturity date they increase also more and more the fraction invested in fixed income. The conventional wisdom of financial planners suggesting investor to invest in risky securities when they are young and mostly in fixed income when they approach retirement, is here made rigorous.  The paper provides  a strong warning for life insurers, pension fund  managers and long term investors to take the possibility of less expensive products seriously to avoid the adverse consequences of the low interest rate regimes that many developed economies face.
\end{abstract}


\bigskip


\section{Introduction}

This paper aims to draw the attention to a more general modeling world than available under the classical no-arbitrage paradigm in finance.  To explain the new approach and illustrate first important consequences, the less expensive pricing of long dated zero coupon bonds will be demonstrated in this paper.  Under the benchmark approach, see \cite{bookplaten10}, many payoffs can be less expensively produced than current theory and practice suggest. The intuitive verbal advice of financial planners to invest in risky securities when the investor is young and mostly in fixed income when he/she approaches retirement, is made rigorous. The long dated zero coupon bonds we derive, invest first primarily in risky securities and, when approaching more and more the maturity date, they increase also more and more the fraction invested in fixed income. These less expensive zero coupon bonds provide only a first example that is indicative for the changes that the new approach offers in a much wider modeling world than the classical one.\\
Historically, \cite{long90} was the first who observed that one can rewrite the risk neutral pricing formula into a pricing formula that takes its expectation under the real world probability measure and employs  the, so called, num\'eraire portfolio (NP) as num\'eraire.   The benchmark approach assumes only the existence of the NP and no longer the rather restrictive classical no-arbitrage assumptions, which are equivalent to the existence of an equivalent risk neutral probability measure. Under this much weaker assumption, one can still perform all  essential tasks of valuation and risk management.  The only condition imposed is that the NP, which is the in the long run pathwise best performing portfolio, remains finite in finite time. Obviously, when this assumption is violated for a model, then some economically meaningful arbitrage must exist causing  the candidate for the NP to explode. In this case  the respective model makes not much theoretical and practical sense.  Note that when a finite NP exists, various forms of classical arbitrage may be present in the market; see e.g. \cite{Loewenstein00} and \cite{Loewenstein07} for various examples in the literature on bubbles.\\

The current paper illustrates the divergence of the benchmark approach from the classical approach by focusing on the currency market, which is one of the most active markets. 
We present and calibrate a hybrid model describing the dynamics of a vector of foreign exchange (FX) rates and the associated interest rates. We extend and unify the FX multifactor stochastic volatility models of \cite{gnoatto11} and \cite{martinoplaten13} by means of the general transform formula presented in \cite{grasselli13}. The resulting general model that we develop allows for the simultaneous presence of multiple stochastic volatility factors both of square root (see \cite{Heston93}) and 3/2 type (see \cite{heston97} and \cite{platen97}). More explicitly, the square root of each CIR factor appears in both the numerator and the denominator of the diffusion terms. Based on  \citet{martinoplaten13}, we refer to this model as the 4/2 model. Our specification for the volatility process spans a large class of dynamics ranging from the 3/2 to the Heston model. This means that we can let market data dictate the relative importance of the two stochastic volatility effects that we consider. While the 4/2 model might appear as an involved choice, we will show in Subsection \ref{sec:unifying} that it naturally emerges e.g. in a simple Heston setting for a suitable choice of the risk premium.  Moreover, such CIR factors can be freely combined in order to drive stochastic interest rates. Therefore, the model is suitable for the valuation of long-dated FX products, for which interest rate risk becomes a relevant risk factor, see the discussions in  \cite{gg13}. \\
The framework we propose is general to the extent that for suitable parameter combinations our model may not admit the existence of an equivalent risk-neutral probability measure for some economies. In spite of this feature, the problem of pricing and hedging contingent claims can always be solved under the more general benchmark approach of \cite{bookplaten10}: the real world pricing formula (see \cite{bookplaten10}) and benchmarked risk minimization for hedging (see \cite{duplaten14}) will be the tools we employ to solve both problems.\\
Despite of the richness of our framework, it is possible to efficiently solve and implement the pricing of plain vanilla instruments via Fourier-based techniques, see \cite{article_Carr99} and \cite{lewis2001}. Semi analytical closed form solutions for products such as European FX options can be computed thanks to the availability of the exact formula for the joint Fourier transform of the model's state variables, see \cite{grasselli13}. \\

We test our model on real datasets of vanilla FX options by performing several calibration experiments. 
Our empirical results are twofold. On the one side, we confirm the empirical findings of \cite{martinoplaten13} on the violation of the risk neutral pricing paradigm. The appearence of such violations may change over time and across currencies. In fact, our multiple calibration experiment seems to suggest the presence of \textit{regime switches} in traded FX-option prices between the standard risk neutral and the real world pricing  approach. Such a feature calls for a modelling framework which is able to span both valuation principles, which is provided by the benchmark approach. On the other side, we quantify the impact of such violations in a simplifying yet clarifying situation, namely the pricing and hedging of zero coupon bonds in a deterministic interest rate setting. We find that the impact is significant, mostly for long dated bonds, where the difference between the price given by the real world pricing of the benchmark approach and the one provided by the putative (and empirically rejected) risk neutral paradigm, may be very large (about $30\%$ for maturities around 20 years). 
 This surprising result can explain  the effect that has been intuitively exploited in financial planning, where one uses the extra growth present in risky securities to accumulate over long time periods more wealth with little fluctuations at maturity than by investing over the entire period in fixed income, see e.g. \cite{Gourinchas02}. The possibility of  less expensive production costs for targeted long dated payoffs has major practical implications for annuities, life insurance contracts, pensions and many other long term contracts. \\

 The paper is structured as follows: In Section~\ref{sec:model} we introduce the general multi-currency modeling framework and recall some notions from the benchmark approach. Section~\ref{sec:4over2model} motivates and formally introduces the 4/2 model as a unifying framework for stochatic volatility models driven by the CIR process as the Heston-based model of \cite{gnoatto11} and the 3/2-based model of  \cite{martinoplaten13}. The 4/2 model extends the Heston model and allows for the possibility of a failure of the risk neutral paradigm. The analytical tractability of the 4/2 model is demonstrated in Section~\ref{sec:valuationOfDerivatives}, which constitutes a prerequisite for an efficient model calibration, presented in Section~\ref{sec:calibrationResults}. A practical consequence of the calibration is analyzed in Section~\ref{sec:hedgingLongDatedSecurities}, where we perform hedging in an incomplete market setting without the existence of a risk neutral probability measure. In particular, in the absence of a risk neutral probability measure, we follow  the concept of benchmarked risk minimization of \cite{duplaten14}. We show how to hedge long term products significantly less expensively than the classic risk neutral paradigm permits.

\section{General Setup}\label{sec:model}

In this section we present the general modeling framework of  the benchmark approach for a foreign exchange (FX) market. Subsection \ref{sec:specification} provides a general setup driven by a multi-dimensional diffusion process. 

\subsection{Specification of the Currency Market}\label{sec:specification}


We use superscripts to reference different currencies, and  employ bold letters for vectors and subscripts for elements thereof. Unless specified by a suitable superscript, all expectations are considered with respect to  the real world probability measure $\PP$.
We model the currency market on a probability space $\left(\Omega,\cF_{\bar T},\PP\right)$, where $\bar T<\infty$ is a finite time horizon. On this space we introduce a filtration $(\cF_t)_{0\leq t\leq \bar T}$ to model the evolution of information, satisfying the usual assumptions. The above filtered probability space supports a standard $d$-dimensional $\PP$-Brownian motion $\bZ$=$\{\bZ (t)=(Z_1(t),\ldots , Z_d(t)),\ 0\leq t\leq \bar T\}$ for modeling the traded uncertainty. The constant $N$ denotes the number of currencies in the model, whereas $d$ is the number of traded risky factors we employ. \\

In each economy we postulate the existence of a money market account, i.e. the $i$-th money market account,when  denominated in units of the $i$-th currency, evolves according to the relation
\begin{align}
\de B^i(t) =B^i(t)r^i(t)\de t, \ B^i(0)=1, \ 0\leq t \leq \bar T;
\end{align}
with the $\RR$-valued, adapted $i$-th short rate process $r^i=\left\{r^i(t), \ 0 \leq t\leq \bar T \right\}$.
We denote by $S^{i,j}$$=\left\{S^{i,j}(t), \ 0 \leq t\leq \bar T \right\}$ the continuous exchange rate process between currency $i$ and $j$. Here $S^{i,j}(t)$ denotes the price of one unit of currency $i$ in units of currency $j$, meaning that, e.g. for $i=USD$ and $j=EUR$ and $S_t^{i,j}=0.92$ we have, in line with the standard FORDOM convention, that the price of one USD is $0.92$ EUR at time $t$.

Let us follow \cite{bookplaten10} and \cite{currencyplaten06} and introduce a family of primary security account processes via
\begin{align*}
B^{i,j}(t)=S^{i,j}(t)B^j(t), \ 0\leq t\leq \bar T
\end{align*}
for $i\neq j$. Obviously, for $i=j$ we have $B^{i,i}(t)=B^{i}(t)$. We take the perspective of a generic currency referenced with superscript $i$ and introduce the vector of money market accounts of the form $\boldsymbol{B}^i(t)=\left(B^{i,1}(t),...,B^{i,N}(t)\right)$, $i=1,...,N$. Given this vector of primary security accounts, an investor may trade on them. This is represented by introducing a family of predictable $\boldsymbol{B}^i$-integrable stochastic processes $\boldsymbol\delta=\left\{\boldsymbol\delta (t)=\left(\delta_1(t),...,\delta_N(t)\right), \ 0 \leq t\leq \bar T \right\}$ for $i=1,...,N$, called strategies. Each $\delta_j(t)\in \RR$ denotes the number of units that an agent holds in the $j$-th primary security account at time $t$. Let us introduce the process ${\mathtt{V}}^{i,\delta}=\left\{\mathtt{V}^{i,\delta}(t), \ 0 \leq t\leq \bar T \right\}$, which describes the value process in $i$-th currency denomination corresponding to the portfolio strategy  $\boldsymbol\delta$, i.e.
\begin{align}\label{eq:portfolio}
\mathtt{V}^{i,\delta}(t)=\sum_{j=1}^N\delta_j(t)B^{i,j}(t).
\end{align} 
The strategy $\delta$ is said to be self-financing if
\begin{align}\label{eq:portfolioSefFinancing}
\de \mathtt{V}^{i,\delta}(t)=\sum_{j=1}^N\delta_j(t)\de B^{i,j}(t).
\end{align} 
In line with \cite{bookplaten10}, we assume limited liability for all investors. For this purpose, we introduce $\cV^+$ as the set of all self-financing strategies forming strictly positive portfolios. 
For our purposes, we will be interested in a particular strategy $\delta^\star\in\cV^+$, which yields the  \textbf{growth optimal portfolio} (GOP), which can be shown to be equivalent to the num\'eraire portfolio (NP), and is defined as follows:
\begin{definition}\label{def:gop} A solution $\delta^{\star}$ of the maximization problem
\begin{align*}
\sup_{\delta\in \cV^+}\mathbb{E}\left[\log\left(\frac{\mathtt{V}^{i,\delta}(T)}{\mathtt{V}^{i,\delta}(0)}\right)\right],
\end{align*}
for all $i = 1,\ldots ,N$  and $0\leq T\leq {\bar T}$ is called a growth optimal portfolio strategy.
\end{definition}
It has been shown in \cite{bookplaten10} that the GOP value process is unique in an incomplete jump-diffusion market setting. We summarize the discussion above in the following assumptions.
\begin{assumption}\label{assumption:gop}
We assume the existence of the \textbf{growth optimal portfolio} (GOP), and denote by $D^i=\left\{D^i(t)\in (0,+\infty), \ 0 \leq t\leq \bar T \right\}, \ i=1,...,N$ the value of the GOP denominated in the $i$-th currency. The dynamics of the GOP are given by
\begin{align}
\begin{aligned}
\frac{\de D^i(t)}{D^i(t)}=  r^i(t)\de t+\langle \boldsymbol{\pi}^i(t),\boldsymbol{\pi}^i(t)\de t+\de\boldsymbol{Z}(t)\rangle, \ D^i(0)>0 
\label{eq:sdegop}
\end{aligned}
\end{align}
for $t\in [0, \bar T]$ and $i = 1,\ldots N$, where for $N,d\in \NN$, the $N$-dimensional family of predictable, $\RR^d$-valued stochastic processes $\boldsymbol{\pi}=\left\{\boldsymbol{\pi}^i(t)=(\pi^i_1(t),\ldots ,\pi^i_d(t)), \ 0 \leq t\leq \bar T \right\}$ represent the market prices of risk with respect to the $i$-th currency denomination. The processes $\boldsymbol{\pi}^i$ are assumed to be integrable with respect to the $d$-dimensional standard Brownian motion $\bZ$.
\end{assumption}

The GOP can be shown to be in many ways the \textit{best} performing portfolio. In particular, in the long run its value outperforms almost surely those of any other strictly positive portfolio. Here we assume that it remains finite in finite time in all currency denominations. If we were to consider a model where the GOP explodes in any of the currency  denominations, then the model would allow an obvious form of economically meaningful arbitrage, since one could generate in that currency denomination in finite time unbounded wealth from finite initial capital. Given the uniqueness of the GOP, all exchange rates $S^{i,j}(t)$ can be  uniquely determined as ratios of different denominations of the GOP in the respective currencies.
\begin{assumption}\label{AnsatzFX}
The family of exchange rate processes $S^{i,j}=\left\{ S^{i,j}(t), \ 0 \leq t\leq \bar T\right\}$ is determined by the ratios
\begin{align}
S^{i,j}(t)=\frac{D^i(t)}{D^j(t)}
\end{align}
for $0 \leq t \leq \bar T$ and $i,j=1,\ldots,N$.
\end{assumption}
Given Assumption \ref{AnsatzFX}, it is immediate to compute via a direct application of the It\^{o} formula the dynamics of all exchange rates and all primary security accounts in all currency denominations.
\begin{lemma}
The exchange rate $S^{i,j}(t)$ under the real world probability measure $\PP$ evolves according to the dynamics
\begin{align}
\begin{aligned}
\label{exchangerates}
\frac{\de S^{i,j}(t)}{S^{i,j}(t)}&= ( r^i(t)-r^j(t))\de t+\langle \boldsymbol{\pi}^i(t)-\boldsymbol{\pi}^j(t),\boldsymbol{\pi}^i(t)\de t+\de\boldsymbol{Z}(t)\rangle,\\
  S^{i,j}(0)&=s^{i,j}>0, 
  \end{aligned}
\end{align}
and the generic $j$-th primary security account $B^{i,j}$, in $i$-th currency denomination and under the real world probability measure $\PP$, evolves according to the dynamics
\begin{align}
\begin{aligned}
\label{eq:primarySecuritySDE}
\frac{\de B^{i,j}(t)}{B^{i,j}(t)}&= r^i(t)\de t+\langle \boldsymbol{\pi}^i(t)-\boldsymbol{\pi}^j(t),\boldsymbol{\pi}^i(t)\de t+\de\boldsymbol{Z}(t)\rangle, \\
  B^{i,j}(0)&=b^{i,j},
  \end{aligned}
\end{align}
for $\ i,j = 1,\ldots,N$ and $t\in[0,\bar T]$.
\end{lemma}

\subsection{The Benchmark Approach}\label{sec:benchmarkApproach}
In the present paper we evaluate contingent claims under the \textit{benchmark approach} of \cite{bookplaten10}. Under this approach, price processes denominated  in terms of the GOP  are called benchmarked price processes. More precisely, for $i,j=1,\ldots N$, let us introduce the benchmarked price process
\begin{align*}
\hat{B}^{j}=\left\{\hat{B}^{j}(t):=\frac{B^{i,j}(t)}{D^i(t)},\ 0\leq t\leq \bar{T}\right\}.
\end{align*}
 We call $\hat{B}^{j}$ the \textit{benchmarked $j$-th primary security account}. Note that $\hat{B}^{j}$ does not depend on the index $i$ of the currency denomination we started from. Given \eqref{eq:primarySecuritySDE} and \eqref{eq:sdegop}, upon an application of the It\^{o} formula, it is immediate to conclude that all benchmarked price processes $\hat{B}^{j}$ form $\PP$-local martingales. Even more, they are non-negative $\PP$-local martingales. Hence, due to Fatou's lemma, they are also $\PP$-supermartingales. Analogously, we also have that benchmarked non-negative portfolio values  $\hat{\tV}^{\delta}(t):={\tV}^{i,\delta}(t)/D^i(t)$ form $\PP$-supermartingales. 
Besides the exclusion of forms of economically meaningful arbitrage, which  are equivalent to the explosion of the GOP,  forms of classical arbitrage that are excluded under classical no-arbitrage assumptions may exist in our model, see e.g.  \cite{Loewenstein00}.

Let us now introduce for the $i$-th currency denomination the Radon-Nikodym derivative process, denoted by $\Lambda^i=\left\{\Lambda^i(t), \ 0\leq t \leq \bar{T}\right\}$, by setting 
\begin{align}
\Lambda^i(t)=\frac{\hat{B}^{i}(t)}{\hat{B}^{i}(0)},\quad i=1,\ldots, N\label{radoni}.
\end{align}
This is the risk neutral density for the putative risk-neutral measure $\mathbb{Q}^i$ of the $i$-th currency denomination. It arises e.g. when we consider replicable claims  and assume the existence of an equivalent risk neutral probability measure $\mathbb{Q}^i$. As each $\Lambda^i$ equals  the corresponding benchmarked savings account $\hat{B}^{i}$ (up to a constant factor), it is clear that $\Lambda^i$ is a $\PP$-local martingale, for $i=1,\ldots, N$. The classical assumption in the foreign exchange literature that there exists an equivalent risk-neutral probability measure for each currency denomination, corresponds to the requirement that each process $\Lambda^i$ is a true martingale for $i=1,\ldots, N$. Such a requirement is rather strong and may be empirically rejected, see e.g. \cite{currencyplaten06}, the findings in \cite{martinoplaten13} and the necessary and sufficient conditions of \cite{ruf15}. Hence, in the present paper, we shall allow each $\Lambda^i$ to be either a true martingale or a strict local martingale. To work in  such a generalized setting requires a more general pricing concept than the one provided under the classical risk neutral paradigm. In the following, we will employ the notion of \textit{real world pricing}: a price process $\fV^i=\left\{\fV^i(t), 0\leq t \leq \bar{T}\right\}$, here denominated in $i$-th currency, is said to be \textit{fair} if, when expressed in units of the GOP $D^i$, forms a $\PP$-martingale, this means its benchmarked value forms a true $\PP$-martingale, see Definition 9.1.2 in \cite{bookplaten10}. For a fixed maturity $T\in [0, \bar{T}]$, we let $\cH(T)=\fV(T)$ be an $\cF_T$-measurable non-negative benchmarked contingent claim, such that when expressed in units of the $i$-th currency denomination as  $\cH^i(T)=D^i(T)\cH(T)$ we have
\begin{align*}
\mathbb{E}\left[\left.\cH(T)\right|\cF_t\right]=\mathbb{E}\left[\left.\frac{\cH^i(T)}{D^i(T)}\right|\cF_t\right]<\infty
\end{align*}
for all $0\leq t\leq T\leq \bar{T}$, $i=1,...,N$. The benchmarked fair price ${\hat V}(t)=V^i(t)/D^i(t)$ of this contingent claim is the minimal possible price and given by the following conditional expectation under the real-world probability measure $\PP$:
\begin{align}
\label{eq:realWorldPricingFormula}
{\hat \fV(t)}=\mathbb{E}\left[\left.{\hat \cH(T)}\right|\cF_t\right],
\end{align}
which is known in the literature as \textit{real-world pricing formula}, see Corollary 9.1.3 in \cite{bookplaten10}. Note that benchmarked risk minimisation, described in  \cite{duplaten14}, gives \eqref{eq:realWorldPricingFormula} generally. In case $\Lambda^i$  is a true martingale we obtain, by changing in \eqref{eq:realWorldPricingFormula} from the real world probability measure $\PP$ to the equivalent risk neutral probability measure $\mathbb{Q}^i$, the risk neutral pricing formula 
\begin{align}
\label{eq:realWorldEqualsRiskNeutral}
\begin{aligned}
\fV^i(t)&=\mathbb{E}\left[\left.\frac{B^i(t)}{B^i(T)}\frac{B^i(T)}{B^i(t)}\frac{D^i(t)}{D^i(T)}\cH^i(T)\right|\cF_t\right]\\
&=\mathbb{E}\left[\left.\frac{\Lambda^i(T)}{\Lambda^i(t)}\frac{B^i(t)}{B^i(T)}\cH^i(T)\right|\cF_t\right]=\Excond{\QQ^i}{\frac{B^i(t)}{B^i(T)}\cH^i(T)}{\cF_t}.
\end{aligned}
\end{align}
This shows that in this case the real-world pricing formula generalises  the classical risk-neutral valuation formula and the Radon-Nikodym derivative for the respective risk neutral probability measure is given by \eqref{radoni}. In general, due to the supermartingale property of benchmarked price processes in the case when $\Lambda^i$ is a strict supermartingale, a formally obtained risk neutral price is greater than or equal to the real world price, see  \cite{duplaten14}.\\

Later on we consider the pricing of zero coupon bonds. According to \eqref{eq:realWorldPricingFormula}, the minimal possible price $P^i(t,T)$ at time $t$ of a zero coupon bond, denominated in the $i$-th currency and paying at maturity $T$ one unit of the $i$-th currency, is given by the formula 
\begin{align}
\label{zerocouponbond}
P^i(t,T)=D^i(t)\mathbb{E}\left[\left.\frac{1}{D^i(T)}\right|\cF_t\right].
\end{align}

All benchmarked non-negative price processes are supermartingales. Therefore, as already mentioned, the fair price process of a contingent claim is the minimal possible price process. There may exist more expensive price processes that deliver the same payoff. As we will show later on, over long periods of time a formally obtained  risk neutral price can be significantly more expensive than the fair one. This paper focuses on the emerging possibility to potentially produce long dated zero coupon bonds less expensively than suggested under the classical paradigm. This permits, e.g. the less expensive production of contracts involving long dated zero coupon bonds as building blocks as in the case for pensions, annuities and life insurances.

\section{The 4/2 Model}\label{sec:4over2model}

To demonstrate the fact that in reality there may exist hedgeable long dated zero coupon bond price processes that are significantly less expensive than respective risk neutral price processes, we need some model that could capture this phenomenon when it is present in the market.
In this section we provide such model, called the 4/2 model. Below we describe the 4/2 model that unifies naturally several well-known models. In Subsection \ref{subsecstrictlm} we state the  precise conditions under which the crucial martingale property of the benchmarked savings account fails for the 4/2 model.

\subsection{The 4/2 Model as a Unifying Framework}\label{sec:unifying}

In this subsection, we  provide a  specification of the volatility dynamics of the exchange rates. In particular, we consider, for simplicity, two currencies with $D^1(t)$ denoting the GOP in domestic currency and $D^2(t)$ denoting the GOP in foreign currency. For example,  $S^{1,2}(t)=D^1(t)/D^2(t)$ can follow a stochastic volatility model of Heston type (see \cite{Heston93}), where
\begin{align}\label{eq:basicHestonDynamics}
\begin{aligned}
\frac{\de S^{1,2}(t)}{S^{1,2}(t)} &= \left(r^1(t) - r^2(t)\right) \de t + \sqrt{V(t)}\left(\de Z(t) + \lambda(t)\de t\right),\\
S^{1,2}(0) &= s^{1,2}>0,\\
\de V(t)=& \kappa ( \theta - V(t))\de t + \sigma V(t)^{1/2} \left(\rho \de Z(t)+\sqrt{1-\rho^2}dZ^{\bot}(t)\right),\\
V(0) &= v>0.
\end{aligned}
\end{align}
Here $Z^{\bot}=\left\{Z^{\bot}(t), \ 0\leq t \leq \bar T\right\}$ is a $\PP$-Brownian motion independent of $Z$, $\kappa>0,\theta>0,\rho\in [-1,1]$ with the predictable processes $r^1,r^2$ and $\lambda$. 

In the remainder of the paper, we will repeatedly employ the following terminology: \textit{Heston (type) model, 3/2 (type) model, 4/2 (type) model}. Let us now clarify the respective models. In the following, we consider $a,b\in\RR$, and let $D^i=\left\{D^i(t),\ 0\leq t\leq \bar T\right\}$ denote a generic place-holder for a GOP process satisfying a scalar diffusive stochastic differential equation (SDE). Moreover, let $V=\left\{V(t),\ 0\leq t\leq \bar T\right\}$ be a square root process as given in \eqref{eq:basicHestonDynamics}.  A model is said to be of \textit{Heston type} (resp. of \textit{3/2 type}, resp. of \textit{4/2 type})  if the diffusion coefficient in the dynamics of the GOP $D^i$ is proportial to $a\sqrt{V(t)}$ (resp.  $b/\sqrt{V(t)}$, resp.  $(a\sqrt{V(t)} + b/\sqrt{V(t)})$).\\

The question we would like to address is the following: Given a specification of the market price of risk process $\lambda$ for the domestic currency denomination of securities, \textbf{what are the associated dynamics of the domestic and foreign  specifications of the GOP?}

\begin{lemma}\label{lemmacasilambda}
Consider a two-currency model, where the exchange rate model for  $S^{1,2}$ is of Heston type \eqref{eq:basicHestonDynamics}. The following statements hold true:
\begin{enumerate}
\item \label{lemma1:point1}If $\lambda(t)=a\sqrt{V(t)}$, $a\in\RR$, then the GOP denominations $D^1$ and $D^2$ follow both Heston-type models.
\item \label{lemma1:point3}If $\lambda(t) = \frac{b}{\sqrt{V(t)}}$, $b\in\RR$, then the GOP denomination $D^1$ follows a $3/2$ model, whereas $D^2$ follows a $4/2$ model.
\item \label{lemma1:point2}If $\lambda(t) = a\sqrt{V(t)}+ \frac{b}{\sqrt{V(t)}}$, $a,b\in\RR$, then the GOP denominations $D^1$ and $D^2$ follow both $4/2$ type models.
\end{enumerate} 
\end{lemma}

The proof for this result is given in Appendix \ref{Appendixlemmacasilambda}.

Note that if we had started in \eqref{eq:basicHestonDynamics} with the volatility $1/\sqrt{V(t)}$, then for $\lambda(t) = \frac{b}{\sqrt{V(t)}}$ we would have always fallen  into the world of 4/2 type models.

Lemma \ref{lemmacasilambda} highlights an interesting interplay between several well-known  financial models. It shows that the $4/2$ model arises naturally from a standard Heston model when the market price of risk belongs to the \textit{essentially affine} class (see \cite{duf02}). Furthermore, it demonstrates that the $4/2$ model provides a general framework that nests other popular model choices. 

Different specifications of the market price of risk do not only impact on the shape of the GOP dynamics. In fact, depending on the  calibrated values of the model parameters, we may incur situations where classical  risk neutral pricing is no longer possible because an equivalent risk neutral probability measure does not exist. To see this, we observe that from a direct inspection of the dynamics in 
\eqref{eq:basicHestonDynamics} in the Heston model setting, it is tempting to define the following two continuous processes
\begin{align}
\label{eq:candidateBMs}
\begin{aligned}
Z^{\QQ^{1}}(t)&:=Z(t)+\int_0^t\lambda(s)\de s\\
Z^{\QQ^{2}}(t)&:=Z(t)+\int_0^t\left(\lambda(s)-\sqrt{V(s)}\right)\de s,\\
\end{aligned}
\end{align} 
which, if the assumptions of the Girsanov theorem were in both cases fulfilled, would then be $\QQ^1$- (resp. $\QQ^2$-) Brownian motions. Let us assume that under the real world probability measure $\PP$ the Feller condition (see \cite{kar91}, Section 5.5) is fulfilled by the parameters of the volatility process $V$, i.e. we have $2\kappa\theta-\sigma^2\geq0$, so that the square root process $V$ remains strictly positive $\PP$-a.s. for all $ t\in[0,\bar T]$. The following lemma shows that, depending on the specification of $\lambda$, it is possible to obtain a variance process $V$  under the putative risk neutral measure that may not satisfy the Feller condition, implying that the putative risk neutral measure may fail to be equivalent to the real world probability measure.

\begin{lemma}
Consider a two-currency model, where the dynamics of the exchange rate $S^{1,2}$ is of the Heston type \eqref{eq:basicHestonDynamics}, such that the variance process $V$ fulfills the Feller condition, i.e. $2\kappa\theta-\sigma^2\geq0$. Let $Z^{\QQ^{1}},Z^{\QQ^{2}}$, as in \eqref{eq:candidateBMs}, be the candidate Brownian motions under the putative risk neutral measures $\QQ^1$ and $\QQ^2$, respectively. The following holds:
\begin{enumerate}
\item For the putative risk neutral measure $\QQ^1$ we get:
\begin{enumerate}
\item \label{lemma2:point11} If $\lambda(t)=a\sqrt{V(t)}$, $a\in\RR$, then the drift of the variance process $V$ under $\QQ^1$ is 
\begin{align*}
\kappa\left(\theta-V(t)\right)-\sigma\rho aV(t)
\end{align*}
and the Feller condition is always satisfied under $\QQ^1$, which is then a true equivalent martingale measure.
\item \label{lemma2:point12} If $\lambda(t) = a\sqrt{V(t)}+ \frac{b}{\sqrt{V(t)}}$, $a,b\in\RR$, then the drift of the variance process $V$ under $\QQ^1$ equals 
\begin{align*}
\kappa\left(\theta-V(t)\right)-\sigma\sqrt{V(t)}\rho\left(a\sqrt{V(t)}+ \frac{b}{\sqrt{V(t)}}\right)
\end{align*}
and the Feller condition may be violated, implying that $\QQ^1$ may not be  equivalent to $\PP$.
\item \label{lemma2:point13} If $\lambda(t) = \frac{b}{\sqrt{V(t)}}$, $b\in\RR$, then the drift of the variance process $V$ under $\QQ^1$ is 
\begin{align*}
\kappa\left(\theta-V(t)\right)-\sigma\rho b
\end{align*}
and the Feller condition may be violated, implying that $\QQ^1$ may not be  equivalent to $\PP$.
\end{enumerate}
\item For the putative risk neutral measure $\QQ^2$ we have
\begin{enumerate}
\item \label{lemma2:point21} If $\lambda(t)=a\sqrt{V(t)}$, $a\in\RR$,  then the drift of the variance process $V$ under $\QQ^2$ equals 
\begin{align*}
\kappa\left(\theta-V(t)\right)-\sigma\rho (a-1)V(t)
\end{align*}
and the Feller condition is always satisfied under $\QQ^2$, which is then a true equivalent martingale measure.
\item \label{lemma2:point22} If $\lambda(t) = a\sqrt{V(t)}+ \frac{b}{\sqrt{V(t)}}$, $a,b\in\RR$,  then the drift of the variance process $V$ under $\QQ^2$ equals 
\begin{align*}
\kappa\left(\theta-V(t)\right)-\sigma\sqrt{V(t)}\rho\left((a-1)\sqrt{V(t)}+ \frac{b}{\sqrt{V(t)}}\right)
\end{align*}
and the Feller condition may be violated, implying that $\QQ^2$ would be in such case not equivalent to $\PP$.
\item \label{lemma2:point23} If $\lambda(t) = \frac{b}{\sqrt{V(t)}}$, $b\in\RR$,  then the drift of the variance process $V$ under $\QQ^2$ is 
\begin{align*}
\kappa\left(\theta-V(t)\right)-\sigma\rho b +\sigma\rho V(t)
\end{align*}
and the Feller condition may be violated, implying that $\QQ^2$ would be in such case not equivalent to $\PP$.
\end{enumerate}
\end{enumerate}
\end{lemma}

The proof of these statements is straightforward using our previous notation and relationships and, therefore, omitted.

\subsection{Formal Presentation of the 4/2 Model}\label{sec:modelPresentation}

 To provide in the case of more than two currencies  a concrete specification of the market prices of risk,  we proceed to introduce the $\RR^d$-valued nonnegative stochastic process, called the volatility factor process, that we denote by $\bV=\left\{\bV(t)= (V_1(t),...,V_d(t)), 0\leq t \leq \bar T\right\}$ . The $k$-th component $V_k$ of the vector process $\bV$ is assumed to solve the  SDE
\begin{align}
\begin{aligned}
\de V_k(t)=& \kappa_k ( \theta_k - V_k(t))\de t + \sigma_k V_k(t)^{1/2} \de W_k(t),\\
V_k(0) &= v_k>0,
\label{CIRs}
\end{aligned}
\end{align}
for $t\in[0, \bar T]$, where the parameters $\kappa_k>0,\theta_k>0,\sigma_k>0, $ are admissible in the sense of \cite{article_DFS}, $k=1,...,d$. In addition, to avoid zero volatility factors, we impose the following assumption.
\begin{assumption}
For every $k=1,\ldots , d,$ the parameters in \eqref{CIRs} satisfy the relation
\begin{align}
2\kappa_k\theta_k-\sigma^2_k\geq 0.
\end{align}
\end{assumption}
We also allow for non-zero correlation between assets and their volatilities via the following condition:
\begin{assumption}\label{assumptionCorrelation} The Brownian motions  $\bZ$ and $\bW$ have a covariation satisfying
\begin{align}
\frac{\de \langle W_k,Z_l\rangle(t)}{\de t}=\delta_{kl}\rho_k, \ k,l=1,..., d,
\end{align}
where $\delta_{kl}$ denotes the Dirac delta function for the indices $k$ and $l$.
\end{assumption}
We then proceed to provide a general specification for the family of market prices of risk.
\begin{assumption}\label{piandr} We assume that the  $i$-th market price of risk vector $\boldsymbol{\pi}^i(t)$ is a projection of the common volatility factor $\bV$, along a direction parametrized by a constant vector $\ba^i \in \RR^d$ and a projection of the inverted elements of $\bV$ along another direction parametrized by $\bb^i\in \RR^d$, according to the following relations
\begin{align}
\boldsymbol{\pi}^i(t)=Diag^{1/2}(\boldsymbol{V}(t))\boldsymbol{a}^i +Diag^{-1/2}(\boldsymbol{V}(t))\boldsymbol{b}^i ,\ \ i=1,...,N,
\end{align}
where $Diag^{1/2}(\boldsymbol{u})$ denotes the diagonal matrix whose diagonal entries are the respective square roots of the components of the vector $\boldsymbol{u}\in \mathbb R^d$.
The family of short-rate processes $r^i, i=1,..., N$ is assumed to be given in the form
\begin{align}
\label{eq:shortRateDynamics}
r^i(t)=h^i+\langle \boldsymbol{H}^i,\bV(t)  \rangle+\langle \boldsymbol{G}^i,\bV^{-1} (t) \rangle,
\end{align}
where $\bV^{-1} $ is a vector whose components are the inverses of those of $\bV$.
\end{assumption}
Under Assumption \ref{piandr}, we can express the dynamics of the GOP as
\begin{align*}
\frac{\de D^i(t)}{D^i(t)}=& \left( r^i(t)+(\boldsymbol{a}^i)^\top Diag(\boldsymbol{V}(t))\boldsymbol{a}^i+(\boldsymbol{b}^i)^\top Diag^{-1}(\boldsymbol{V}(t))\boldsymbol{b}^i
+2(\boldsymbol{a}^i)^\top \boldsymbol{b}^i\right)\de t\\
&+ (\boldsymbol{a}^i)^\top Diag^{1/2}(\boldsymbol{V}(t))d\boldsymbol{Z}(t)+(\boldsymbol{b}^i)^\top Diag^{-1/2}(\boldsymbol{V}(t))d\boldsymbol{Z}(t).
\end{align*}
Here  we suppress the explicit formulation of the dependence of $r^i$ on $\boldsymbol{V}$. Consequently, the dynamics of the exchange rate $S^{i,j}$ is given by the SDE
\begin{align}\label{exchangerates}
\begin{aligned}
\frac{\de S^{i,j}(t)}{S^{i,j}(t)}=& \Big((r^i(t)-r^j(t))+2(\boldsymbol{a}^i)^\top \boldsymbol{b}^i -(\boldsymbol{a}^i)^\top \boldsymbol{b}^j-(\boldsymbol{a}^j)^\top \boldsymbol{b}^i\Big.\\ 
&\Big.+(\boldsymbol{a}^i)^\top Diag(\boldsymbol{V}(t))(\boldsymbol{a}^i-\boldsymbol{a}^j)
+(\boldsymbol{b}^i)^\top Diag^{-1}(\boldsymbol{V}(t))(\boldsymbol{b}^i-\boldsymbol{b}^j)\Big)\de t\\ 
&+ ((\boldsymbol{a}^i-\boldsymbol{a}^j)^\top Diag^{1/2}(\boldsymbol{V}(t))+(\boldsymbol{b}^i-\boldsymbol{b}^j)^\top Diag^{-1/2}(\boldsymbol{V}(t)))\de\boldsymbol{Z}(t).
\end{aligned}
\end{align}
Notice that the dynamics of the exchange rates are fully functionally symmetric w.r.t. the construction of product/ratios thereof.
\subsection{Strict Local Martingality} \label{subsecstrictlm}

In Section \ref{sec:unifying}, we observed that, for a sufficiently general specification of the market price of risk for a given currency, we can have a situation  where an equivalent risk neutral probability  measure may fail to exist, due to the   behavior of the  variance process near zero under the putative risk neutral measure. 

In this subsection, we investigate the conditions under which the $i$-th benchmarked savings account, $\hat{B}^i(t)=\frac{B^i(t)}{D^i(t)}$, is a strict $\mathbb P$-local martingale, $i=1,...,N$. As observed in Section \ref{sec:benchmarkApproach}, $\hat{B}^i(t)$, after normalization to one at the initial time, corresponds to the Radon-Nikodym derivative for the putative risk neutral measure of the $i$-th currency denomination. Should $\hat{B}_i(t)$ be a strict $\mathbb P$-local martingale, we note that classical risk neutral pricing is not applicable. However,  real world pricing in  line with \eqref{eq:realWorldPricingFormula} is still applicable, see \cite{bookplaten10}, and provides the minimal possible price.

Given \eqref{eq:primarySecuritySDE} and \eqref{eq:sdegop}, the dynamics of $\hat{B}^i(t)$ are given by the SDE
\begin{align}
\begin{split}
\de \hat{B}^i (t) &= - \hat{B}^i(t)((\boldsymbol{a}^i)^\top (Diag(\boldsymbol{V}(t)))^{1/2}\de\boldsymbol{Z}(t)\\
&\quad+(\boldsymbol{b}^i)^\top (Diag(\boldsymbol{V}(t)))^{-1/2}\de\boldsymbol{Z}(t)).
\end{split}\label{putativeRadonNikodym}
\end{align}
Upon integration of the above SDE  we obtain
\begin{align*}
\mathbb{E} \left[ \hat{B}^i (t) \right]&
=\hat{B}^i_0 \prod^d_{k=1} \mathbb{E} \left[ \xi_k^i (t)\right], 
\end{align*}
where we define the exponential local martingale process $\xi_k^i = \left\{ \xi_k^i(t) \, , \, t \geq 0 \right\}$ via
\begin{align} \label{eqRNprocess}
\begin{split}
\xi_k^i (t)&:= \exp \left\{ -  \rho_k \int^t_0 (a^i_kV_k(s)^{1/2}+b^i_kV_k(s)^{-1/2}) \de W_k(s) \right.\\
&\quad\quad\quad\left.- \frac{1}{2}  \rho_k^2 \int^t_0
(a^i_kV_k(s)^{1/2}+b^i_kV_k(s)^{-1/2})^2 \de s \right\} \, .
\end{split}
\end{align}
The putative change of measure with respect to the $i$-th currency denomination  involves
\begin{align*}
\de \tilde{W}_k(t) = dW_k(t)+ \rho_k (a^i_kV_k(t)^{1/2}+b^i_kV_k(t)^{-1/2}) \de t,
\end{align*}
where under classical assumptions $\tilde{W}_k$ should be a Wiener process under the putative risk neutral measure $\mathbb{Q}^i$.
Under this measure the process $V_k$ would  then solve the SDE
\begin{align} \label{varprocess2}
\de V_k(t) &= \kappa_k  ( \theta_k - V_k(t)) \de t - \rho_k \sigma_k (a^i_kV_k(t)+b^i_k)\de t+ \sigma_k V_k(t)^{\frac{1}{2}} \de\tilde{W}_k(t)\nonumber \\
&= (\kappa_k  \theta_k- \rho_k \sigma_k b^i_k )\de t -\kappa_k (1+\frac{\rho_k \sigma_k a^i_k}{\kappa_k})V_k(t)\de t+ \sigma_k V_k(t)^{\frac{1}{2}} 
\de \tilde{W}_k(t). 
\end{align}
Under $\mathbb{P}$, the process $V_k$ does not reach 0 if the Feller condition is satisfied, i.e.
\begin{equation*}
2\kappa_k \theta_k \geq \sigma_k^2,
\end{equation*}
while under the putative risk neutral  measure  the process  $V_k$ would not reach 0 if the corresponding Feller condition would be satisfied, that is 
\begin{equation*}
2\kappa_k \theta_k \geq\sigma_k^2+2\rho_k \sigma_k b^i_k.
\end{equation*}
Therefore, the process $V_k$ would have a different behavior at 0 under the two measures, provided that
\begin{equation}
\sigma_k^2\leq 2\kappa_k \theta_k <\sigma_k^2+2\rho_k \sigma_k b^i_k.
\label{eq:strictLocalMartingaleTest}
\end{equation}

In this case the putative risk neutral measure would not be an equivalent probability measure and classical risk neutral pricing would not be well-founded.\\

In order to get an intuition of what is the typical path behavior  when dealing with true and strict local martingales, we simulate some paths of the Radon-Nikodym derivative for the putative risk neutral measure of the
$i$-th currency denomination $\hat{B}^i(t)=\frac{B^i(t)}{D^i(t)}$ according to  the corresponding SDE \eqref{putativeRadonNikodym}, together with the respective quadratic variation processes,  for  time horizon $t=10$ years. In this illustration we consider a one factor specification  of the 4/2 model (i.e. $d=1$) and fix the parameters as follows: $\kappa = 0.49523; \theta = 0.53561;\sigma = 0.67128; V(0) = 1.4338;\rho = 
-0.89728; a = 0.047360.$ We let the parameter $b$ range in the interval $[-0.4,0.4]$ in order to generate situations in which the process $\hat{B}^i$ is a true martingale ($b$ positive) or a strict local martingale ($b$ negative).
We see in Figure \ref{strictlocalmartingale10} that the quadratic variation of the strict local martingale process almost explodes from time to time and increases through these upward jumps visually  much faster than in the case corresponding to the true martingale process, in line with the well-known unbounded expected  quadratic variation process for square integrable strict local martingales, see e.g. \cite{bookplaten10}. \\


To conclude this section, let us observe that we can compute the prices of zero coupon bonds for all currency denominations, meaning that it is $a$ $priori$ possible to devise a model for long-dated FX products, in the spirit of \cite{gg13}, where a joint calibration to FX surfaces and yield curves is performed. Depending on the parameter values, our general framework may be interpreted both from the point of view of real world pricing  and classical risk-neutral valuation, respectively:
\begin{itemize}
\item Should market data imply the existence of a risk-neutral probability measure for the $i$-th currency denomination, then it would be possible to equivalently employ the $i$-th money market account as num\'eraire.
\item In the other case, i.e. when risk neutral pricing is not possible for the $i$-th currency denomination due to the strict local martingale property of the $i$-th benchmarked money-market account, then discounting should be  performed via the  GOP. \end{itemize}

\section{Valuation of Derivatives}\label{sec:valuationOfDerivatives}

In the present section we solve the valuation problem for various contingent claims. The general valuation tool will be given by the real-world pricing formula \eqref{eq:realWorldPricingFormula}. Subsection \ref{sec:fxOptions} concentrates on plain vanilla European FX options, for which a semi-closed form valuation is available by means of Fourier techniques. These require the knowledge of the characteristic function of the log-underlying, given below in Theorem \ref{TheoremTransform}, which provides as  a by-product a closed form valuation formula for benchmarked zero-coupon bonds.

\subsection{FX options}\label{sec:fxOptions}
We first provide the calculation of the discounted conditional Fourier/Laplace transform of $x^{i,j}(t):=\ln (S^{i,j}(t))$, which will be useful for option pricing purposes. Let us consider a European call option $C(S^{i,j}(t),K^{i,j},\tau)$ at time $t$, $ i,j=1,...,N,i\not=j,$ on a generic exchange rate process  $S^{i,j}$ with
strike $K^{i,j}$, maturity $T=t+\tau$   and face value equal to one unit of the foreign currency. We denote via $Y^i(t)=\log (D^i(t))$ the logarithm of the  GOP in $i$-th currency denomination. Hence the log-exchange rate  may be written as $x^{i,j}(t)=\log (S^{i,j}(t))=Y^i(t)-Y^j(t)$. Let us introduce the following conditional expectation
\begin{align}
\begin{split}
    \phi^{i,j}_{t,T}( z) &= D^{i}(t) \mathbb{E}\left[\left.\frac{1}{D^{i}(T)}e^{\mathtt{i} z x^{i,j}(T)}\right|\cF_t  \right]\\
    &=e^{Y^i(t)}\mathbb{E}\left[\left.e^{-Y^i(T)+\mathtt{i} z(Y^i(T)-Y^j(T))}\right|\cF_t  \right]
    \end{split}
\end{align}
for $\mathtt{i}=\sqrt{-1}$. For $ z=u\in\mathbb{R}$ we will use the terminology of a \textit{discounted characteristic function}, whereas for $ z\in\mathbb{C}$ when the expectation exists, the function $ \phi^{i,j}_{t,T}$ will be called a \textit{generalized discounted characteristic function}. If we denote by $\Psi_{t,T}(z)$ the joint conditional (generalized) characteristic function of the vector of GOP denominations $Y(T)=(Y^1(T),...,Y^N(T))$, that is 
\begin{align}
\Psi_{t,T}(\zeta):=\Excond{}{e^{\imag\langle\zeta,Y(T)\rangle}}{\cF_t}, \ \ \zeta\in\mathbb{C}^N,
\label{jointPsi}\end{align}
then we have
\begin{align}
 \phi^{i,j}_{t,T}( z)=D^i(t)\Psi_{t,T}(\zeta),\label{cfrelations}
\end{align}
for $\zeta$ being a vector with $\zeta_i=z+\mathtt{i}$, $\zeta_j=-z$ and all other entries being equal to zero.
Now, from the real-world pricing formula \eqref{eq:realWorldPricingFormula}, the time $t$  price of a call option can be written as the following expected value:
\begin{equation*}
C(S^{i,j}(t),K^{i,j},\tau)=D^i(t)\mathbb{E}\left[\left.\frac{1}{D^i(T)}\left(S^{i,j}(T)-K^{i,j}\right)^+\right|\cF_t\right].
\end{equation*}
Following \cite{lewis2001}, we know that option prices may be interpreted as a convolution of the payoff and the probability density function of the (log)-underlying. As a consequence, the pricing of a derivative may be solved in Fourier space by relying on the Plancherel/Parseval identity, see \cite{lewis2001}, where we have for $f,g\in L^2(\mathbb{R},\mathbb{C})$
\begin{align*}
\int_{-\infty}^{\infty}\overline{f(x)}g(x)dx=\frac{1}{2\pi}\int_{-\infty}^{\infty}\overline{\hat{f}(u)}\hat{g}(u)du
\end{align*}
for $u\in\mathbb{R}$ and $\hat{f},\hat{g}$ denoting the Fourier transforms of $f,g$, respectively. Applying the reasoning above in an option pricing setting requires some additional care. In fact, most payoff functions do not admit a Fourier transform in the classical sense. For example, it is well-known that for the call option one has
\begin{equation*}
\Phi( z)=\int_{\mathbb{R}}e^{\mathtt{i}  z x}\left( e^{x}-K^{i,j}\right)^+\de x=-\frac{\left(K^{i,j}\right)^{\mathtt{i} z+1}}{ z( z-\mathtt{i})},
\end{equation*}
provided we let $ z\in\mathbb{C}$ with $\Im( z)>1$, meaning that $\Phi( z)$ is the Fourier transform of the payoff function in the generalized sense. Such restrictions must be coupled with those that identify the domain where the generalized characteristic function of the log-price is well defined. The reasoning we just reported is developed in Theorem 3.2 in \cite{lewis2001}, where the following general formula is presented (here we write $\phi^{i,j}$ for $\phi^{i,j}_{t,T}$ in order to simplify notation):
\begin{align}
C(S^{i,j}(t),K^{i,j},\tau)& =\frac{1}{ 2\pi }\int_{\mathcal{Z}}\phi^{i,j}(- z)\Phi( z)\de z,  \label{price3}
\end{align}
with $\mathcal{Z}$ denoting the line in the complex plane, parallel to the real axis, where the integration is performed. The article \cite{article_Carr99} followed a different procedure by introducing the concept of a dampened option price. However, as \cite{lewis2001} and \cite{lee2004} point out, this alternative approach is just a particular case of the first one.  In \cite{lee2004},  the Fourier representation of option prices is extended to the case where interest rates are stochastic. Moreover, the shifting of contours, pioneered by \cite{lewis2001}, is employed to prove Theorem 5.1 in \cite{lee2004}. There the following general option pricing formula is presented:
\begin{align}
\begin{split}
C(S^{i,j}(t),&K^{i,j},\tau)\\
& =R\left(S^{i,j}(t),K^{i,j},\alpha\right)+\frac{1}{\pi }\int_{0-\mathtt{i}\alpha}^{\infty-\mathtt{i}\alpha}\Re\left(e^{-\mathtt{i} z k^{i,j}}\frac{\phi^{i,j}( z-\mathtt{i})}{- z( z-\mathtt{i})}\right)d z.  \label{price2}
\end{split}
\end{align}%
Here $k^{i,j}=\log (K^{i,j})$, $\alpha$ denotes the contour of integration and the term coming from the application of the residue theorem is given by
\begin{equation}
R\left(S^{i,j}(t),K^{i,j},\alpha\right) = 
\begin{cases}
 \phi^{i,j}(-\mathtt{i})-K^{i,j}\phi^{i,j}(0), & \mbox{if } \alpha<-1 \\ 
 \phi^{i,j}(-\mathtt{i})-\frac{K^{i,j}}{2}\phi^{i,j}(0), & \mbox{if } \alpha=-1 \\
   \phi^{i,j}(-\mathtt{i}) & \mbox{if } -1<\alpha<0\\
    \frac{1}{2}\phi^{i,j}(-\mathtt{i})&\mbox{if } \alpha = 0\\
    0&\mbox{if } \alpha >0.
 \end{cases}
\end{equation}
The following theorem provides the explicit computation of the generalized discounted characteristic function.
\begin{theorem} \label{TheoremTransform}The joint conditional generalized characteristic function $\Psi_{t,T}(\cdot)$ in (\ref{jointPsi}) 
is given by
\begin{align}
\label{eq:resultTheoremTransform}
\begin{aligned}
\Psi_{t,T}(\zeta)=&\exp \left\{\sum_{i=1}^N \imag\zeta^i \left(Y^i(t)+h^i (T-t) \right)\right\}\\
&\times\prod_{k=1}^d\exp\left\{\sum_{i=1}^N \Bigg[(T-t)\Bigg(\frac{1-\rho^2_k}{2}\sum_{j=1}^N \imag\zeta^i\imag\zeta^j(a^i_kb^j_k+a^j_kb^i_k)\Bigg.\Bigg.\right.\\
&\quad\quad\quad\quad\quad\Bigg.+\imag\zeta^ia^i_kb^i_k+\imag\zeta^i\frac{\kappa_k\rho_k}{\sigma_k}\left(b^i_k-\theta_ka^i_k\right)\Bigg)\\
&\quad\quad\quad\quad\quad\Bigg.\Bigg.-V_k(t)\imag\zeta^i\frac{\rho_k a^i_k}{\sigma_k}-\imag\zeta^i\frac{\rho_kb^i_k}{\sigma_k}\log(V_k(t))\Bigg]\Bigg\}\\
&\times\left( \frac{\beta_k(t,V_k)}{2}\right)^{m_k+1} V_k(t)^{-\frac{\kappa_k\theta_k}{\sigma_k^2}}
(\lambda_k+K_k(t))^{-\left(\frac{1}{2}+\frac{m_k}{2}-\alpha_k +\frac{\kappa_k\theta_k}{\sigma_k^2}\right)}\\
& \times e^{\frac{1}{\sigma_k^2}\left( \kappa_k^2\theta_k(T-t) - \sqrt{A_k}V_k(t)\coth\left(\frac{\sqrt{A_k}(T-t)}{2}\right)+\kappa_k V_k(t)\right)}
\frac{\Gamma \left(\frac{1}{2}+\frac{m_k}{2}-\alpha_k +\frac{\kappa_k\theta_k}{\sigma_k^2}\right) }{\Gamma (m_k+1)}\\
& \times{}_1F_1 \left( \frac{1}{2}+\frac{m_k}{2}-\alpha_k +\frac{\kappa_k\theta_k}{\sigma_k^2}, m_k+1, \frac{\beta_k^2(t,V_k)}{4(\lambda_k +K_k(t))}\right),
\end{aligned}
\end{align}
where
\begin{align*}
m_k &=\frac{2}{\sigma_k^2}\sqrt{\left(\kappa_k\theta_k-\frac{\sigma_k^2}{2}\right)^2+2\sigma_k^2 \nu_k},\\
A_k&=\kappa_k^2 +2\mu_k\sigma_k^2,\\
\beta_k (t, x)&= \frac{2\sqrt{A_k  x}}{\sigma_k^2 \sinh\left(\frac{\sqrt{A_k}(T-t)}{2}\right)},\\
K_k(t)&= \frac{1}{\sigma_k^2}\left( \sqrt{A_k}\coth\left(\frac{\sqrt{A_k}(T-t)}{2}\right)+\kappa_k\right),
\end{align*}
and
\begin{align}
\alpha_k&=-\frac{\rho_k}{\sigma_k}\sum_{i=1}^N\imag\zeta^ib^i_k\label{alphak}\\
\lambda_k&=-\frac{\rho_k}{\sigma_k}\sum_{i=1}^N \imag\zeta^ia_k^i\label{lambdak}\\
\mu_k&=-\sum_{i=1}^N\left(\imag\zeta^iH^i_k+\frac{\imag\zeta^i}{2}(a^i_k)^2+\frac{1-\rho^2_k}{2}\sum_{j=1}^N\imag\zeta^i\imag\zeta^ja^i_ka^j_k+\imag\zeta^i\rho_ka^i_k\frac{\kappa_k}{\sigma_k}\right)\label{muk}\\
\begin{split}
\nu_k&=-\sum_{i=1}^N\left(\imag\zeta^iG^i_k+\frac{\imag\zeta^i}{2}(b^i_k)^2+\frac{1-\rho^2_k}{2}\sum_{j=1}^N\imag\zeta^i\imag\zeta^jb^i_kb^j_k\right.\\
&\quad\quad\quad\quad\left.-\frac{\imag\zeta^i\rho_kb^i_k}{\sigma_k}\left(\kappa_k\theta_k-\frac{\sigma^2_k}{2}\right)\right)\label{nuk}.
\end{split}
\end{align}
Let be given  the functions
\begin{align*}
f^1_k(-\Im(\boldsymbol{\zeta}))&:=\kappa_k^2+2\sigma_k^2\Bigg(-\sum_{i=1}^N\left[-\Im(\zeta^i)H^i_k
\right.\Bigg.\\
&\Bigg.\left.-\frac{\Im(\zeta^i)}{2}(a^i_k)^2+\frac{1-\rho^2_k}{2}\sum_{j=1}^N\Im(\zeta^i)\Im(\zeta^j)a^i_ka^j_k-\Im(\zeta^i)\rho_ka^i_k\frac{\kappa_k}{\sigma_k}\right]\Bigg)\\
f^2_k(-\Im(\boldsymbol{\zeta}))&:=\left(\kappa_k\theta_k-\frac{\sigma_k^2}{2}\right)^2+2\sigma_k^2 \Bigg(-\sum_{i=1}^N\left[-\Im(\zeta^i)G^i_k-\frac{\Im(\zeta^i)}{2}(b^i_k)^2\right.\Bigg.\\
&\quad\quad\quad\Bigg.\left.+\frac{1-\rho^2_k}{2}\sum_{j=1}^N\Im(\zeta^i)\Im(\zeta^j)b^i_kb^j_k+\frac{\Im(\zeta^i)\rho_kb^i_k}{\sigma_k}\left(\kappa_k\theta_k-\frac{\sigma^2_k}{2}\right)\right]\Bigg)\\
f^3_k(-\Im(\boldsymbol{\zeta}))&:=\frac{\kappa_k\theta_k+\frac{\sigma_k^2}{2}+\sqrt{f^2_k(-\Im(\boldsymbol{\zeta}))}}{\sigma_k^2}\\
f^4_k(-\Im(\boldsymbol{\zeta}))&:=\frac{\rho_k}{\sigma_k}\sum_{i=1}^N \Im(\zeta^i)a^i_k+\frac{\sqrt{f^2_k(-\Im(\boldsymbol{\zeta}))}+\kappa_k}{\sigma_k^2},
\end{align*}
in conjunction with the following conditions
\begin{enumerate}
\item[(i)]\label{theorem:req1} $f^1_k(-\Im(\boldsymbol{\zeta}))>0, \ \forall k=1,\ldots, d$; \\
\item[(ii)]\label{theorem:req2} $f^2_k(-\Im(\boldsymbol{\zeta}))\geq 0, \ \forall k=1,\ldots, d$;\\
\item[(iii)]\label{theorem:req3} $f^3_k(-\Im(\boldsymbol{\zeta}))>0, \ \forall k=1,\ldots, d$; \\
\item[(iv)]\label{theorem:req4} $f^4_k(-\Im(\boldsymbol{\zeta}))\geq 0, \ \forall k=1,\ldots, d$. \\
\end{enumerate}
The transform formula \eqref{eq:resultTheoremTransform} is well defined for all $t\in [0,T]$ when the complex vector $\mathtt{i}\boldsymbol{\zeta}$ belongs to the strip $\mathcal{D}_{t,+\infty}=\mathcal{A}_{t,+\infty}\times\mathtt{i}\RR^N\subset\CC^N$, where the convergence set $\mathcal{A}_{t,+\infty}\subset\RR^N$ is given by
\begin{align*}
\mathcal{A}_{t,+\infty}:=\left\{\left.-\Im(\boldsymbol{\zeta})\in\RR^N\right| \ f^l_k(-\Im(\boldsymbol{\zeta})), \ l=1,\ldots,4\ \text{satisfying  (i)-(iv)}\right\}
\end{align*}
Moreover, for $\imag\boldsymbol{\zeta}\in \mathcal{D}_{t,t^\star}=\mathcal{A}_{t,t^\star}\times\mathtt{i}\RR^N$ with
\begin{align*}
\mathcal{A}_{t,t^\star}&:=\left\{\left.-\Im(\boldsymbol{\zeta})\in\RR^N\right| \ f^l_k(-\Im(\boldsymbol{\zeta})), \ l=1,\ldots,3\ \text{satisfying (i)-(iii) and }\right.\\
&\quad\left.f^4_k(-\Im(\boldsymbol{\zeta}))<0\ \text{for some }k\right\}\supset\mathcal{A}_{t,+\infty}
\end{align*}
the transform is well defined until the maximal time $t^\star$ given by
\begin{align}
\label{eq:explosionTime}
t^\star=\min_{k\text{ s.t. }f^4_k(-\Im(\boldsymbol{\zeta}))<0}\frac{1}{\sqrt{A_k}}\log\left(1-\frac{2\sqrt{A_k}}{\kappa_k+\sigma_k\rho_k\sum_{i=1}^N \Im(\zeta^i)a_k^i+\sqrt{A_k}}\right).
\end{align}
%
\end{theorem}
\begin{proof}
See Appendix \ref{AppendixTheoremTransform}.
\end{proof}

The general transform formula above is a powerful tool, however, checking the validity of \eqref{eq:resultTheoremTransform}, may not be very practical in a calibration setting. For this reason, we provide a simple, yet handy, criterion. Recall that we introduced in \eqref{zerocouponbond} the price $P^i(t,T)$ at time $t\in [0, T], \ 0 \leq T\leq \bar T$  of a zero coupon bond for one unit of the $i$-th currency to be paid at $T$, $i =1,\ldots,N$, via the following conditional expectation
\begin{align}
P^i(t,T):=D^i(t)\mathbb{E}\left[\left.\frac{1}{D^i(T)}\right|\cF_t\right]=\phi_{t,T}^{i,j}(0).\label{zcb}
\end{align}
The criterion is provided by the next lemma.
\begin{lemma}\label{criterioncf}
Let $-1<\alpha<0$ and $z\in \mathbb{C}$ with $z=u+\mathtt{i}\alpha$. Assume
\begin{align*}
P^i(t,T)\vee P^j(t,T)<\infty,
\end{align*}
then
\begin{align*}
D^i(t)\mathbb{E}\left[\left.\frac{1}{D^i(T)}\left(S^{i,j}(T)\right)^{-\alpha}\right|\cF_t\right]<\infty,
\end{align*}
moreover, the discounted characteristic function $\phi^{i,j}(z)$ admits an analytic extension to the strip
\begin{align*}
\mathcal{Z}=\left\{\left.z\in \mathbb{C}\right| z = u+\mathtt{i}\alpha, \ \alpha\in(-1,0)\right\}.
\end{align*}
\end{lemma}

\begin{proof}
See Appendix \ref{Appendixcriterioncf}.
\end{proof}

Given the result in Lemma \ref{criterioncf}, we shall proceed to calibrate the model by employing the generalized Carr-Madan formula of Lee, i.e. \eqref{price2}, by setting 
\[
R\left(S^{i,j}(t),K^{i,j},\alpha\right) = \phi^{i,j}(-\mathtt{i}).
\]

\section{Model Calibration to FX Triangles}\label{sec:calibrationResults}

In line with \cite{gnoatto11}, \cite{gg13} and \cite{martinoplaten13}, we perform a joint calibration to a triangle of FX implied volatility surfaces. More specifically, we consider the data set employed in \cite{gg13}, featuring implied volatility surfaces for EURUSD, USDJPY and EURJPY as of July $22^{nd}$ 2010. We choose such a date so as to obtain a calibration that can be approximately compared with the one of \cite{gnoatto11}, based on data as of July $23^{rd}$ 2010. We perform our calibration to options with expiry dates ranging from one up to 18 months and moneyness ranging from 15 delta put up to 15 delta call, and we consider a total of 126 contracts. The model we consider for the calibration is the full 4/2 stochastic volatility model, i.e. both the Heston and the $3/2$ effects are simultaneously considered. As we calibrate options with maturity up to 18 months, we do not consider stochastic interest rates due to the limited interest rate risk, see \cite{gg13}. In line with the references above, we choose the following penalty function
\[
\sum_i\left(\sigma^{imp}_{i,mkt}-\sigma^{imp}_{i,model}\right)^2,
\]
where $\sigma^{imp}_{i,mkt}$ is the $i$-th observed market volatility and $\sigma^{imp}_{i,model}$ is the $i$-th model-derived implied volatility. For each option contract, $\sigma^{imp}_{i,model}$ is constructed along the following steps: first, given a set of model parameters, \eqref{price2} for $-1<\alpha<0$ is employed so as to obtain the corresponding model derived price, secondly, the obtained price is converted into $\sigma^{imp}_{i,model}$, via a standard implied volatility solver. As far as the implementation of \eqref{price2} is concerned, we approximated the integral via a $4096$-point FFT routine, with grid spacing equal to $0.1$, so that the improper integral is truncated at the point  $e^{409}$ . The corresponding strike range is then given by $\left[e^{-31.4159},e^{31.4159}\right]$ and Simpson's rule weights are introduced for increased accuracy, see \cite{article_Carr99}. The FFT returns then a vector of option prices for a fixed grid of strikes. Option prices for the strikes of interest are obtained via a linear interpolation. We assume that the model is driven by two square root factors. The parameters we need to calibrate are given by those appearing in the dynamics of each square root process, i.e. $\kappa_k,\theta_k,\sigma_k,V_k(0), \ k=1,2$,  coupled with a two-dimensional vector of correlations and six two-dimensional vectors of projections for each currency area, i.e. $a^i,b^i, \ i=1,2,3$, meaning that we proceed to estimate a total of 22 parameters. Clearly, in order to prevent instability and over-parametrization issues, simplified versions of the model may be considered.

The result of the calibration is presented in Figure \ref{eurusd} for July $22^{nd}$, 2010, while the corresponding parameters are reported in Table \ref{parameters}. We obtain a good fit over all three surfaces we consider, in line with \cite{gnoatto11}.
This shows that a satisfactory calibration of the model can be achieved. It allows us to perform the following analysis, which constitutes an interesting empirical result of the current paper. Given the set of parameters we obtain from the calibration, we can try to analyze whether market data of FX options are supporting the common use of classical risk neutral pricing. Our approach is so flexible, that we can, in the setting of a single model, span both the risk-neutral valuation and the pricing under the real-world measure. Such an analysis is summarized in Table \ref{fellertest}. We consider different measures for pricing: the real world probability measure $\mathbb{P}$, and the putative	 risk neutral measures $\mathbb{Q}^{usd}$, $\mathbb{Q}^{eur}$ and $\mathbb{Q}^{jpy}$. For each measure we compute the corresponding Feller condition for each square root process. \\

Under the real world probability measure $\mathbb{P}$, we observe that the Feller condition $2\kappa_k\theta_k\geq \sigma^2_k$ is satisfied by both $V_1$ and $V_2$. We next proceed to perform the same analysis under the two putative risk neutral measures $\mathbb{Q}^{usd}$ and $\mathbb{Q}^{jpy}$, respectively. We observe that for the first one we still have that both processes do never reach zero, whereas for $\mathbb{Q}^{jpy}$ we have that the Feller condition is not satisfied by the second component. As discussed in Section \ref{subsecstrictlm}, if at least one of the square root processes has a different behavior under the putative risk-neutral measure, then we have that classical risk-neutral pricing is not well founded. In summary, we have a situation where market data suggest that for the USD currency denomination risk-neutral pricing is potentially applicable, while in the JPY denomination it is not theoretically founded. \\

We also perform a second calibration experiment. The structure of the sample of the dataset is the same as in the previous case and market data were provided as of Feb 23${}^{rd}$, March 23${}^{rd}$, April 22${}^{nd}$, May 22${}^{nd}$ and June 22${}^{nd}$, 2015. Essentially, we are taking the perspective of a derivative desk following the market practice that involves a periodic model re-calibration across different trading dates. Such analysis allows us to provide some first evidence regarding the stability of the parameter estimates we obtain. By looking at Table \ref{parameters2015}, we observe a satisfactory stability of the calibrated parameters. A relevant change in the estimates is observed only between the February and March calibration. The quality of the fit is comparable with the one obtained in our first calibration and the above mentioned papers of \cite{martinoplaten13} and \cite{gnoatto11}. 

 Calibrated parameter values are listed in Table~\ref{parameters2015}, whereas the Feller condition under all measures is reported in Table~\ref{fellertest2015}. We observe in this case a violation of the Feller condition for the second factor under the $\mathbb{Q}^{usd}$ putative risk neutral measure, whereas  for the $\mathbb{Q}^{jpy}$ measure the condition is passed.  For $\mathbb{Q}^{eur}$, instead, we observe that the condition is initially passed and then, starting from April 22${}^{nd}$, we have repeated violations.
The overall results of our  analysis allow us to suggest that markets are subject to what we may term as \textit{regime switches} in pricing between the classical risk neutral and the more general real world pricing  approach. Such a feature would clearly provide a strong motivation for the introduction of  models that are able to accomodate both valuation frameworks, like the 4/2 type specification that we propose.

\section{Pricing and Hedging of Long-Dated Zero Coupon Bonds}\label{sec:hedgingLongDatedSecurities}

In the previous section we obtained a prototypical calibration to market data that shows the coexistence of pricing under the risk-neutral and the more general benchmark approach. In the present section, we take the calibrated values as given and consider the problem of hedging contingent claims under the 4/2 model. We start in Subsection \ref{sec:hedgingTheory} by providing the necessary background on quadratic hedging and, in particular, on \textit{benchmarked risk minimization}. We restrict our attention to a very simple contingent claim, namely a zero coupon bond, which is a core building block of annuities and many other insurance products. Such an experiment is simple and yet extremely powerful in showing how the benchmark approach allows for the hedging of contingent claims for a lower cost. The construction of the hedging scheme requires a martingale representation for the claim under consideration, where the initial price represents the starting point of the value of the strategy. In Subsection \ref{sec:hedgingPart1} we analyze the valuation formula for a zero coupon bond and explicitly highlight the consequences of the failure of the martingale property from an analytical point of view. Finally, in Subsection \ref{sec:hedgingPart2} we explicitly compute the dynamic hedging scheme.

\subsection{Benchmarked Risk Minimization}\label{sec:hedgingTheory}
The 4/2 model is  a  stochastic volatility model. Hence, due to the presence of the instantaneous volatility uncertainty, we have an example of an incomplete market. Hedging claims in an incomplete market setting is a non-trivial task, which may be performed by means of different possible criteria. Incompleteness means that, in general, it is not possible to construct a self-financing trading strategy that delivers at maturity $T$ the final payoff $\cH(T)$ almost surely. According to the survey paper of \cite{schweizer01}, in an incomplete market setting one may relax the requirements on the family of possible trading strategies in either two directions:
\begin{enumerate}
\item A first possibility is to relax the requirement that the terminal value process of the strategy reaches the final payoff $\cH(T)$ a.s., and hence one is induced to minimize the expected quadratic hedging error at maturity over a suitable set of self-financing trading strategies. This first approach is called \textit{mean-variance hedging}, and is presented in \cite{bola89}, \cite{duri91} and \cite{schweizer94a}.
\item Alternatively, one may relax the self-financing requirement and insist on attaining $\cH(T)$ a.s. at maturity $T$. In this second case one is induced to minimize a quadratic function of the cost process of the (non-self-financing) trading strategy. This second approach, referred to as local risk minimization, was introduced in \cite{foso86} and then generalized in \cite{schweizer91} in a general semimartingale setting assuming the existence of a minimal equivalent martingale measure; see also \cite{mol01}.
\end{enumerate} 
In this paper we adopt some type of local risk minimization. It is important to stress that we need to consider a generalized concept of local risk minimization in the context of the benchmark approach. A first generalization in this sense was provided by \cite{bcp14}. In the current paper, we adopt the more general approach of \cite{duplaten14}, known as \textit{benchmarked risk minimization}, which does not require second moment conditions. 

In Section \ref{sec:model} we introduced the notion of a self-financing trading strategy. We need to generalize our  notation, since we will consider, in general, also non-self-financing trading strategies. To this end, in line with Definition 2.1 of \cite{duplaten14}, we will call a dynamic trading strategy, initiated at $t=0$, an $\RR^{N+1}$-valued process of the form $\boldsymbol{v}=\left\{v(t)=\left(\eta(t),\vartheta^{1}(t),\ldots,\vartheta^{N}(t)\right)^\top, \ 0\leq t\leq T\leq \bar{T}\right\}$ for a predictable $\widehat{\boldsymbol{B}}$-integrable process $\boldsymbol{\vartheta}=\left\{\boldsymbol{\vartheta}(t)=(\vartheta^{1}(t),\ldots,\vartheta^{N}(t)), \ 0\leq t\leq T\leq \bar{T}\right\}$, which describes the units invested in the benchmarked primary security accounts $\widehat{\boldsymbol{B}}$ and forms the self-financing part of the associated portfolio. The corresponding benchmarked value process is then $\widehat{\tV}^{\boldsymbol{v}}=
\boldsymbol{\vartheta}(t)^{\top}\widehat{\boldsymbol{B}}(t)+\eta(t)$, where $\eta =\left\{\eta(t), \ 0\leq t\leq T\leq \bar{T}\right\}$ with $\eta(0)=0$ monitors the non-self-financing part, so that we can write
\begin{align*}
\widehat{\tV}^{\boldsymbol{v}}(t)=\widehat{\tV}^{\boldsymbol{v}}(0)+\int_0^t\boldsymbol{\vartheta}(s)^{\top}\de\widehat{\boldsymbol{B}}(s)+\eta(t).
\end{align*}
The process $\eta$ monitors the inflow/outflow of extra capital and hence measures the cost of the strategy, see Corollary 4.2 in \cite{duplaten14}. When the monitoring process is a local martingale, we say that the strategy $\boldsymbol{v}$ is \textit{mean-self-financing}, see Definition 4.4 in \cite{duplaten14}. Moreover, when the monitoring part $\eta$ is orthogonal to the primary security accounts, in the sense that $\eta(t)\widehat{\boldsymbol{B}}(t)$ forms a vector local martingale, we say that the trading strategy $\boldsymbol{v}$ has an orthogonal benchmarked profit and loss, see Definition 5.1 in \cite{duplaten14}. Given $\cV_{\widehat{\cH}(T)}$, the set of all mean-self-financing trading strategies which deliver the final benchmarked payoff $\widehat{\cH}(T)$ $\PP$-a.s., with an orthogonal benchmarked profit and loss, see Definition 5.2 in \cite{duplaten14}, we say that a strategy $\widetilde{\boldsymbol{v}}\in \cV_{\widehat{\cH}(T)}$ is \textit{benchmarked risk minimizing} if, for all strategies ${\boldsymbol{v}}\in \cV_{\widehat{\cH}(T)}$, we have that $\widehat{\tV}^{\widetilde{\boldsymbol{v}}}(t)\leq \widehat{\tV}^{{\boldsymbol{v}}}(t)$ $\PP$-a.s. for every $0\leq t\leq T\leq \bar{T}$, see Definition 5.3 in \cite{duplaten14}.

The practical application of the concept of benchmarked risk minimization necessitates the availability of martingale representations for benchmarked contingent claims, which will be given under the real-world probability measure $\PP$. In the tractable Markovian setting of the 4/2 model, such representation can be explicitly obtained, so that one  represents a benchmarked contingent claim $\widehat{\cH}(T)$ in the  form
\begin{align}
\label{eq:martingaleRepresentationGeneral}
\widehat{\cH}(T)=\Excond{}{\widehat{\cH}(T)}{\cF_t}+\int_t^T\boldsymbol{\vartheta}(s)^{\top}\de\widehat{\boldsymbol{B}}(s)+\eta(T)-\eta(t),
\end{align}
see Equation (6.1) in \cite{duplaten14}, and there  exists a benchmarked risk minimizing strategy $\boldsymbol{v}$ for $\widehat{\cH}(T)$. In summary, to  determine the strategy, one  has first to compute the conditional expectation in \eqref{eq:martingaleRepresentationGeneral}, which will be straightforward  in our case due to the analytical tractability of the 4/2 model. In a second step, an application of the It\^o formula  allows for the identification of the possible holdings $\boldsymbol{\vartheta}$ in the self-financing part of the price process $\widehat{\tV}^{{\boldsymbol{v}}}$. Finally, the monitoring part is given by $\eta(t)=\widehat{\tV}^{{\boldsymbol{v}}}(t)-\boldsymbol{\vartheta}(t)^{\top}\widehat{\boldsymbol{B}}(t)$, which when multiplied with $\widehat{\boldsymbol{B}}(t)$ needs to have zero drift. In the next subsections, we  carry out this procedure for a zero-coupon bond in the Japanese currency denomination, where, based on  the calibration results discussed in Section \ref{sec:calibrationResults}, we observed a failure of the risk neutral paradigm. Notice that the calibration of Section \ref{sec:calibrationResults} refers to a model where the short rate is assumed to be constant. We do recognize that this represents a simplification. Note, instead one could also use as payoff one unit of the saving accounts to demonstrate the consequences of the failure of the risk neutral approach. However, taking the short rate constant allows us to illustrate easily the impact of the violation of the risk neutral paradigm on zero coupon bonds. Introducing a stochastic short rate can be easily achieved by allowing for non-zero projection vectors $\boldsymbol{H}^i,\boldsymbol{G}^i$ in \eqref{eq:shortRateDynamics}. For an example of a hedging scheme in the presence of stochastic interest rates we refer to \cite{bfip14}.

\subsection{Pricing of a Long-Dated Zero Coupon Bond}\label{sec:hedgingPart1}

In this subsection we first quantify the impact of the violation of the classical risk neutral assumption on the pricing of a zero coupon bond written on JPY as domestic currency when assuming a constant interest rate $r^{JPY}$. If risk neutral valuation were possible in this market, then pricing as well as hedging of this elementary security would be trivial: It would consist in keeping the JPY amount $\exp(-r^{JPY} (T-t))$ in the domestic bank account $B^{JPY}$ and nothing in any other asset.

However, for the calibrated market we showed that under the 4/2 model the risk neutral probability measure most likely does not exist for the JPY currency denomination, so that risk neutral pricing is not allowed. 
As a consequence, we price the zero coupon bond under the real world probability measure using the real world pricing formula \eqref{zerocouponbond}. Recall that in our stochastic volatility context the market is incomplete, so that perfect hedging is not possible. As already indicated, we adopt  benchmarked risk minimization, developed in \cite{duplaten14}, and  find the corresponding mean-self-financing  hedging strategy that delivers  the payoff.

The first step consists in finding the real-world price of the zero coupon bond by using the real-world pricing formula. Assume  that the domestic currency is indexed by $i=1$ and set $d=2$. Then we focus on  the zero coupon bond value in domestic currency paying  one unit of domestic currency:
\begin{align*}
P^1(t,T)&=D^1(t)\mathbb{E}\left[\left.\frac{1}{D^1(T)}\right|\cF_t\right]\\
&=D^1(t)\Psi_{t,T}(\imag,0),
\end{align*}
where $\Psi_{t,T}(\imag,0)=\Psi_{t,T}(\imag,0,...,0)$, see \eqref{jointPsi}. Now, from Theorem \ref{TheoremTransform} we have for $k=1,...,d$ that
\begin{align*}
\alpha_k&=-\frac{\rho_k}{\sigma_k}b^1_k;\\
\lambda_k&=-\frac{\rho_k}{\sigma_k}a_k^1;\\
\mu_k&=\rho_k a^1_k \left( \frac{\kappa_k}{\sigma_k}+\frac{\rho_ka^1_k}{2}\right);\\
\nu_k&=\rho_k b^1_k \left( -\frac{\kappa_k\theta_k}{\sigma_k}+\frac{\rho_kb^1_k}{2}+\frac{\rho_k\sigma_k b^1_k}{2}\right);\\
A_k&=(\kappa_k+\sigma_k\rho_ka^1_k)^2
\end{align*}
and the argument of the hypergeometric function in \eqref{eq:resultTheoremTransform} becomes
\begin{align*}
\frac{\beta_k^2(t,V_k)}{4(\lambda_k +K_k(t))}&=
V_k \frac{2}{\sigma^2_k}(\kappa_k+\sigma_k\rho_k a^1_k)\left(e^{\sqrt{A_k}(T-t)}-1\right)^{-1}.
\end{align*}
The parameter $m_k$ becomes 
\begin{align*}
m_k= \frac{2}{\sigma^2_k}\vert \kappa_k\theta_k- \frac{\sigma_k^2}{2}-\rho_k\sigma_k b^1_k\vert .
\end{align*}

Since we consider the case  where the JPY is the domestic currency, we have according to Table~\ref{parameters}
\begin{align*}
m_1&= \frac{2}{\sigma^2_1}\left( \kappa_1\theta_1- \frac{\sigma_1^2}{2}-\rho_1\sigma_1 b^1_1\right)>0;\\
m_2&= -\frac{2}{\sigma^2_2}\left( \kappa_2\theta_2- \frac{\sigma_2^2}{2}-\rho_2\sigma_2 b^1_2\right)>0.
\end{align*}
Notice that for $k=1$ we have 
\begin{align*}
\frac{1}{2}+\frac{m_1}{2}-\alpha_1+\frac{\kappa_1\theta_1}{\sigma^2_1}&=\frac{2}{\sigma^2_1}\left(\kappa_1\theta_1-\rho_1\sigma_1 b_1^1\right)\\
&= m_1+1,
\end{align*}
which implies that the confluent hypergometric function becomes a simple exponential and the factor involving the term $k=1$ in the expression of $\Psi$ simplifies to 1, so that no contribution to the formula comes from the term $k=1$. This is not surprising, as it is related to a volatility factor for which the risk neutral measure exists, related to the martingale part that simplifies in the conditional expectation.

The situation is different for $k=2$, where it turns out that 
\begin{align*}
\frac{1}{2}+\frac{m_2}{2}-\alpha_2+\frac{\kappa_2\theta_2}{\sigma^2_2}&=1.
\end{align*}
Now the expression in  \eqref{eq:resultTheoremTransform} becomes
\begin{align*}
\Psi_{t,T}(\imag,0)=&\ e^{-Y^1(t)}
\frac{\exp \left( -\rho_2^2 a_2^1 b_2^1\tau -\frac{\rho_2 \kappa_2}{\sigma_2}(b_2^1-a_2^1\theta_2)\tau+\frac{\kappa_2^2\theta_2}{\sigma_2^2}\tau\right)
\left( \frac{\sqrt{A_2}}{\sigma_2^2}\right)^{m_2}}
{\left(\coth\left(\frac{\sqrt{A_2}\tau}{2}\right)+1\right)\left(\sinh\left(\frac{\sqrt{A_2}\tau}{2}\right)\right)^{1+m_2}\Gamma (1+m_2)}\\
&.V_2^{m_2}\exp\left( V_2 \frac{\sqrt{A_2}\tau}{2}\left(1-\coth\left(\frac{\sqrt{A_2}\tau}{2}\right)\right)\right)\\
& ._1F_1\left(1,1+m_2,V_2 \frac{2\sqrt{A_2}}{\sigma^2_2}\left(e^{\sqrt{A_2}\tau}-1\right)^{-1}\right)\\
=&\ \frac{\exp \left( -\rho_2^2 a_2^1 b_2^1\tau -\frac{\rho_2 \kappa_2}{\sigma_2}(b_2^1-a_2^1\theta_2)\tau+\frac{\kappa_2^2\theta_2}{\sigma_2^2}\tau\right)
\left( \frac{\sqrt{A_2}}{\sigma_2^2}\right)^{m_2}}
{\left(\coth\left(\frac{\sqrt{A_2}\tau}{2}\right)+1\right)\left(\sinh\left(\frac{\sqrt{A_2}\tau}{2}\right)\right)^{1+m_2}\Gamma (1+m_2)}\\
&.V_2^{m_2}\exp\left( -V_2 \frac{2\sqrt{A_2}}{\sigma^2_2}\left(e^{\sqrt{A_2}\tau}-1\right)^{-1}\right)\\
& ._1F_1\left(1,1+m_2,V_2 \frac{2\sqrt{A_2}}{\sigma^2_2}\left(e^{\sqrt{A_2}\tau}-1\right)^{-1}\right).
\end{align*}
Using the Kummer transformation $e^{-z}_1F_1(a,b,z)=_1F_1(b-a,b,-z)$, see Equation (13.1.27) in \cite{book_AbramStegun}, we get
\begin{align*}
\Psi_{t,T}(\imag,0)=&\ e^{-Y^1(t)}
\frac{\exp \left( -\rho_2^2 a_2^1 b_2^1\tau -\frac{\rho_2 \kappa_2}{\sigma_2}(b_2^1-a_2^1\theta_2)\tau+\frac{\kappa_2^2\theta_2}{\sigma_2^2}\tau\right)
\left( \frac{\sqrt{A_2}}{\sigma_2^2}\right)^{m_2}}
{\left(\coth\left(\frac{\sqrt{A_2}\tau}{2}\right)+1\right)\left(\sinh\left(\frac{\sqrt{A_2}\tau}{2}\right)\right)^{1+m_2}\Gamma (1+m_2)}\\
&.V_2^{m_{2}} {}_1F_1\left(m,1+m_2,-V_2 \frac{2\sqrt{A_2}}{\sigma^2_2}\left(e^{\sqrt{A_2}\tau}-1\right)^{-1}\right).
\end{align*}
We now use the relation ${}_1F_1(a,a+1,-z)=az^{-a}\gamma (a,z)$, see Equation (13.6.10) in \cite{book_AbramStegun}, where $\gamma(a,z)$ denotes the (lower) incomplete Gamma function defined as
\begin{align*}
\gamma (a,z):=\int_0^z t^{a-1} e^{-t}dt.
\end{align*}
After some computations we  finally obtain the price of a zero coupon bond as
\begin{align}
\label{zcb1}
P^1(t,T)&= e^{-r^1\tau} \frac{\gamma\left(m_2, V_2 \frac{2\sqrt{A_2}}{\sigma^2_2}\left(e^{\sqrt{A_2}\tau}-1\right)^{-1}\right)}{\Gamma(m_2)},
\end{align}
where we recall that $\sqrt{A_2}=\kappa_2+\sigma_2\rho_2a^1_2$ and $\tau=T-t$.

We observe that the function $\Psi_{t,T}(\imag,0)$ is decreasing with respect to $\tau=T-t$, so that we can analytically confirm that the failure of the risk neutral assumption has a direct impact on the price of long dated securities such as  zero coupon bonds. In Table~\ref{ZCBJPY} we compute the price of a zero coupon bond in the JPY market for which the risk neutral probability measure does not exist. We emphasize that we obtain prices that are markedly lower than those typically given by the risk neutral approach. As a comparison, in Table~\ref{ZCBUSD} we compare, for the USD currency denomination, the price obtained from the (trivial) risk neutral formula for deterministic short rate, and the corresponding application of the real-world pricing formula for the zero-coupon bond, which follows from Theorem~\ref{TheoremTransform}. Since, according to our calibration for the USD currency denomination, a risk-neutral measure $\QQ^{usd}$ exists, (see the corresponding Feller test in Table~\ref{fellertest}) we observe that the  risk neutral and the real-world pricing  formula coincide, thus, providing us with a concrete example of the correspondence we highlighted in \eqref{eq:realWorldEqualsRiskNeutral}.\\

Notice in \eqref{zcb1} that $P^1(t,T)< e^{-r^1\tau}$,  meaning that under the benchmark approach, hedging the unit amount of JPY turns out to be less expensive in comparison to the trivial risk neutral strategy consisting in putting the amount $e^{-r^1\tau}$ in the domestic  money market account $B^1$ as suggested under the risk neutral paradigm. On the other hand,  benchmarked risk minimization  requires a more sophisticated hedging strategy. In Subsection  \ref{sec:hedgingPart2} below, we describe such hedging 
 involving the GOP and two money market accounts. We are not introducing a sort of inconsistency in the pricing procedure since the benchmark approach allows  to associate formally the amount $$e^{-r^1\tau}$$ to the price of a zero coupon bond in our constant interest rate setting. Note that this is not the minimal possible price,   as there exists no equivalent  risk neutral probability measure. On the other hand, the real-world pricing formula provides under the given model the minimal possible price. The existence of those two self-financing price processes does not present an economically meaningful arbitrage opportunity because the best performing portfolio in the long run of this market, the GOP, does not explode in finite time. As 
 pointed out by \cite{Loewenstein00}, the classical no-arbitrage concept that is equivalent to the existence of an equivalent  risk neutral probability measure is too restrictive.\\

One may wonder if this from the classical theory perspective observed "\textit{pricing puzzle}" represents a pathological or marginal situation. To give an answer to this question, we repeat the pricing experiment, by employing the calibrated parameters of Table \ref{parameters2015}. 
Tables \ref{ZCBUSD201502} - \ref{ZCBUSD201506} report the values of the corresponding zero coupon bond prices for all the maturities available from a major provider. Note that also in this experiment, starting from maturities longer than ten years, the differences between the classical and the benchmark approach become substantial in so far as benchmarked risk minimization leads to much lower prices. By comparing the violations of the Feller test and the pricing results, we observe a one to one correspondence between the two phenomena.

\subsection{Dynamic Hedging of a Long-Dated Zero Coupon Bond}\label{sec:hedgingPart2}

In this subsection,
 we develop the second part of our analysis that involves the determination of the dynamic part of the benchmarked risk minimizing strategy, namely the holdings in the benchmarked primary security accounts and the non-self-financing/monitoring part. First of all we compute the SDE for the martingale representation associated to the benchmarked price of the zero coupon bond, that we denote by $\hat{P}^1(t,T)=P^1(t,T)/D^1(t)$. We obtain
\begin{align*}
d\hat{P}^1(t,T)=&\ -\Psi_{t,T}(\imag,0)\left(a^1_1\sqrt{V_1(t)}+\frac{b_1^1}{\sqrt{V_1(t)}}\right)\de Z_1(t)\\
& -\Psi_{t,T}(\imag,0)\left(a^1_2\sqrt{V_2(t)}+\frac{b_2^1}{\sqrt{V_2(t)}}\right)\de Z_2(t)\\
&+ \frac{e^{-r^1\tau}}{D_1(t)\Gamma(m_2)} \left( \frac{2\sqrt{A_2}}{\sigma_2^2\left(e^{\sqrt{A_2}\tau}-1\right)}\right)^{m_2}\\
&. e^{V_2 \frac{2\sqrt{A_2}}{\sigma^2_2}\left(e^{\sqrt{A_2}\tau}-1\right)^{-1}}V_2^{m_2-1/2}\sigma_2\left(\rho_2 \de Z_2(t)+\sqrt{1-\rho_2^2}\de Z_2^\perp (t)\right),
\end{align*}
which we rewrite, in line with \cite{duplaten14}, Proposition 7.1, as 
\begin{align}
d\hat{P}^1(t,T)&=x_1(t)dZ_1(t)+x_2(t)dZ_2(t)+x_3(t)dZ_2^\perp (t),
\end{align}
with the implicit introduction of the respective shorthand notations $x_1,x_2,x_3$. Let us recall that the components of the vector of primary benchmarked assets $\hat{\boldsymbol{B}}$ are
\begin{align*}
\hat{B}^1(t)&=\frac{B^1(t)}{D_1(t)},\\
\hat{B}^{2}(t)&=\frac{B^2(t)S^{1,2}(t)}{D_1(t)}=\frac{B^2(t)}{D_2(t)},\\
\hat{B}^{3}(t)&=\frac{B^3(t)S^{1,3}(t)}{D_1(t)}=\frac{B^3(t)}{D_3(t)},
\end{align*}
where, for $i=1,2,3$ we have
\begin{align*}
\de\hat{B}^{i}(t)&=-\hat{B}^{i}(t)\left(a^i_1\sqrt{V_1(t)}+\frac{b_1^i}{\sqrt{V_1(t)}}\right)\de Z_1(t)\\
&\quad-\hat{B}^{i}(t)\left(a^i_2\sqrt{V_2(t)}+\frac{b_2^i}{\sqrt{V_2(t)}}\right)\de Z_2(t).
\end{align*}
We introduce the $3\times 3$ matrix $\Phi^{i,k}_t$ as in Equation (7.2) of \cite{duplaten14} in the form
\begin{equation}
\Phi^{i,k}(t)= 
\begin{cases}
 a^i_k\sqrt{V_k(t)}+\frac{b_k^i}{\sqrt{V_k(t)}}, & k=1,2;\  i=1,2,3 \\ 
 1, & k=3;\  i=1,2,3. \\  \end{cases}
\end{equation}
We now define the vector $\boldsymbol{\xi}(t)=\left( -x_1(t)  , -x_2(t),  \Psi_{t,T}(\imag,0)-x_3(t)\right)^\top$,
so that the investment in the self-financing part is given as follows:
\begin{align}
\label{eq:explictSelfFinancingPart}
\boldsymbol{\vartheta}^1(t)= Diag^{-1}(\hat{\boldsymbol{B}}^1(t))\left(\left(\Phi^{i,k}_t\right)^\top\right)^{-1} \boldsymbol{\xi}(t),
\end{align}
see Equation (7.4) in \cite{duplaten14}. The non-self-financing or monitoring part of the strategy is then given by the local martingale
\begin{align}
\label{eq:explicitMonitoringPart}
\eta^1(t)=\int_0^t x_3(s) \de Z_2^\perp (s),
\end{align}
which is orthogonal to $\hat{\boldsymbol{B}}$. Thus, the conditional expectation \eqref{zcb1}, the vector \eqref{eq:explictSelfFinancingPart} and the local martingale \eqref{eq:explicitMonitoringPart} fully determine the benchmarked risk minimizing strategy for the zero coupon bond in the Japanese currency denomination.

The vector $\boldsymbol{\vartheta}^1(t)$ describes the number of units held in the primary security accounts, and the non-hedgeable part $\eta^1(t)$ in \eqref{eq:explicitMonitoringPart} forms a local martingale orthogonal to the benchmarked primary security accounts. Recall that orthogonality means that the product  $\eta^1(t)\hat{\boldsymbol{B}}(t)$ forms a vector local martingale. In this sense the hedge error $\eta^1(t)$ is minimised. A special case of the 4/2 model is the 3/2 model or minimal market model, where one can see that in this case the hedge is involving only the GOP and the domestic money market account, see \cite{bookplaten10}. As shown in   \cite{martinoplaten13}, for many years to maturity the above strategy invests mainly in the GOP and slides later more and more into investing in the domestic money market account. Thus, the above strategy makes financial planning rigorous under the 4/2 model. Similarly, one could now study real-world pricing and hedging of European put options on the GOP and other derivatives, which may become significantly less expensive than their formally risk neutral priced counterparts.

\section{Conclusion} 

In this paper we introduced a more general modeling world than available under the classical no-arbitrage paradigm in finance and insurance. In the context of a flexible model for exchange  rates calibrated to market data, we showed how to hedge a zero coupon bond with a smaller amount of initial capital than required under the classical risk neutral paradigm. Moreover, the corresponding  hedging strategy is no longer only investing in fixed income as in the risk neutral classical world: on the contrary it suggests to invest first primarily in the GOP, that is risky securities and, when approaching more and more the maturity date, it increases also more and more the fraction invested in fixed income. Based on the benchmark approach, this gives a quantitative validation to the  conventional wisdom of financial planners. Their rule of thumb has been typically described in the economic literature via models of optimal life-cycle consumption expenditures and saving, see e.g. \cite{Gourinchas02}, whereas our approach can be seen as a market implied alternative. It does not postulate a particular form for the utiliy function of a representative agent and only asks for the minimal possible price process that can be reasonably hedged to deliver the bond payoff. We also showed that the risk neutral price of a long dated zero coupon bond can be significantly higher than the minimal possible price obtained using the  benchmark approach under the real world probability measure. 

The main mathematical phenomenon underlying these surprising effects is the potential  strict supermartingale property of benchmarked savings accounts under the real world probability measure, as suggested in several cases by our calibration exercise on the foreign exchange option market. 
The presented results represent only some example for new phenomena that can be captured under the benchmark approach.\\
There is ample room for many new research questions and interesting related studies in insurance. Our stochastic volatility model allows for a non linear market price of volatility risk, a desiderable property in order to explain non-linear effects under the real-world probability measure for the risk factors. Second, we performed our calibrations at some particular trading dates. However,  much more information could be extracted from a statistical estimation on a time series of option prices. The proposed 4/2 model  allows for the possibility of a failure of the classical risk neutral pricing assumption and  deserves more theoretical and numerical investigation. In view of this, our paper aims to  stimulate further econometric studies.

\appendix

\section{Proof of Lemma \ref{lemmacasilambda}}\label{Appendixlemmacasilambda}

Let us first observe from \eqref{eq:basicHestonDynamics} and \eqref{exchangerates} that in the present setting we have $\pi^1(t)=\lambda(t)$ and $\pi^2(t)=\lambda(t)-\sqrt{V(t)}$. 

Given the general specification of the GOP dynamics, we immediately have for Case \ref{lemma1:point1}:
\begin{align*}
\frac{\de D^1(t)}{D^1(t)}&=  r^1(t)\de t+ a\sqrt{V(t)}\left(a\sqrt{V(t)}\de t+\de {Z}(t)\right),\\
\frac{\de D^2(t)}{D^2(t)}&=  r^2(t)\de t+ (a-1)\sqrt{V(t)}\left((a-1)\sqrt{V(t)}\de t+\de {Z}(t)\right),
\end{align*}
which are both Heston-type models. 

For Case \ref{lemma1:point3}., we have
\begin{align*}
\frac{\de D^1(t)}{D^1(t)}&=  r^1(t)\de t+ \frac{b}{\sqrt{V(t)}}\left(\frac{b}{\sqrt{V(t)}}\de t+\de {Z}(t)\right),\\
\frac{\de D^2(t)}{D^2(t)}&=  r^2(t)\de t+ \left(\frac{b}{\sqrt{V(t)}}-\sqrt{V(t)}\right)\left(\left(\frac{b}{\sqrt{V(t)}}-\sqrt{V(t)}\right)\de t+\de {Z}(t)\right),
\end{align*}
so that $D^1$ follows a $3/2$ model (see \cite{heston97}, \cite{platen97}) whereas $D^2$ follows a $4/2$ type model (see \cite{grasselli13}).
On the other hand, for Case \ref{lemma1:point2}., we have
\begin{align*}
\frac{\de D^1(t)}{D^1(t)}&=  r^1(t)\de t+\left(a\sqrt{V(t)}+ \frac{b}{\sqrt{V(t)}}\right)\left(\left(a\sqrt{V(t)}+ \frac{b}{\sqrt{V(t)}}\right)\de t+\de {Z}(t)\right),\\
\frac{\de D^2(t)}{D^2(t)}&=  r^2(t)\de t+\left((a-1)\sqrt{V(t)}+ \frac{b}{\sqrt{V(t)}}\right)\\
&\quad\quad\quad\quad\quad\quad\times\left(\left((a-1)\sqrt{V(t)}+ \frac{b}{\sqrt{V(t)}}\right)\de t+\de {Z}(t)\right),
\end{align*}
which are both $4/2$ type processes.

\section{Proof of Theorem \ref{TheoremTransform}}\label{AppendixTheoremTransform}
We start by directly considering the conditional expectation for $\Psi_t( \boldsymbol{\zeta})$. We perform several manipulations so that the results of \cite{grasselli13} can be directly employed. To this end, we parametrize the correlation structure from Assumption~\ref{assumptionCorrelation} by writing
\begin{align*}
\boldsymbol{Z}(s)=Diag(\boldsymbol{\rho})\boldsymbol{W}(s) +\sqrt{I_d-Diag(\boldsymbol{\rho})^2}\boldsymbol{W}^{\perp}(s),
\end{align*}
so that we have
\begin{align*}
\Psi_t( \boldsymbol{\zeta})&=\exp\left\{\sum_{i=1}^N \imag\zeta^i \left(Y^i(t)+h^i (T-t) \right)\right\}\\
&\quad\times\mathbb{E}\left[\left	.\exp\left\{\sum_{i=1}^N\left(\imag\zeta^i\int_t^T \langle \bsH^i, \boldsymbol{V}(s)\rangle +\langle \bsG^i, \boldsymbol{V}(s)^{-1}\rangle\de s\right.\right.\right.\right.\\
&\quad\left.\left.\left.\left.+\frac{\imag\zeta^i}{2}\int_t^T\boldsymbol{\pi}^i(s)^\top \boldsymbol{\pi}^i(s) \de s+
\imag\zeta^i\int_t^T\boldsymbol{\pi}^i(s)^\top \de \boldsymbol{Z}(s)\right)\right\}\right|\cF_t\right]\\
&=\exp\left\{\sum_{i=1}^N \imag\zeta^i \left(Y^i(t)+h^i (T-t) \right)\right\}\\
&\quad\times\mathbb{E}\left[\left.\exp\left\{\sum_{i=1}^N\left(\imag\zeta^i\int_t^T \langle \bsH^i, \boldsymbol{V}(s)\rangle +\langle \bsG^i, \boldsymbol{V}(s)^{-1}\rangle\de s\right.\right.\right.\right.\\
&\quad\left.\left.+\frac{ \imag\zeta^i}{2}\int_t^T\boldsymbol{\pi}^i(s)^\top \boldsymbol{\pi}^i(s) ds+
\int_t^T   \imag\zeta^i \boldsymbol{\pi}^i(s)^\top Diag(\boldsymbol{\rho})d\boldsymbol{W}(s)\right)\right\}\\
&\quad\left.\left.\times \mathbb{E}\left[\left.\exp\left\{\sum_{i=1}^N \int_t^T
  \imag\zeta^i \boldsymbol{\pi}^i(s)^\top\sqrt{I_d -Diag(\boldsymbol{\rho})^2}\de\boldsymbol{W}^{\perp}(s)\right\}\right|\boldsymbol{V}\right]\right|\cF_t\right]\\
  &=\exp\left\{\sum_{i=1}^N \imag\zeta^i \left(Y^i(t)+h^i (T-t) \right)\right\}\\
&\quad\times\mathbb{E}\left[\left.\exp\left\{\sum_{i=1}^N\left(\imag\zeta^i\int_t^T \langle \bsH^i, \boldsymbol{V}(s)\rangle +\langle \bsG^i, \boldsymbol{V}(s)^{-1}\rangle\de s\right.\right.\right.\right.\\
&\quad+\frac{ \imag\zeta^i}{2}\int_t^T\boldsymbol{\pi}^i(s)^\top \boldsymbol{\pi}^i(s) ds+
\int_t^T   \imag\zeta^i \boldsymbol{\pi}^i(s)^\top Diag(\boldsymbol{\rho})d\boldsymbol{W}(s)\\
&\quad\left.\left.\left.\left.+\frac{1}{2}\int_t^T\left(\sum_{i=1}^N \imag\zeta^i\boldsymbol{\pi}^i(s)^\top\right)\left(I_d -Diag(\boldsymbol{\rho})^2\right)\left(\sum_{i=1}^N \imag\zeta^i\boldsymbol{\pi}^i(s)\right)\de s\right)\right\}\right|\cF_t\right]\\
\displaybreak
  &=\exp\left\{\sum_{i=1}^N \imag\zeta^i \left(Y^i(t)+h^i (T-t) \right)\right\}\\
&\quad\times\prod_{k=1}^d\mathbb{E}\left[\left.\exp\left\{\sum_{i=1}^N\left(\imag\zeta^i\int_t^T H^i_k, V_k(s)+G^i_k, V_k(s)^{-1}\de s\right.\right.\right.\right.\\
&\quad+\frac{ \imag\zeta^i}{2}\int_t^T\boldsymbol{\pi}^i_k(s)^\top \boldsymbol{\pi}^i_k(s) ds+
\int_t^T   \imag\zeta^i \boldsymbol{\pi}^i_k(s)^\top \rho_k \de l{W}_k(s)\\
&\quad\left.\left.\left.\left.+\frac{1-\rho^2_k}{2}\int_t^T\left(\sum_{i=1}^N \imag\zeta^i\boldsymbol{\pi}^i_k(s)\right)^2\de s\right)\right\}\right|\cF_t\right].
\end{align*}
Let us rewrite the final term in the exponent as follows
\begin{align*}
\frac{1-\rho^2_k}{2}&\int_t^T\left(\sum_{i=1}^N \imag\zeta^i\boldsymbol{\pi}^i_k(s)\right)^2\de s\\
&=\frac{1-\rho^2_k}{2}\left(\int_t^T\left(\sum_{i=1}^N \imag\zeta^ia^i_k\sqrt{V_k(s)}\right)^2+\sum_{i=1}^N \sum_{j=1}^N \imag\zeta^i\imag\zeta^j(a^i_kb^j_k+a^j_kb^i_k)\right.\\
&\quad\quad\quad\quad\left.\left(\sum_{i=1}^N \imag\zeta^ib^i_k\frac{1}{\sqrt{V_k(s)}}\right)^2\de s\right)\\
&=\frac{1-\rho^2_k}{2}\left(\int_t^T V_k(s)\de s\left(\sum_{i=1}^N \sum_{j=1}^N \imag\zeta^i\imag\zeta^ja^i_ka^j_k\right)+\sum_{i=1}^N \sum_{j=1}^N \imag\zeta^i\imag\zeta^j(a^i_kb^j_k+a^j_kb^i_k)\right.\\
&\quad\quad\quad\quad\left.+\int_t^T\frac{1}{V_k(s)}\de s\left(\sum_{i=1}^N \sum_{j=1}^N \imag\zeta^i\imag\zeta^jb^i_kb^j_k\right)\right).
\end{align*}
Hence we have
\begin{align*}
\Psi_t( \boldsymbol{\zeta})&=\exp\left\{\sum_{i=1}^N \imag\zeta^i \left(Y^i(t)+h^i (T-t) \right)\right\}\\
&\quad\times\prod_{k=1}^d\exp\left\{\frac{1-\rho^2_k}{2}\sum_{i=1}^N \sum_{j=1}^N \imag\zeta^i\imag\zeta^j(a^i_kb^j_k+a^j_kb^i_k)(T-t)+\sum_{i=1}^N\imag\zeta^ia^i_kb^i_k(T-t)\right\}\\
&\quad\times\mathbb{E}\left[\exp\left\{\sum_{i=1}^N\left(\imag\zeta^i\int_t^T H^i_kV_k(s)+\frac{G^i_k}{V_k(s)}\de s+\frac{\imag\zeta^i}{2}(a^i_k)^2\int_t^TV_k(s)\de s\right.\right.\right.\\
&\quad+\frac{\imag\zeta^i}{2}(b^i_k)^2\int_t^T\frac{1}{V_k(s)}\de s+\imag\zeta^i\rho_ka^i_k\int_t^T\sqrt{V_k(s)}dW_k(s)\\
&+\imag\zeta^i\rho_kb^i_k\int_t^T\frac{1}{\sqrt{V_k(s)}}dW_k(s)+\frac{1-\rho^2_k}{2}\left(\int_t^T V_k(s)\de s\left(\sum_{i=1}^N \sum_{j=1}^N \imag\zeta^i\imag\zeta^ja^i_ka^j_k\right)\right.\\
&\left.\left.\left.\left.\left.+\int_t^T\frac{1}{V_k(s)}\de s\left(\sum_{i=1}^N \sum_{j=1}^N \imag\zeta^i\imag\zeta^jb^i_kb^j_k\right)\right)\right)\right\}\right|\cF_t\right]\\
\displaybreak
&=\exp\left\{\sum_{i=1}^N \imag\zeta^i \left(Y^i(t)+h^i (T-t) \right)\right\}\\
&\quad\times\prod_{k=1}^d\exp\left\{\frac{1-\rho^2_k}{2}\sum_{i=1}^N \sum_{j=1}^N \imag\zeta^i\imag\zeta^j(a^i_kb^j_k+a^j_kb^i_k)(T-t)+\sum_{i=1}^N\imag\zeta^ia^i_kb^i_k(T-t)\right\}\\
&\quad\mathbb{E}\left[\exp\left\{\sum_{i=1}^N\left(\imag\zeta^iH^i_k+\frac{\imag\zeta^i}{2}(a^i_k)^2+\frac{1-\rho^2_k}{2}\sum_{j=1}^N\imag\zeta^i\imag\zeta^ja^i_ka^j_k\right)\int_t^TV_k(s)\de s\right.\right.\\
&+\sum_{i=1}^N\left(\imag\zeta^iG^i_k+\frac{\imag\zeta^i}{2}(b^i_k)^2+\frac{1-\rho^2_k}{2}\sum_{j=1}^N\imag\zeta^i\imag\zeta^jb^i_kb^j_k\right)\int_t^T\frac{1}{V_k(s)}\de s\\
&\left.\left.\left.+\sum_{i=1}^N\left(\imag\zeta^i\rho_ka^i_k\int_t^T\sqrt{V_k(s)}dW_k(s)+\imag\zeta^i\rho_kb^i_k\int_t^T\frac{1}{\sqrt{V_k(s)}}dW_k(s)\right)\right\}\right|\cF_t\right].
\end{align*}
Let us recall the well-known relations
\begin{align*}
\int_t^T \sqrt{V_k(s)}dW_k(s)=&\frac{1}{\sigma_k}(V_k(T)-V_k(t))-\int_t^T\frac{\kappa_k}{\sigma_k} ( \theta_k - V_k(s))ds
\end{align*}
and 
\begin{align*}
\int_t^T \frac{1}{\sqrt{V_k(s)}}dW_k(s)=& \frac{1}{\sigma_k}\ln\frac{V_k(T)}{V_k(t)} -\int_t^T\frac{\frac{\kappa_k\theta_k}{\sigma_k}-\frac{\sigma_k}{2}}{V_k(s)}ds
+\frac{\kappa_k}{\sigma_k}(T-t).
\end{align*}
Substitution of the above relations yields
\begin{align*}
\Psi_t( \boldsymbol{\zeta})&=\exp\left\{\sum_{i=1}^N \imag\zeta^i \left(Y^i(t)+h^i (T-t) \right)\right\}\\
&\quad\times\prod_{k=1}^d\exp\left\{\frac{1-\rho^2_k}{2}\sum_{i=1}^N \sum_{j=1}^N \imag\zeta^i\imag\zeta^j(a^i_kb^j_k+a^j_kb^i_k)(T-t)+\sum_{i=1}^N\imag\zeta^ia^i_kb^i_k(T-t)\right\}\\
&\quad\times\mathbb{E}\left[\exp\left\{\sum_{i=1}^N\left(\imag\zeta^iH^i_k+\frac{\imag\zeta^i}{2}(a^i_k)^2+\frac{1-\rho^2_k}{2}\sum_{j=1}^N\imag\zeta^i\imag\zeta^ja^i_ka^j_k+\imag\zeta^i\rho_ka^i_k\frac{\kappa_k}{\sigma_k}\right)\int_t^TV_k(s)\de s\right.\right.\\
&\quad+\sum_{i=1}^N\left(\imag\zeta^iG^i_k+\frac{\imag\zeta^i}{2}(b^i_k)^2+\frac{1-\rho^2_k}{2}\sum_{j=1}^N\imag\zeta^i\imag\zeta^jb^i_kb^j_k-\frac{\imag\zeta^i\rho_kb^i_k}{\sigma_k}\left(\kappa_k\theta_k-\frac{\sigma^2_k}{2}\right)\right)\int_t^T\frac{1}{V_k(s)}\de s\\
&\quad+\sum_{i=1}^N\left[\imag\zeta^i\rho_ka^i_k\left(\frac{V_k(T)-V_k(t)}{\sigma_k}-\frac{\kappa_k\theta_k}{\sigma_k}(T-t)\right)\right.\\
&\quad+\left.\left.\left.\left.\imag\zeta^i\rho_kb^i_k\left(\frac{1}{\sigma_k}\log\left(\frac{V_k(T)}{v_k(t)}\right)+\frac{\kappa_k}{\sigma_k}(T-t)\right)\right]\right\}\right|\cF_t\right]\\
&=\exp\left\{\sum_{i=1}^N \imag\zeta^i \left(Y^i(t)+h^i (T-t) \right)\right\}\\
&\quad\times\prod_{k=1}^d\exp\left\{\sum_{i=1}^N (T-t)\left(\frac{1-\rho^2_k}{2}\sum_{j=1}^N \imag\zeta^i\imag\zeta^j(a^i_kb^j_k+a^j_kb^i_k)\right.\right.\\
&\quad\quad\quad\quad\quad\quad\left.+\imag\zeta^ia^i_kb^i_k+\imag\zeta^i\frac{\kappa_k\rho_k}{\sigma_k}\left(b^i_k-\theta_ka^i_k\right)\right)\\
&\quad\quad\quad\quad\quad\quad\Bigg.-V_k(t)\imag\zeta^i\frac{\rho_k a^i_k}{\sigma_k}-\imag\zeta^i\frac{\rho_kb^i_k}{\sigma_k}\log(V_k(t))\Bigg\}\\
&\quad\times\mathbb{E}\left[\left.V_k(T)^{-\alpha_k}e^{-\lambda_kV_t(T)-\mu_k\int_t^T V_k(s)\de s-\nu_k\int_t^T\frac{1}{V_k(s)}\de s}\right|\cF_t\right],
\end{align*}
where we introduced \eqref{alphak}-\eqref{nuk}. The result is obtained via an iterated application of Theorem \ref{general} from \cite{grasselli13} with $a=\kappa_k\theta_k, \ b=\kappa_k, \sigma = \frac{\sigma_k^2}{2}$ and $\tau=T-t$
\begin{align}
\begin{aligned}
&\mathbb{E} \left[
 V_k(T)^{-\alpha_k}e^{-\lambda_k V_k(T)- \mu_k \int_t^T V_k(s) ds -\nu_k \int_t^T\frac{ds}{V_k(s)}}\right]\\
&= \left( \frac{\beta_k(\tau,V_k(t))}{2}\right)^{m_k+1} V_k(t)^{-\frac{\kappa_k\theta_k}{\sigma_k^2}}
(\lambda_k+K_k(\tau))^{-\left(\frac{1}{2}+\frac{m_k}{2}-\alpha_k +\frac{\kappa_k\theta_k}{\sigma_k^2}\right)}\\
& \times e^{\frac{1}{\sigma_k^2}\left( \kappa_k^2\theta_k(T-t) - \sqrt{A_k}V_k(t)\coth\left(\frac{\sqrt{A_k}(T-t)}{2}\right)+\kappa_k V_k(t)\right)}
\frac{\Gamma \left(\frac{1}{2}+\frac{m_k}{2}-\alpha_k +\frac{\kappa_k\theta_k}{\sigma_k^2}\right) }{\Gamma (m_k+1)}\\
& \times{}_1F_1 \left( \frac{1}{2}+\frac{m_k}{2}-\alpha_k +\frac{\kappa_k\theta_k}{\sigma_k^2}, m_k+1, \frac{\beta_k^2(\tau,V_k(t))}{4(\lambda_k +K_k(\tau))}\right),\label{applicationtheorem}
\end{aligned}
\end{align}
with
\begin{align}
m_k &=\frac{2}{\sigma_k^2}\sqrt{\left(\kappa_k\theta_k-\frac{\sigma_k^2}{2}\right)^2+2\sigma_k^2 \nu_k},\\
A_k&=\kappa_k^2 +2\mu_k\sigma_k^2,\\
\beta_k(\tau,V_k(t)) &= \frac{2\sqrt{A_k  V_k(t)}}{\sigma_k^2 \sinh\left(\frac{\sqrt{A_k}(T-t)}{2}\right)},\\
K_k(\tau)&= \frac{1}{\sigma_k^2}\left( \sqrt{A_k}\coth\left(\frac{\sqrt{A_k}(T-t)}{2}\right)+\kappa_k\right).
\end{align}
Concerning the convergence set, we have
\begin{align*}
\Excond{}{\left|e^{\imag\langle\zeta,Y(T)\rangle}\right|}{\cF_t}=\Excond{}{e^{-\langle\Im(\zeta),Y(T)\rangle}}{\cF_t}
\end{align*}
so that we are induced to study the regularity at the point $-\Im(\boldsymbol{\zeta})$. Now, for each $k=1,\ldots,d$ Equation \eqref{conditionmu} up to \eqref{conditionLambda} from \cite{grasselli13} allow us to obtain conditions i-iv on $f^l_k(-\Im(\boldsymbol{\zeta})), \ l=1,\ldots,4$. Also, \eqref{eq:explosionTime} can be directly inferred from \eqref{eq:generalExplosionTime}. The conclusions on the convergence set follow along the arguments of \cite{grasselli13}.

\section{Proof of Lemma \ref{criterioncf}}\label{Appendixcriterioncf}

Let $-1<\alpha<0$, then we have
\begin{align*}
\phi^{i,j}(z)&=D^i(t)\mathbb{E}\left[\left.\frac{1}{D^i(T)}e^{\mathtt{i}z x^{i,j}(T)}\right|\cF_t\right]\\
&\leq D^i(t)\mathbb{E}\left[\left.\frac{1}{D^i(T)}e^{-\alpha x^{i,j}(T)}\right|\cF_t\right] =  D^i(t)\mathbb{E}\left[\left.\frac{1}{D^i(T)}\left(\frac{D^{i}(T)}{D^{j}(T)}\right)^{-\alpha}\right|\cF_t\right].
\end{align*}
We introduce, in line with Section 7.3.3 of \cite{bookjanplaten13}, a forward measure $\mathbb{P}^T$ with associated num\'eraire given by the zero coupon bond price process $P^i(.,T)=\left\{P^i(t,T), 0\leq t\leq T\right\}$. Recall that the minimal possible benchmarked price of a derivative is a martingale under the benchmark approach, hence the benchmarked zero coupon bond price allows us to define the following density process
\begin{align*}
\left.\frac{\partial \mathbb{P}^T}{\partial \mathbb{P}}\right|_{\cF_t}=Z_t:=\frac{P^i(t,T)D^i(0)}{P^i(0,T)D^i(t)}, \ t\leq T.
\end{align*}
Using such a forward measure, in conjunction with Jensen's inequality for concave functions, allows us to write
\begin{align*}
&D^i(t)\mathbb{E}\left[\left.\frac{1}{D^i(T)}\left(\frac{D^{i}(T)}{D^{j}(T)}\right)^{-\alpha}\right|\cF_t\right]=P^i(t,T)\mathbb{E}^{\mathbb{P}^T}\left[\left.\left(\frac{D^{i}(T)}{D^{j}(T)}\right)^{-\alpha}\right|\cF_t\right]\\
&\leq P^i(t,T)\mathbb{E}^{\mathbb{P}^T}\left[\left.\left(\frac{D^{i}(T)}{D^{j}(T)}\right)\right|\cF_t\right]^{-\alpha}=P^i(t,T)\left(\frac{D^i(t)}{P^i(t,T)}\right)^{-\alpha}\mathbb{E}\left[\left.\left(\frac{1}{D^{j}(T)}\right)\right|\cF_t\right]^{-\alpha}\\
&=P^{i}(t,T)\left(\frac{D^i(t)}{D^j(t)}\right)^{-\alpha}\left(\frac{P^j(t,T)}{P^i(t,T)}\right)^{-\alpha}=P^{i}(t,T)\left(S^{i,j}(t)\right)^{-\alpha}\left(\frac{P^j(t,T)}{P^i(t,T)}\right)^{-\alpha}.
\end{align*}
which proves the finiteness of the characteristic function. The analytic extension to the set $\mathcal{Z}$ is then a direct consequence of Theorem 7.1.1 in \cite{lukacs}.

\section{The General Transform of the CIR Process}
For  convenience we recall the formula in \cite{grasselli13} giving the general transform for the CIR process $X=\left\{ X_t,t\in [0,T] \right\}$ that extends formulas in \cite{cradlennox09} to the case 
\begin{align}
\label{eq:generalExpectationCIR}
\mathbb{E}_t \left[
 X_T^{-\alpha}e^{-\lambda X_T - \mu \int_t^T X_s ds -\nu \int_t^T\frac{ds}{X_s}}\right].
\end{align}

\begin{theorem}(\cite{grasselli13})\label{general}
Assume that $X_t$ satisfies the SDE  
\begin{equation}
dX_t=(a-bX_t)dt+\sqrt{2\sigma X_t}dW_t,\ \ X_0=x>0,
\end{equation}
with $a,b,\sigma >0$ and $a >\sigma$ (Feller condition). Given $\alpha, \mu, \nu,\lambda$ such that
\begin{align}
b^2 +4\mu\sigma &\geq0 \label{conditionmu}\\
(a-\sigma)^2+4\sigma \nu\geq 0\\
\frac{1}{2}+\frac{a}{2\sigma} + \frac{1}{2\sigma}\sqrt{(a-\sigma)^2+4\sigma \nu}&>\alpha,
\label{condition}\\
\lambda\geq -\frac{\sqrt{b^2+4\mu\sigma}+b}{2\sigma}\label{conditionLambda}
\end{align}
then the transform \eqref{eq:generalExpectationCIR} is well defined for all $t\geq 0$ and is given by
\begin{align}
\begin{split}
\mathbb{E} &\left[
 X_t^{-\alpha}e^{-\lambda X_t - \mu \int_0^t X_s ds -\nu \int_0^t\frac{ds}{X_s}}\right]\\
 &=\left( \frac{\beta(t,x)}{2}\right)^{m+1} x^{-\frac{a}{2\sigma}}(\lambda+K(t))^{-\left(\frac{1}{2}+\frac{m}{2}-\alpha +\frac{a}{2\sigma}\right)}\\
&\times e^{\frac{1}{2\sigma}\left( abt - \sqrt{A}x\coth\left(\frac{\sqrt{A}t}{2}\right)+bx\right)}
\frac{\Gamma \left(\frac{1}{2}+\frac{m}{2}-\alpha +\frac{a}{2\sigma}\right) }{\Gamma (m+1)}\\
&\times{}_1F_1 \left( \frac{1}{2}+\frac{m}{2}-\alpha +\frac{a}{2\sigma}, m+1, \frac{\beta^2(t,x)}{4(\lambda +K(t))}\right),
\end{split}
\end{align}
with
\begin{align}
m &= \frac{1}{\sigma}\sqrt{(a-\sigma)^2+4\sigma \nu},\\
A&= b^2 +4\mu\sigma,\\
\beta(t,x) &= \frac{\sqrt{Ax}}{\sigma \sinh\left(\frac{\sqrt{A}t}{2}\right)},\\
K(t)&= \frac{1}{2\sigma}\left( \sqrt{A}\coth\left(\frac{\sqrt{A}t}{2}\right)+b\right).
\end{align}
If 
\begin{align}
\lambda< -\frac{\sqrt{b^2+4\mu\sigma}+b}{2\sigma},\label{conditionLambdaExplosion}
\end{align}
then the transform is well defined for all $t< t^\star$, with
\begin{align}
\label{eq:generalExplosionTime}
t^\star=\frac{1}{\sqrt{A}}\log\left(1-\frac{2\sqrt{A}}{b+2\sigma\lambda+\sqrt{A}}\right).
\end{align}
\end{theorem}

\section{Figures and Tables}

\begin{figure}[h]
\centering
  \subfloat{\label{10yprocess}\includegraphics[scale=0.35]{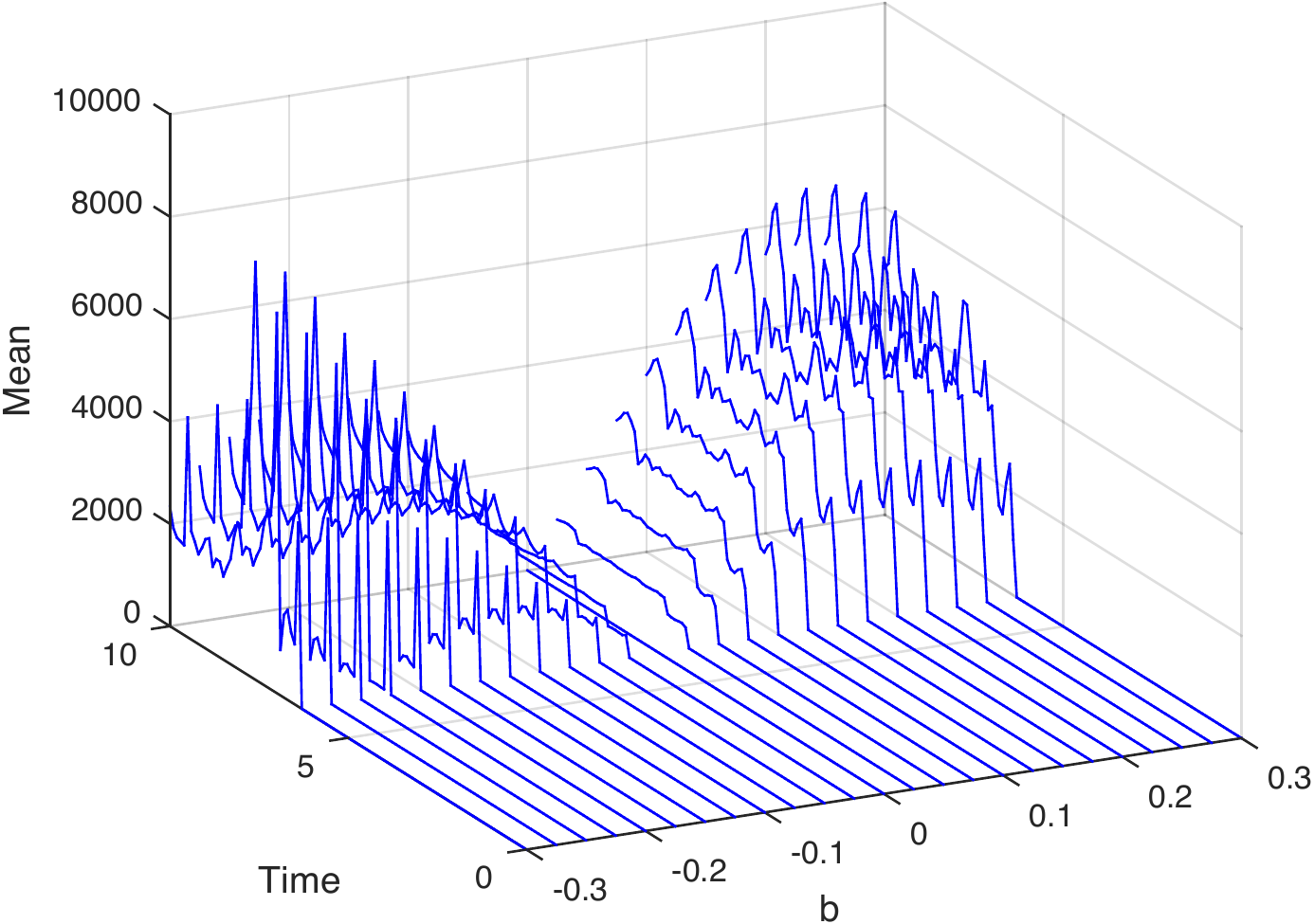}} 
  \ \subfloat{\label{10yQV}\includegraphics[scale=0.35]{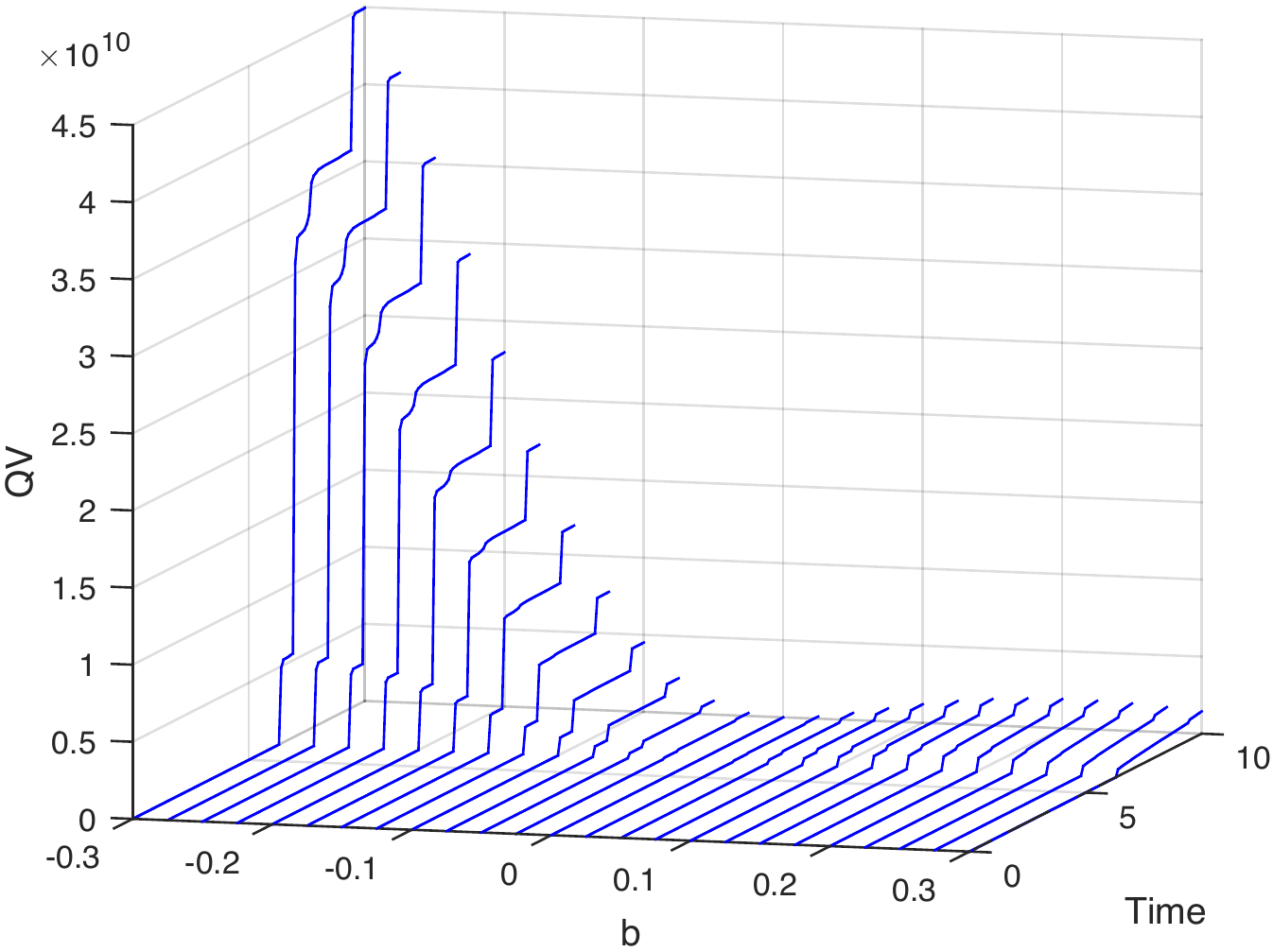}}
\caption{Simulation of the Radon-Nikodym derivative for the putative risk neutral measure of the
$i$-th currency denomination $\hat{B}^i(t)=\frac{B^i(t)}{D^i(t)}$ given by the SDE \eqref{putativeRadonNikodym}, together with the relative quadratic variation process. The time horizon is $t=10$ years. We consider a one factor specification of the model (i.e. $d=1$) and we fix the parameters as $\kappa = 0.49523; \theta = 0.53561;\sigma = 0.67128; V(0) = 1.4338;\rho = -0.89728; a = 0.047360.$ Parameter $b$ ranges between $[-0.4, 0.4]$. For $b$ positive the process is a true martingale, while for $b$ negative the process is a strict local martingale.\label{strictlocalmartingale10}}
\end{figure}

\begin{figure}[h]
\centering
 \subfloat{\label{eurusd}\includegraphics[scale=0.35]{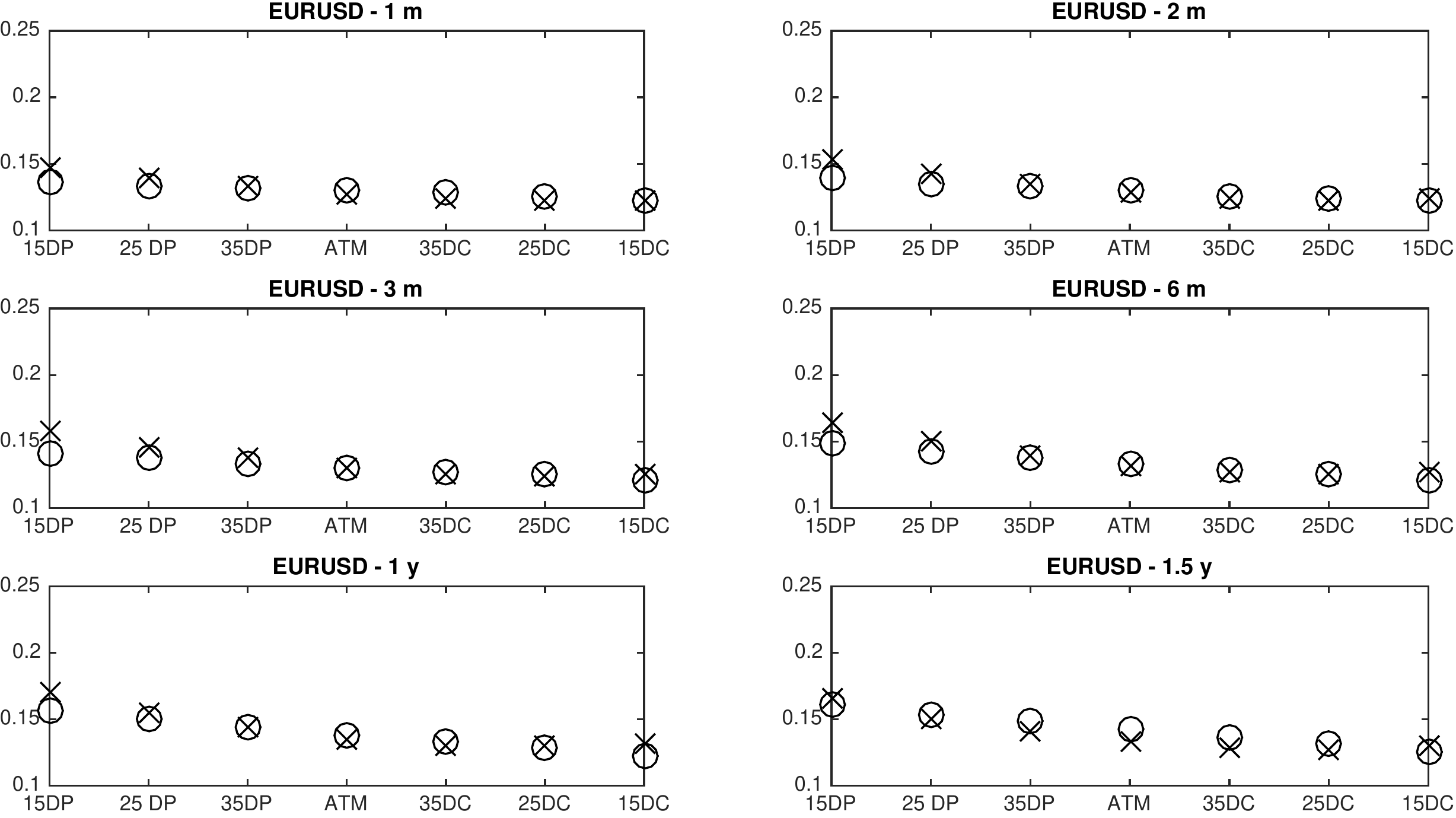}} 
  \   \subfloat{\label{eurjpy}\includegraphics[scale=0.35]{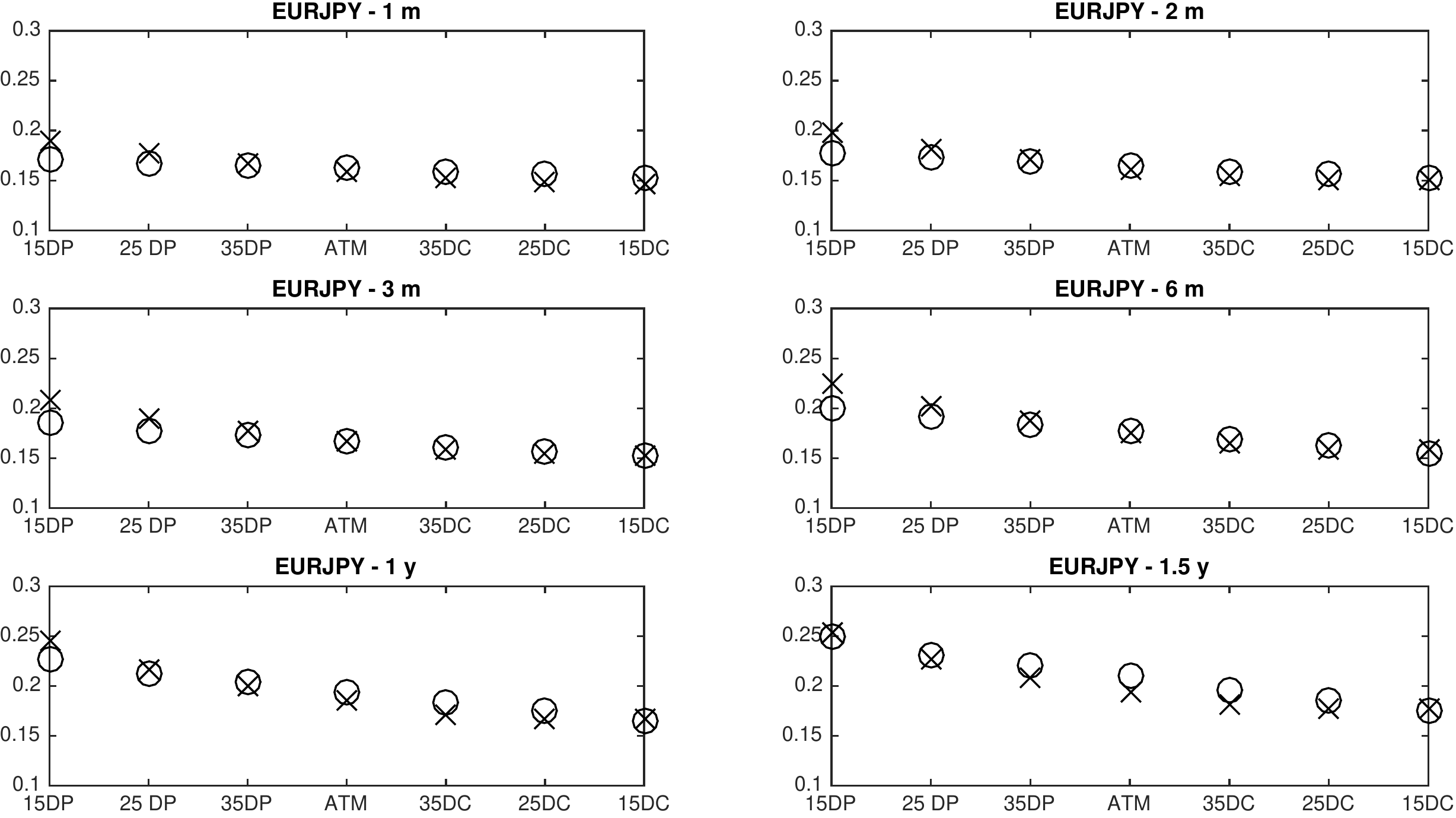}}
  \ \subfloat{\label{usdjpy}\includegraphics[scale=0.35]{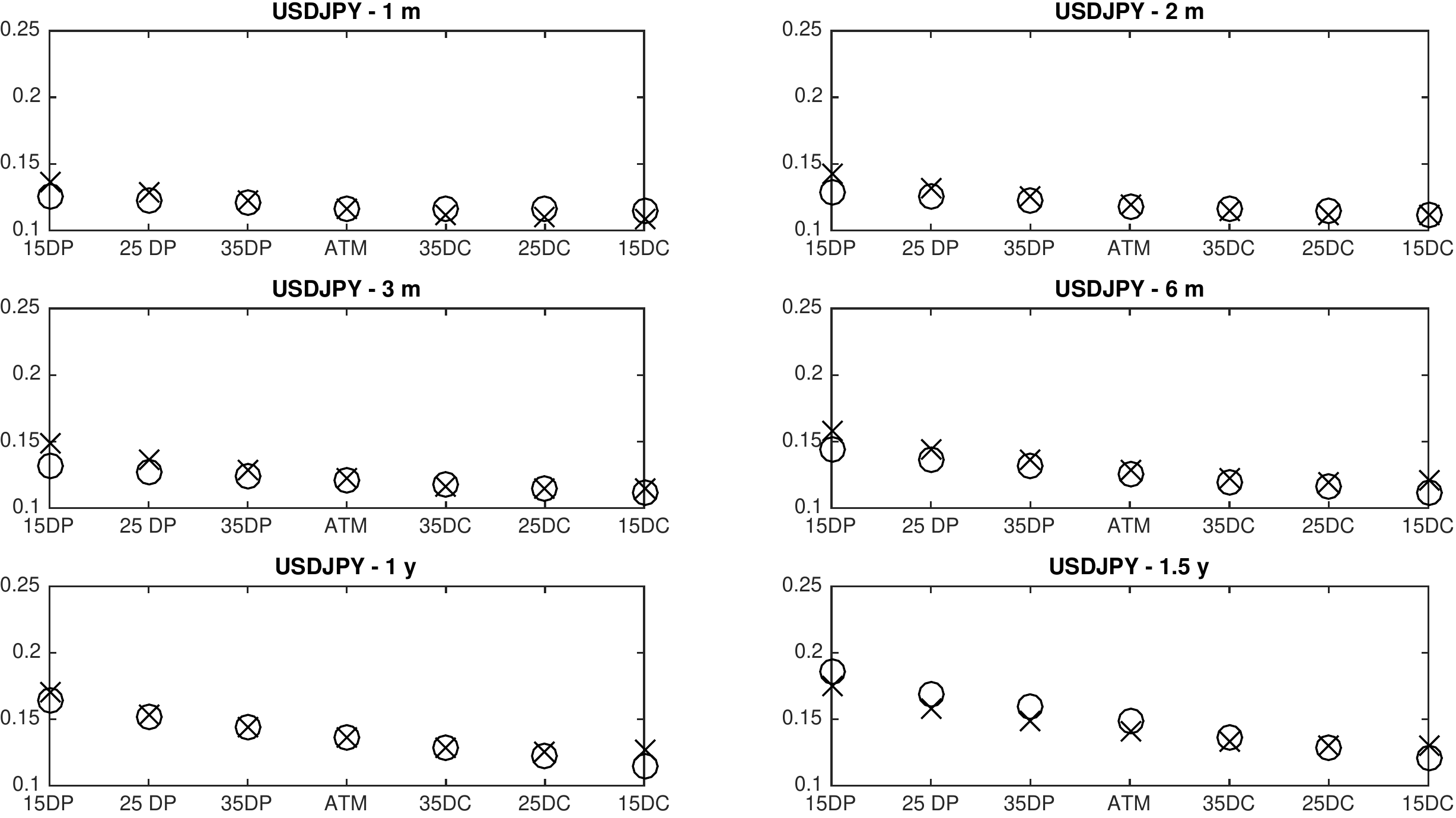}}
\caption{Simultaneous FX calibration. Market data as of July 22${}^{nd}$, 2010. Market volatilities are denoted by crosses, model volatilities are denoted by circles. Moneyness levels follow the standard Delta quoting convention in the FX option market. DC and DP stand for "delta call" and "delta put" respectively.\label{calibration2010}}
\end{figure}

\begin{table}[h]
\centering
\begin{tabular}{cccc}
Parameter & Value & Parameter & Value \\ 
\hline 
$V_1(0)$ & 0.8992 & $V_2(0)$ & 1.3011 \\ 
$\kappa_1$ & 1.1705 & $\kappa_2$ & 0.5110 \\ 
$\theta_1$ & 0.6853 & $\theta_2$ & 0.5206 \\ 
$\sigma_1$ & 0.3980 & $\sigma_2$ & 0.6786 \\ 
$\rho_1$ &  -0.8637 & $\rho_2$ &  -0.8925 \\ 
$a^{usd}_1$ & 0.0729 & $a^{usd}_2$ & 0.0621 \\ 
$a^{eur}_1$ & 0.0218 & $a^{eur}_2$ & 0.0414 \\ 
$a^{jpy}_1$ & 0.1497 & $a^{jpy}_2$ & 0.0687 \\  
$b^{usd}_1$ & 0.0192 & $b^{usd}_2$ & 0.0678 \\ 
$b^{eur}_1$ & 0.1805 & $b^{eur}_2$ & 0.0578 \\ 
$b^{jpy}_1$ & -0.0601 & $b^{jpy}_2$ & -0.0712 \\  
\hline 
\end{tabular} 
\vspace{0.3cm}
\caption{Parameter values resulting from the calibration procedure of a two factor specification with deterministic interest rate. Market data as of July 22${}^{nd}$, 2010. Each column corresponds to a volatility factor.}\label{parameters}
\end{table}
\begin{table}[h]
\centering
\begin{tabular}{cccc}
Measure & Test & Value & Risk Neutral Pricing Possible \\ 
\hline
\multirow{2}{*}{$\mathbb{P}$} & \multirow{2}{*}{$2\kappa_k\theta_k-\sigma^2_k$} &1.4459 &\\
&&0.0715&\\ 
\multirow{2}{*}{$\mathbb{Q}^{usd}$} &  \multirow{2}{*}{$2\kappa_k\theta_k-\left(\sigma^2_k+2\rho_k\sigma_k b^{usd}_k\right)$}&1.4591 &\multirow{2}{*}{YES}\\
&&0.1536&\\ 
\multirow{2}{*}{$\mathbb{Q}^{eur}$} &  \multirow{2}{*}{$2\kappa_k\theta_k-\left(\sigma^2_k+2\rho_k\sigma_k b^{eur}_k\right)$}&1.5700 &\multirow{2}{*}{YES}\\
&&0.1415&\\ 
\multirow{2}{*}{$\mathbb{Q}^{jpy}$} &  \multirow{2}{*}{$2\kappa_k\theta_k-\left(\sigma^2_k+2\rho_k\sigma_k b^{jpy}_k\right)$} & 1.4046&\multirow{2}{*}{NO}\\
&&-0.0147&\\ 
\hline 
\end{tabular} 
\vspace{0.3cm}
\caption{This table reports the Feller test under different (putative) measures as introduced in Section~\ref{subsecstrictlm}. Market data as of July 22${}^{nd}$, 2010. The values of the test for $V_1$ (upper value) and $V_2$ (lower value) are obtained from the calibrated parameters reported in Table~\ref{parameters}.}\label{fellertest}
\end{table}

%
\begin{table}[h]
\centering
\begin{tabular}{lccl}
$T$	& Risk Neutral Approach & Benchmark approach & Difference\\
\hline
$1$ week & $0.99990$ & $0.99990$ & $1.3922\times 10^{-13}$ \\
$2$ weeks & $0.99981$ & $0.99981$ & $6.6835\times 10^{-14}$ \\
$3$ weeks & $0.99971$ & $0.99971$ & $3.0198\times 10^{-14}$ \\
$1$ month& $0.99959$ & $0.99959$ & $4.6740\times 10^{-14}$ \\
$2$ months& $0.99917$ & $0.99917$ & $7.3275\times 10^{-15}$ \\
$3$ months& $0.99879$ & $0.99879$ & $2.3093\times 10^{-14}$ \\
$6$ months& $0.99740$ & $0.99740$ & $4.6629\times 10^{-15}$ \\
$1$ year& $0.99377$ & $0.99377$ & $1.5543\times 10^{-15}$ \\
$1.5$ years & $0.98849$ & $0.98849$ & $4.9960\times 10^{-15}$ \\
$2$ years& $0.97776$ & $0.97776$ & $6.6613\times 10^{-16}$ \\
$3$ years& $0.95556$ & $0.95556$ & $1.6653\times 10^{-15}$ \\
$5$ years& $0.91031$ & $0.91031$ & $5.4401\times 10^{-15}$ \\
$7$ years& $0.84064$ & $0.84064$ & $1.4211\times 10^{-14}$ \\
$10$ years& $0.73800$ & $0.73800$ & $1.4100\times 10^{-14}$ \\
\hline
\end{tabular}
\vspace{0.3cm}
\caption{Computation of Zero coupon bond prices for the USD market July 22${}^{nd}$, 2010, where risk neutral valuation is potentially allowed.}\label{ZCBUSD}
\end{table}
%

\begin{table}[h]
\centering
\begin{tabular}{lccl}
$T$	& Risk Neutral Approach & Benchmark approach & Difference\\
\hline
$1$ week & $0.99993$ & $0.99993$ & $1.6187\times 10^{-13}$ \\
$2$ weeks& $0.99987$ & $0.99987$ & $3.2196\times 10^{-14}$ \\
$3$ weeks& $0.99980$ & $0.99980$ & $9.7811\times 10^{-14}$ \\
$1$ month& $0.99971$ & $0.99971$ & $8.9928\times 10^{-15}$ \\
$2$ months& $0.99943$ & $0.99943$ & $1.5543\times 10^{-15}$ \\
$3$ months& $0.99913$ & $0.99913$ & $8.9928\times 10^{-13}$ \\
$6$ months& $0.99827$ & $0.99827$ & $1.4083\times 10^{-7}$ \\
$1$ year& $0.99679$ & $0.99671$ & $7.6955\times 10^{-5}$ \\
$1.5$ years & $0.99556$ & $0.99482$ & $0.00073726$ \\
$2$ years& $0.99475$ & $0.99229$ & $0.0024617$ \\
$3$ years& $0.99229$ & $0.98319$ & $0.0090951$ \\
$5$ years& $0.98594$ & $0.95580$ & $0.030140$ \\
$7$ years& $0.96885$ & $0.91430$ & $0.054544$ \\
$10$ years& $0.92073$ & $0.83195$ & $0.088777$ \\
$15$ years& $0.83391$ & $0.69949$ & $0.13442$ \\
$20$ years& $0.74894$ & $0.58296$ & $0.16598$ \\  
\hline
\end{tabular}
\vspace{0.3cm}
\caption{Computation of zero coupon bond prices for the JPY market as of July 22${}^{nd}$, 2010, where risk neutral valuation is not possible. The associated price difference increases as the maturity becomes larger.}\label{ZCBJPY}
\end{table}

\begin{table}[h]
\centering
\begin{tabular}{crrrrr}
Parameter & 02.23.2015&03.23.2015&04.22.2015&05.22.2015&06.22.2015\\
\hline
$V_1(0)$ &$0.73431$ & $1.1740$ & $1.1719$ & $1.1736$ & $1.1801$ \\
$V_2(0)$ &$1.1943$ & $1.4352$ & $1.4338$ & $1.4356$ & $1.4345$ \\
$\kappa_1$&$0.46330$ & $0.48651$ & $0.48812$ & $0.48968$ & $0.48950$ \\
$\kappa_2$& $0.52193$ & $0.49237$ & $0.49523$ & $0.49191$ & $0.49282$ \\
$\theta_1$ &$0.35253$ & $0.29821$ & $0.29999$ & $0.30107$ & $0.30082$ \\
$\theta_2$ &$0.75203$ & $0.53498$ & $0.53561$ & $0.53464$ & $0.53455$ \\
$\sigma_1$&$0.49603$ & $0.49110$ & $0.49443$ & $0.49604$ & $0.49561$ \\
$\sigma_2$& $0.66899$ & $0.66938$ & $0.67128$ & $0.67135$ & $0.67102$ \\
$\rho_1$ & $-0.99327$ & $-0.86543$ & $-0.85665$ & $-0.86540$ & $-0.86434$ \\
$\rho_2$ & $-0.94111$ & $-0.89690$ & $-0.89728$ & $-0.89706$ & $-0.89558$ \\
$a^{usd}_1$& $0.12445$ & $0.13047$ & $0.12885$ & $0.12399$ & $0.12616$ \\
$a^{usd}_2$& $0.062812$ & $0.036533$ & $0.047360$ & $0.058480$ & $0.058355$ \\
$a^{eur}_1$ &$0.045806$ & $0.038959$ & $0.049467$ & $0.053229$ & $0.058580$ \\
$a^{eur}_2$ &$0.044107$ & $0.054640$ & $0.051325$ & $0.045735$ & $0.048780$ \\
$a^{jpy}_1$ &$0.16458$ & $0.15589$ & $0.14943$ & $0.14803$ & $0.14984$ \\
$a^{jpy}_2$ &$0.057533$ & $0.080895$ & $0.073395$ & $0.067905$ & $0.073300$ \\
$b^{usd}_1$&$-0.10105$ & $-0.041716$ & $-0.057950$ & $-0.059954$ & $-0.055886$ \\
$b^{usd}_2$&$-0.11178$ & $-0.094706$ & $-0.12360$ & $-0.12433$ & $-0.12753$ \\
$b^{eur}_1$ & $-0.065631$ & $-0.041879$ & $-0.063733$ & $-0.072158$ & $-0.074323$ \\
$b^{eur}_2$ & $0.0076657$ & $-0.037459$ & $-0.039780$ & $-0.044664$ & $-0.041933$ \\
$b^{jpy}_1$ &$-0.059432$ & $-0.047197$ & $-0.053469$ & $-0.054806$ & $-0.055076$ \\
$b^{jpy}_2$ &$-0.065707$ & $-0.051282$ & $-0.058217$ & $-0.062116$ & $-0.056492$ \\  
\hline
Res. Norm.&$0.0131$&$0.0121$&$0.0106$&$0.0269$&$ 0.0141$\\ 
\hline
\end{tabular} 
\vspace{0.3cm}
\caption{Parameter values resulting from the repeated calibration procedure of a two factor specification with constant interest rates.}\label{parameters2015}
\end{table}

\begin{table}[h]
\centering
\begin{tabular}{cc|rrrrr}
Measure & Test &
02.23.2015&03.23.2015&04.22.2015&05.22.2015&06.22.2015\\
\hline
\multirow{2}{*}{$\mathbb{P}$} & \multirow{2}{*}{$2\kappa_k\theta_k-\sigma^2_k$}&$0.0806$&$0.0490$&$0.0484$&$0.0488$ &$0.0489$ \\
&&$0.3375$&$0.0787$&$0.0799$&$0.0753$&$0.0766$\\ 
\multirow{2}{*}{$\mathbb{Q}^{usd}$} &  \multirow{2}{*}{$2\kappa_k\theta_k-\left(\sigma^2_k+2\rho_k\sigma_k b^{usd}_k\right)$}&$-0.0189$&$0.0135$&$-0.0007$&$-0.0027$&$0.0010$ \\
&&$0.1967$&$-0.0350$&$-0.0690$&$-0.0745$&$-0.0767$\\ 
\multirow{2}{*}{$\mathbb{Q}^{eur}$} &  \multirow{2}{*}{$2\kappa_k\theta_k-\left(\sigma^2_k+2\rho_k\sigma_k b^{eur}_k\right)$}&$0.0159$&$0.0134$&$-0.0056$&$-0.0132$&$-0.0148$ \\
&&$0.3471$&$0.0338$&$ 0.0320$&$ 0.0215$&$0.0262$\\ 
\multirow{2}{*}{$\mathbb{Q}^{jpy}$} &  \multirow{2}{*}{$2\kappa_k\theta_k-\left(\sigma^2_k+2\rho_k\sigma_k b^{jpy}_k\right)$}&$0.0221$&$0.0089$&$0.0031$&$0.0017$ &$0.0017$\\
&&$0.2547$&$0.0172$&$0.0097$&$0.0005$&$0.0087$\\ 
\hline 
\hline
$\mathbb{P}$ 		  &  & &&&&\\
$\mathbb{Q}^{usd}$& \multirow{3}{*}{Risk neutral pricing?} &NO&NO&NO &NO &NO\\
$\mathbb{Q}^{eur}$&																			 &YES&YES&NO &NO &NO\\
$\mathbb{Q}^{jpy}$&																			 &YES&YES&YES&YES&YES\\
\hline
\end{tabular} 
\vspace{0.3cm}
\caption{This table reports the Feller test under different (putative) measures as introduced in Section~\ref{subsecstrictlm}. The values of the test are obtained from the calibrated parameters reported in Table~\ref{parameters2015}.}\label{fellertest2015}
\end{table}

\begin{table}[htp]
\centering
\begin{scriptsize}
\rotatebox{90}{
\begin{tabular}{rrrr|rrrr|rrrr}
\multicolumn{4}{c}{USD}&\multicolumn{4}{c}{EUR}&\multicolumn{4}{c}{JPY}\\
$T$	& RNA & BA & Difference&$T$	& RNA & BA & Difference&$T$	& RNA & BA & Difference\\
\hline
$0.25008$ & $0.99935$ & $0.99935$ & $6.0851\times 10^{-13}$ & $0.50267$ & $0.99940$ & $0.99940$ & $1.1990\times 10^{-13}$ & $0.50214$ & $0.99931$ & $0.99932$ & $1.6497\times 10^{-7}$ \\
$0.31088$ & $0.99915$ & $0.99915$ & $7.8564\times 10^{-11}$ & $0.58291$ & $0.99930$ & $0.99931$ & $1.7908\times 10^{-13}$ & $0.58058$ & $0.99919$ & $0.99919$ & $1.0711\times 10^{-7}$ \\
$0.38847$ & $0.99889$ & $0.99889$ & $1.6583\times 10^{-9}$ & $0.66872$ & $0.99920$ & $0.99920$ & $2.4092\times 10^{-13}$ & $0.66567$ & $0.99907$ & $0.99907$ & $2.6713\times 10^{-8}$ \\
$0.48311$ & $0.99851$ & $0.99851$ & $1.0495\times 10^{-7}$ & $0.74971$ & $0.99910$ & $0.99910$ & $1.2557\times 10^{-13}$ & $0.74738$ & $0.99895$ & $0.99895$ & $6.0677\times 10^{-9}$ \\
$0.55772$ & $0.99814$ & $0.99814$ & $6.5352\times 10^{-7}$ & $0.84145$ & $0.99899$ & $0.99899$ & $1.3323\times 10^{-15}$ & $0.84052$ & $0.99882$ & $0.99882$ & $4.4317\times 10^{-11}$ \\
$0.65500$ & $0.99766$ & $0.99765$ & $3.6805\times 10^{-6}$ & $0.91573$ & $0.99890$ & $0.99890$ & $1.5266\times 10^{-13}$ & $0.91445$ & $0.99871$ & $0.99871$ & $1.7248\times 10^{-11}$ \\
$0.72984$ & $0.99720$ & $0.99719$ & $1.0367\times 10^{-5}$ & $0.99918$ & $0.99881$ & $0.99881$ & $1.7464\times 10^{-13}$ & $0.99979$ & $0.99859$ & $0.99859$ & $1.8563\times 10^{-13}$ \\
$0.80747$ & $0.99669$ & $0.99666$ & $2.4915\times 10^{-5}$ & $1.0940$ & $0.99870$ & $0.99870$ & $1.8863\times 10^{-13}$ & $1.4984$ & $0.99789$ & $0.99788$ & $4.0501\times 10^{-12}$ \\
$1.0571$ & $0.99467$ & $0.99449$ & $0.00018107$ & $1.1660$ & $0.99862$ & $0.99862$ & $4.8850\times 10^{-15}$ & $2.0078$ & $0.99714$ & $0.99714$ & $6.4399\times 10^{-11}$ \\
$1.3036$ & $0.99210$ & $0.99146$ & $0.00063221$ & $1.2490$ & $0.99852$ & $0.99852$ & $7.3275\times 10^{-15}$ & $3.0043$ & $0.99531$ & $0.99531$ & $1.7337\times 10^{-11}$ \\
$1.5694$ & $0.98872$ & $0.98711$ & $0.0016107$ & $1.3381$ & $0.99842$ & $0.99842$ & $2.6534\times 10^{-14}$ & $4.0007$ & $0.99248$ & $0.99249$ & $5.6222\times 10^{-13}$ \\
$1.8185$ & $0.98507$ & $0.98200$ & $0.0030724$ & $1.4160$ & $0.99833$ & $0.99833$ & $4.1855\times 10^{-14}$ & $5.0000$ & $0.98843$ & $0.98844$ & $2.2760\times 10^{-14}$ \\
$2.0508$ & $0.98128$ & $0.97635$ & $0.0049275$ & $1.4995$ & $0.99823$ & $0.99823$ & $9.4147\times 10^{-14}$ & $6.0010$ & $0.98272$ & $0.98272$ & $4.7500\times 10^{-12}$ \\
$2.3162$ & $0.97646$ & $0.96886$ & $0.0076017$ & $2.0046$ & $0.99756$ & $0.99756$ & $7.1276\times 10^{-14}$ & $6.9994$ & $0.97557$ & $0.97557$ & $3.0261\times 10^{-12}$ \\
$2.5623$ & $0.97165$ & $0.96109$ & $0.010564$ & $3.0006$ & $0.99524$ & $0.99524$ & $1.3289\times 10^{-13}$ & $8.0024$ & $0.96682$ & $0.96682$ & $5.0325\times 10^{-11}$ \\
$2.8109$ & $0.96659$ & $0.95261$ & $0.013975$ & $3.9955$ & $0.99131$ & $0.99131$ & $6.1114\times 10^{-12}$ & $8.9971$ & $0.95673$ & $0.95673$ & $9.7525\times 10^{-11}$ \\
$3.0623$ & $0.96131$ & $0.94352$ & $0.017795$ & $4.9924$ & $0.98537$ & $0.98537$ & $6.8156\times 10^{-11}$ & $9.9940$ & $0.94504$ & $0.94504$ & $5.5121\times 10^{-10}$ \\
$3.9848$ & $0.94113$ & $0.90701$ & $0.034118$ & $5.9886$ & $0.97758$ & $0.97758$ & $1.7274\times 10^{-10}$ & $11.987$ & $0.91834$ & $0.91835$ & $4.0101\times 10^{-9}$ \\
$4.9787$ & $0.91801$ & $0.86411$ & $0.053900$ & $6.9837$ & $0.96781$ & $0.96781$ & $3.5236\times 10^{-10}$ & $14.976$ & $0.86795$ & $0.86795$ & $2.3196\times 10^{-8}$ \\
$5.9721$ & $0.89434$ & $0.82012$ & $0.074226$ & $7.9833$ & $0.95645$ & $0.95645$ & $6.0173\times 10^{-10}$ & $19.954$ & $0.77948$ & $0.77948$ & $1.2711\times 10^{-7}$ \\
$6.9654$ & $0.87063$ & $0.77653$ & $0.094102$ & $8.9739$ & $0.94389$ & $0.94389$ & $8.7084\times 10^{-10}$ & $24.933$ & $0.70159$ & $0.70159$ & $3.1501\times 10^{-7}$ \\
$7.9641$ & $0.84702$ & $0.73393$ & $0.11309$ & $9.9644$ & $0.93107$ & $0.93107$ & $1.1074\times 10^{-9}$ & $29.915$ & $0.63701$ & $0.63701$ & $5.2597\times 10^{-7}$ \\
$8.9545$ & $0.82395$ & $0.69328$ & $0.13067$ & $10.957$ & $0.91722$ & $0.91722$ & $1.3277\times 10^{-9}$ & $34.887$ & $0.57631$ & $0.57631$ & $9.1548\times 10^{-7}$ \\
$9.9446$ & $0.80111$ & $0.65428$ & $0.14683$ & $11.949$ & $0.90298$ & $0.90298$ & $1.5028\times 10^{-9}$ & $39.863$ & $0.51867$ & $0.51867$ & $1.3125\times 10^{-6}$ \\
$10.937$ & $0.77855$ & $0.61696$ & $0.16158$ & $14.925$ & $0.86068$ & $0.86068$ & $2.9209\times 10^{-8}$ & -- & -- & -- & -- \\
$11.930$ & $0.75639$ & $0.58149$ & $0.17490$ & $19.887$ & $0.79260$ & $0.79260$ & $4.8524\times 10^{-7}$ & -- & -- & -- & -- \\
$14.909$ & $0.69362$ & $0.48653$ & $0.20709$ & $24.850$ & $0.73266$ & $0.73266$ & $1.3366\times 10^{-6}$ & -- & -- & -- & -- \\
$19.876$ & $0.60153$ & $0.36187$ & $0.23966$ & $29.812$ & $0.67847$ & $0.67847$ & $1.7194\times 10^{-6}$ & -- & -- & -- & -- \\
$24.844$ & $0.52388$ & $0.27024$ & $0.25364$ & $34.769$ & $0.62967$ & $0.62967$ & $2.1766\times 10^{-6}$ & -- & -- & -- & -- \\
$29.810$ & $0.45699$ & $0.20214$ & $0.25484$ & $39.735$ & $0.58814$ & $0.58815$ & $2.6240\times 10^{-6}$ & -- & -- & -- & -- \\
$39.737$ & $0.34968$ & $0.11375$ & $0.23593$ & $44.705$ & $0.55337$ & $0.55337$ & $1.8836\times 10^{-7}$ & -- & -- & -- & -- \\
$49.677$ & $0.27553$ & $0.065891$ & $0.20964$ & $49.680$ & $0.52631$ & $0.52631$ & $4.6528\times 10^{-12}$ & -- & -- & -- & --   
\\  
\hline
\end{tabular}
}
\end{scriptsize}
\caption{Computation of zero coupon bond prices for the USD, EUR and JPY markets as of February 23${}^{rd}$, 2015. $T$ refers to the maturity of the bonds. RNA and BA stand for risk neutral and benchmark approach respectively.}\label{ZCBUSD201502}
\end{table}
\newpage
\begin{table}[htp]
\centering
\begin{scriptsize}
\rotatebox{90}{
\begin{tabular}{rrrr|rrrr|rrrr}
\multicolumn{4}{c}{USD}&\multicolumn{4}{c}{EUR}&\multicolumn{4}{c}{JPY}\\
$T$	& RNA & BA & Difference&$T$	& RNA & BA & Difference&$T$	& RNA & BA & Difference\\
\hline
$0.25533$ & $0.99932$ & $0.99932$ & $9.5861\times 10^{-11}$ & $0.51135$ & $0.99955$ & $0.99955$ & $9.0594\times 10^{-14}$ & $0.51103$ & $0.99928$ & $0.99928$ & $1.8545\times 10^{-7}$ \\
$0.30542$ & $0.99915$ & $0.99915$ & $8.5724\times 10^{-11}$ & $0.58555$ & $0.99949$ & $0.99949$ & $1.5099\times 10^{-14}$ & $0.58844$ & $0.99915$ & $0.99915$ & $1.1805\times 10^{-7}$ \\
$0.39958$ & $0.99880$ & $0.99880$ & $7.7737\times 10^{-9}$ & $0.66625$ & $0.99942$ & $0.99942$ & $1.8086\times 10^{-13}$ & $0.67126$ & $0.99903$ & $0.99903$ & $3.0184\times 10^{-8}$ \\
$0.47454$ & $0.99848$ & $0.99848$ & $4.1679\times 10^{-8}$ & $0.75840$ & $0.99934$ & $0.99934$ & $1.3567\times 10^{-13}$ & $0.76384$ & $0.99889$ & $0.99889$ & $1.0159\times 10^{-8}$ \\
$0.57174$ & $0.99805$ & $0.99805$ & $5.2282\times 10^{-7}$ & $0.83249$ & $0.99927$ & $0.99927$ & $3.6859\times 10^{-14}$ & $0.83801$ & $0.99877$ & $0.99878$ & $1.9202\times 10^{-11}$ \\
$0.64658$ & $0.99761$ & $0.99761$ & $2.1247\times 10^{-6}$ & $0.91610$ & $0.99920$ & $0.99920$ & $1.5021\times 10^{-13}$ & $0.92271$ & $0.99865$ & $0.99865$ & $1.9128\times 10^{-11}$ \\
$0.72439$ & $0.99714$ & $0.99713$ & $6.8585\times 10^{-6}$ & $1.0107$ & $0.99912$ & $0.99912$ & $3.2863\times 10^{-14}$ & $1.0133$ & $0.99850$ & $0.99851$ & $1.8328\times 10^{-10}$ \\
$0.97381$ & $0.99533$ & $0.99524$ & $8.5524\times 10^{-5}$ & $1.0833$ & $0.99906$ & $0.99906$ & $4.4187\times 10^{-14}$ & $1.5090$ & $0.99774$ & $0.99774$ & $8.3891\times 10^{-11}$ \\
$1.2204$ & $0.99304$ & $0.99265$ & $0.00039463$ & $1.1657$ & $0.99899$ & $0.99899$ & $2.2204\times 10^{-16}$ & $2.0076$ & $0.99688$ & $0.99688$ & $7.0679\times 10^{-11}$ \\
$1.4862$ & $0.99008$ & $0.98888$ & $0.0011957$ & $1.2551$ & $0.99890$ & $0.99890$ & $1.9984\times 10^{-15}$ & $3.0039$ & $0.99483$ & $0.99483$ & $2.0172\times 10^{-11}$ \\
$1.7353$ & $0.98687$ & $0.98434$ & $0.0025306$ & $1.3330$ & $0.99883$ & $0.99883$ & $1.8985\times 10^{-14}$ & $4.0007$ & $0.99188$ & $0.99188$ & $1.7073\times 10^{-12}$ \\
$1.9680$ & $0.98352$ & $0.97917$ & $0.0043494$ & $1.4165$ & $0.99875$ & $0.99875$ & $1.8796\times 10^{-13}$ & $5.0023$ & $0.98762$ & $0.98762$ & $1.5022\times 10^{-11}$ \\
$2.2335$ & $0.97928$ & $0.97218$ & $0.0071035$ & $1.5022$ & $0.99866$ & $0.99866$ & $4.3410\times 10^{-14}$ & $6.0010$ & $0.98222$ & $0.98222$ & $4.6875\times 10^{-12}$ \\
$2.4797$ & $0.97506$ & $0.96479$ & $0.010271$ & $2.0040$ & $0.99804$ & $0.99804$ & $4.8406\times 10^{-14}$ & $6.9994$ & $0.97557$ & $0.97557$ & $2.1447\times 10^{-12}$ \\
$2.7283$ & $0.97060$ & $0.95658$ & $0.014016$ & $3.0007$ & $0.99579$ & $0.99579$ & $2.5235\times 10^{-13}$ & $8.0028$ & $0.96771$ & $0.96771$ & $6.6605\times 10^{-11}$ \\
$2.9798$ & $0.96593$ & $0.94763$ & $0.018296$ & $3.9959$ & $0.99172$ & $0.99172$ & $5.4496\times 10^{-12}$ & $8.9976$ & $0.95872$ & $0.95872$ & $6.0607\times 10^{-11}$ \\
$3.2258$ & $0.96110$ & $0.93820$ & $0.022900$ & $4.9930$ & $0.98605$ & $0.98605$ & $6.3637\times 10^{-11}$ & $9.9952$ & $0.94905$ & $0.94905$ & $3.0073\times 10^{-10}$ \\
$3.9863$ & $0.94654$ & $0.90743$ & $0.039111$ & $5.9895$ & $0.97950$ & $0.97950$ & $1.1563\times 10^{-10}$ & $11.989$ & $0.92458$ & $0.92458$ & $2.5580\times 10^{-9}$ \\
$4.9806$ & $0.92532$ & $0.86252$ & $0.062801$ & $6.9858$ & $0.97205$ & $0.97205$ & $1.9374\times 10^{-10}$ & $14.979$ & $0.87882$ & $0.87882$ & $1.4759\times 10^{-8}$ \\
$5.9746$ & $0.90349$ & $0.81622$ & $0.087274$ & $7.9870$ & $0.96366$ & $0.96366$ & $2.8133\times 10^{-10}$ & $19.962$ & $0.79462$ & $0.79462$ & $7.7188\times 10^{-8}$ \\
$6.9686$ & $0.88158$ & $0.77036$ & $0.11122$ & $8.9767$ & $0.95431$ & $0.95431$ & $3.6996\times 10^{-10}$ & $24.940$ & $0.72118$ & $0.72118$ & $2.2061\times 10^{-7}$ \\
$7.9677$ & $0.85950$ & $0.72547$ & $0.13404$ & $9.9715$ & $0.94464$ & $0.94464$ & $4.4520\times 10^{-10}$ & $29.923$ & $0.65910$ & $0.65910$ & $3.7710\times 10^{-7}$ \\
$8.9558$ & $0.83790$ & $0.68288$ & $0.15502$ & $10.966$ & $0.93435$ & $0.93435$ & $5.1396\times 10^{-10}$ & $34.897$ & $0.60089$ & $0.60089$ & $6.7433\times 10^{-7}$ \\
$9.9494$ & $0.81655$ & $0.64214$ & $0.17441$ & $11.960$ & $0.92407$ & $0.92407$ & $5.5639\times 10^{-10}$ & $39.876$ & $0.54584$ & $0.54584$ & $9.6345\times 10^{-7}$ \\
$10.942$ & $0.79239$ & $0.60114$ & $0.19125$ & $14.944$ & $0.89313$ & $0.89313$ & $1.0177\times 10^{-8}$ & -- & -- & -- & -- \\
$11.936$ & $0.77470$ & $0.56682$ & $0.20788$ & $19.921$ & $0.84496$ & $0.84496$ & $1.6587\times 10^{-7}$ & -- & -- & -- & -- \\
$14.916$ & $0.71570$ & $0.46956$ & $0.24614$ & $24.892$ & $0.80163$ & $0.80164$ & $4.3668\times 10^{-7}$ & -- & -- & -- & -- \\
$19.889$ & $0.62827$ & $0.34345$ & $0.28482$ & $29.869$ & $0.76029$ & $0.76029$ & $5.1555\times 10^{-7}$ & -- & -- & -- & -- \\
$24.855$ & $0.55373$ & $0.25226$ & $0.30148$ & $34.838$ & $0.72270$ & $0.72270$ & $6.0885\times 10^{-7}$ & -- & -- & -- & -- \\
$29.827$ & $0.48890$ & $0.18556$ & $0.30334$ & $39.813$ & $0.68800$ & $0.68800$ & $7.0764\times 10^{-7}$ & -- & -- & -- & -- \\
$39.760$ & $0.38294$ & $0.10093$ & $0.28201$ & $44.794$ & $0.66212$ & $0.66212$ & $5.0278\times 10^{-8}$ & -- & -- & -- & -- \\
$49.704$ & $0.30658$ & $0.056088$ & $0.25049$ & $49.778$ & $0.64076$ & $0.64076$ & $2.5646\times 10^{-13}$ & -- & -- & -- & --   
\\  
\hline
\end{tabular}
}
\end{scriptsize}
\caption{Computation of Zero coupon bond prices for the USD EUR and JPY markets as of March 23${}^{rd}$, 2015. $T$ refers to the maturity of the bonds. RNA and BA stand for risk neutral and benchmark approach respectively.}\label{ZCBUSD201503}
\end{table}

\begin{table}[htp]
\centering
\begin{scriptsize}
\rotatebox{90}{
\begin{tabular}{rrrr|rrrr|rrrr}
\multicolumn{4}{c}{USD}&\multicolumn{4}{c}{EUR}&\multicolumn{4}{c}{JPY}\\
$T$	& RNA & BA & Difference&$T$	& RNA & BA & Difference&$T$	& RNA & BA & Difference\\
\hline
$0.25270$ & $0.99930$ & $0.99930$ & $2.9560\times 10^{-11}$ & $0.51332$ & $0.99965$ & $0.99965$ & $9.3880\times 10^{-11}$ & $0.51337$ & $0.99928$ & $0.99929$ & $1.8200\times 10^{-7}$ \\
$0.31924$ & $0.99907$ & $0.99907$ & $1.9363\times 10^{-10}$ & $0.58321$ & $0.99959$ & $0.99959$ & $9.9334\times 10^{-10}$ & $0.58574$ & $0.99918$ & $0.99918$ & $1.0922\times 10^{-7}$ \\
$0.39414$ & $0.99882$ & $0.99882$ & $2.1562\times 10^{-9}$ & $0.66615$ & $0.99954$ & $0.99954$ & $8.7244\times 10^{-9}$ & $0.66845$ & $0.99907$ & $0.99907$ & $2.7860\times 10^{-8}$ \\
$0.49117$ & $0.99843$ & $0.99843$ & $1.8228\times 10^{-7}$ & $0.75313$ & $0.99948$ & $0.99948$ & $5.1639\times 10^{-8}$ & $0.75611$ & $0.99896$ & $0.99896$ & $7.4406\times 10^{-9}$ \\
$0.56612$ & $0.99807$ & $0.99807$ & $1.1914\times 10^{-6}$ & $0.83296$ & $0.99943$ & $0.99943$ & $1.9228\times 10^{-7}$ & $0.83837$ & $0.99885$ & $0.99885$ & $4.0214\times 10^{-11}$ \\
$0.64376$ & $0.99770$ & $0.99769$ & $5.1237\times 10^{-6}$ & $0.91576$ & $0.99937$ & $0.99937$ & $5.9486\times 10^{-7}$ & $0.91738$ & $0.99875$ & $0.99875$ & $2.0599\times 10^{-12}$ \\
$0.73806$ & $0.99715$ & $0.99713$ & $2.0281\times 10^{-5}$ & $1.0029$ & $0.99932$ & $0.99932$ & $1.6072\times 10^{-6}$ & $1.0051$ & $0.99863$ & $0.99863$ & $1.7454\times 10^{-12}$ \\
$0.89345$ & $0.99613$ & $0.99603$ & $0.00010597$ & $1.0826$ & $0.99927$ & $0.99927$ & $3.4880\times 10^{-6}$ & $1.5034$ & $0.99794$ & $0.99794$ & $1.5695\times 10^{-11}$ \\
$1.1401$ & $0.99412$ & $0.99352$ & $0.00060114$ & $1.1658$ & $0.99921$ & $0.99920$ & $7.0272\times 10^{-6}$ & $2.0021$ & $0.99717$ & $0.99717$ & $6.4672\times 10^{-12}$ \\
$1.4061$ & $0.99147$ & $0.98944$ & $0.0020368$ & $1.2524$ & $0.99915$ & $0.99913$ & $1.3257\times 10^{-5}$ & $3.0016$ & $0.99536$ & $0.99536$ & $3.3628\times 10^{-12}$ \\
$1.6554$ & $0.98855$ & $0.98399$ & $0.0045639$ & $1.3322$ & $0.99909$ & $0.99907$ & $2.2195\times 10^{-5}$ & $4.0007$ & $0.99262$ & $0.99262$ & $8.3755\times 10^{-13}$ \\
$1.8880$ & $0.98546$ & $0.97736$ & $0.0080942$ & $1.4222$ & $0.99902$ & $0.99898$ & $3.7160\times 10^{-5}$ & $5.0026$ & $0.98855$ & $0.98855$ & $1.9938\times 10^{-11}$ \\
$2.1536$ & $0.98148$ & $0.96796$ & $0.013519$ & $1.4991$ & $0.99896$ & $0.99890$ & $5.5059\times 10^{-5}$ & $6.0068$ & $0.98315$ & $0.98315$ & $1.8158\times 10^{-10}$ \\
$2.3998$ & $0.97743$ & $0.95764$ & $0.019794$ & $1.9994$ & $0.99848$ & $0.99813$ & $0.00035315$ & $7.0025$ & $0.97659$ & $0.97659$ & $3.5968\times 10^{-11}$ \\
$2.6486$ & $0.97311$ & $0.94589$ & $0.027224$ & $2.9979$ & $0.99656$ & $0.99409$ & $0.0024724$ & $7.9977$ & $0.96880$ & $0.96880$ & $3.9513\times 10^{-11}$ \\
$2.9001$ & $0.96855$ & $0.93285$ & $0.035694$ & $3.9964$ & $0.99321$ & $0.98626$ & $0.0069478$ & $8.9978$ & $0.95970$ & $0.95970$ & $4.8851\times 10^{-11}$ \\
$3.1462$ & $0.96380$ & $0.91903$ & $0.044765$ & $4.9942$ & $0.98839$ & $0.97506$ & $0.013324$ & $9.9955$ & $0.95025$ & $0.95025$ & $2.4617\times 10^{-10}$ \\
$3.9863$ & $0.94677$ & $0.86705$ & $0.079724$ & $5.9969$ & $0.98263$ & $0.96163$ & $0.021001$ & $11.995$ & $0.92638$ & $0.92638$ & $5.5962\times 10^{-10}$ \\
$4.9805$ & $0.92505$ & $0.80044$ & $0.12461$ & $6.9907$ & $0.97564$ & $0.94631$ & $0.029328$ & $14.979$ & $0.88141$ & $0.88141$ & $1.3357\times 10^{-8}$ \\
$5.9799$ & $0.90237$ & $0.73310$ & $0.16927$ & $7.9833$ & $0.96708$ & $0.92912$ & $0.037961$ & $19.958$ & $0.80015$ & $0.80015$ & $8.8003\times 10^{-8}$ \\
$6.9707$ & $0.87959$ & $0.66886$ & $0.21073$ & $8.9786$ & $0.95803$ & $0.91132$ & $0.046714$ & $24.940$ & $0.72892$ & $0.72892$ & $2.0908\times 10^{-7}$ \\
$7.9615$ & $0.85695$ & $0.60858$ & $0.24836$ & $9.9744$ & $0.95001$ & $0.89452$ & $0.055496$ & $29.922$ & $0.66975$ & $0.66975$ & $3.6738\times 10^{-7}$ \\
$8.9547$ & $0.83428$ & $0.55248$ & $0.28180$ & $10.969$ & $0.93920$ & $0.87524$ & $0.063966$ & $34.905$ & $0.61216$ & $0.61216$ & $5.4720\times 10^{-7}$ \\
$9.9480$ & $0.81206$ & $0.50103$ & $0.31104$ & $11.969$ & $0.92950$ & $0.85714$ & $0.072361$ & $39.887$ & $0.55796$ & $0.55796$ & $7.5849\times 10^{-7}$ \\
$10.941$ & $0.79023$ & $0.45400$ & $0.33623$ & $14.948$ & $0.90062$ & $0.80468$ & $0.095936$ & -- & -- & -- & -- \\
$11.940$ & $0.76854$ & $0.41086$ & $0.35769$ & $19.924$ & $0.85925$ & $0.72804$ & $0.13121$ & -- & -- & -- & -- \\
$14.914$ & $0.70797$ & $0.30511$ & $0.40286$ & $24.900$ & $0.81887$ & $0.65793$ & $0.16094$ & -- & -- & -- & -- \\
$19.880$ & $0.61796$ & $0.18565$ & $0.43231$ & $29.878$ & $0.78306$ & $0.59659$ & $0.18647$ & -- & -- & -- & -- \\
$24.847$ & $0.54144$ & $0.11336$ & $0.42808$ & $34.858$ & $0.74921$ & $0.54124$ & $0.20797$ & -- & -- & -- & -- \\
$29.814$ & $0.47533$ & $0.069355$ & $0.40598$ & $39.841$ & $0.71994$ & $0.49315$ & $0.22679$ & -- & -- & -- & -- \\
$39.756$ & $0.36929$ & $0.026153$ & $0.34314$ & $44.826$ & $0.69814$ & $0.45343$ & $0.24471$ & -- & -- & -- & -- \\
$49.693$ & $0.29391$ & $0.010107$ & $0.28380$ & $49.807$ & $0.67931$ & $0.41835$ & $0.26096$ & -- & -- & -- & --   
\\  
\hline
\end{tabular}
}
\end{scriptsize}
\caption{Computation of zero coupon bond prices for the USD, EUR and JPY markets as of April 22${}^{nd}$, 2015. $T$ refers to the maturity of the bonds. RNA and BA stand for risk neutral and benchmark approach respectively.}\label{ZCBUSD201504}
\end{table}

\begin{table}[htp]
\centering
\begin{scriptsize}
\rotatebox{90}{
\begin{tabular}{rrrr|rrrr|rrrr}
\multicolumn{4}{c}{USD}&\multicolumn{4}{c}{EUR}&\multicolumn{4}{c}{JPY}\\
$T$	& RNA & BA & Difference&$T$	& RNA & BA & Difference&$T$	& RNA & BA & Difference\\
\hline
$0.25563$ & $0.99927$ & $0.99927$ & $1.2972\times 10^{-10}$ & $0.51118$ & $0.99972$ & $0.99972$ & $2.4469\times 10^{-10}$ & $0.51072$ & $0.99930$ & $0.99930$ & $1.7574\times 10^{-7}$ \\
$0.30270$ & $0.99909$ & $0.99909$ & $2.2813\times 10^{-10}$ & $0.58823$ & $0.99967$ & $0.99967$ & $3.2072\times 10^{-9}$ & $0.59164$ & $0.99917$ & $0.99917$ & $1.1215\times 10^{-7}$ \\
$0.39981$ & $0.99873$ & $0.99873$ & $3.3414\times 10^{-9}$ & $0.66611$ & $0.99961$ & $0.99961$ & $2.3930\times 10^{-8}$ & $0.67107$ & $0.99906$ & $0.99906$ & $2.8733\times 10^{-8}$ \\
$0.47462$ & $0.99839$ & $0.99839$ & $1.2277\times 10^{-7}$ & $0.74977$ & $0.99954$ & $0.99954$ & $1.3172\times 10^{-7}$ & $0.76137$ & $0.99894$ & $0.99894$ & $8.5841\times 10^{-9}$ \\
$0.55216$ & $0.99803$ & $0.99803$ & $9.3915\times 10^{-7}$ & $0.84147$ & $0.99946$ & $0.99946$ & $5.8497\times 10^{-7}$ & $0.84044$ & $0.99883$ & $0.99883$ & $6.2951\times 10^{-11}$ \\
$0.64661$ & $0.99752$ & $0.99752$ & $5.7385\times 10^{-6}$ & $0.91575$ & $0.99939$ & $0.99939$ & $1.5835\times 10^{-6}$ & $0.91989$ & $0.99871$ & $0.99871$ & $2.0474\times 10^{-11}$ \\
$0.72154$ & $0.99704$ & $0.99702$ & $1.7472\times 10^{-5}$ & $0.99921$ & $0.99930$ & $0.99930$ & $4.0873\times 10^{-6}$ & $1.0021$ & $0.99860$ & $0.99860$ & $7.9620\times 10^{-12}$ \\
$0.80203$ & $0.99653$ & $0.99648$ & $4.6130\times 10^{-5}$ & $1.0858$ & $0.99922$ & $0.99921$ & $9.4268\times 10^{-6}$ & $1.5090$ & $0.99785$ & $0.99785$ & $4.5244\times 10^{-11}$ \\
$1.0486$ & $0.99455$ & $0.99418$ & $0.00036928$ & $1.1661$ & $0.99914$ & $0.99912$ & $1.8388\times 10^{-5}$ & $2.0104$ & $0.99701$ & $0.99701$ & $9.6194\times 10^{-11}$ \\
$1.3146$ & $0.99187$ & $0.99036$ & $0.0015073$ & $1.2496$ & $0.99904$ & $0.99901$ & $3.3775\times 10^{-5}$ & $3.0069$ & $0.99483$ & $0.99483$ & $6.4644\times 10^{-11}$ \\
$1.5638$ & $0.98885$ & $0.98512$ & $0.0037304$ & $1.3324$ & $0.99894$ & $0.99888$ & $5.7406\times 10^{-5}$ & $4.0032$ & $0.99149$ & $0.99149$ & $2.1211\times 10^{-11}$ \\
$1.7964$ & $0.98558$ & $0.97855$ & $0.0070238$ & $1.4163$ & $0.99884$ & $0.99875$ & $9.2592\times 10^{-5}$ & $5.0020$ & $0.98657$ & $0.98657$ & $1.2929\times 10^{-11}$ \\
$2.0620$ & $0.98132$ & $0.96903$ & $0.012290$ & $1.5054$ & $0.99871$ & $0.99857$ & $0.00014544$ & $6.0004$ & $0.98003$ & $0.98003$ & $8.1002\times 10^{-13}$ \\
$2.3080$ & $0.97698$ & $0.95844$ & $0.018546$ & $1.9989$ & $0.99782$ & $0.99694$ & $0.00088618$ & $6.9983$ & $0.97171$ & $0.97171$ & $2.4473\times 10^{-11}$ \\
$2.5569$ & $0.97232$ & $0.94622$ & $0.026100$ & $3.0030$ & $0.99467$ & $0.98858$ & $0.0060964$ & $8.0040$ & $0.96233$ & $0.96233$ & $1.7680\times 10^{-10}$ \\
$2.8083$ & $0.96736$ & $0.93253$ & $0.034826$ & $3.9973$ & $0.98886$ & $0.97216$ & $0.016701$ & $8.9986$ & $0.95198$ & $0.95198$ & $2.6855\times 10^{-11}$ \\
$3.0543$ & $0.96217$ & $0.91792$ & $0.044254$ & $4.9901$ & $0.98026$ & $0.94887$ & $0.031388$ & $9.9929$ & $0.94066$ & $0.94066$ & $8.7440\times 10^{-10}$ \\
$3.9879$ & $0.94210$ & $0.85656$ & $0.085549$ & $5.9843$ & $0.96919$ & $0.92068$ & $0.048509$ & $11.986$ & $0.91316$ & $0.91316$ & $5.4188\times 10^{-9}$ \\
$4.9787$ & $0.91811$ & $0.78486$ & $0.13325$ & $6.9776$ & $0.95599$ & $0.88931$ & $0.066678$ & $14.977$ & $0.86306$ & $0.86306$ & $2.2262\times 10^{-8}$ \\
$5.9718$ & $0.89314$ & $0.71296$ & $0.18019$ & $7.9698$ & $0.94126$ & $0.85627$ & $0.084983$ & $19.955$ & $0.77348$ & $0.77348$ & $1.2978\times 10^{-7}$ \\
$6.9647$ & $0.86791$ & $0.64427$ & $0.22365$ & $8.9641$ & $0.92550$ & $0.82257$ & $0.10293$ & $24.934$ & $0.69601$ & $0.69601$ & $3.1878\times 10^{-7}$ \\
$7.9656$ & $0.84271$ & $0.57986$ & $0.26285$ & $9.9525$ & $0.90943$ & $0.78936$ & $0.12007$ & $29.915$ & $0.62979$ & $0.62979$ & $5.4640\times 10^{-7}$ \\
$8.9526$ & $0.81802$ & $0.52146$ & $0.29655$ & $10.943$ & $0.89252$ & $0.75623$ & $0.13629$ & $34.883$ & $0.56855$ & $0.56855$ & $1.0099\times 10^{-6}$ \\
$9.9422$ & $0.79335$ & $0.46797$ & $0.32538$ & $11.934$ & $0.87604$ & $0.72442$ & $0.15162$ & $39.860$ & $0.51170$ & $0.51170$ & $1.4134\times 10^{-6}$ \\
$10.935$ & $0.76994$ & $0.41988$ & $0.35006$ & $14.909$ & $0.82897$ & $0.63676$ & $0.19221$ & -- & -- & -- & -- \\
$11.927$ & $0.74631$ & $0.37613$ & $0.37018$ & $19.869$ & $0.76070$ & $0.51638$ & $0.24432$ & -- & -- & -- & -- \\
$14.907$ & $0.68031$ & $0.27023$ & $0.41008$ & $24.831$ & $0.70451$ & $0.42255$ & $0.28196$ & -- & -- & -- & -- \\
$19.871$ & $0.58288$ & $0.15553$ & $0.42735$ & $29.788$ & $0.65375$ & $0.34649$ & $0.30726$ & -- & -- & -- & -- \\
$24.834$ & $0.50311$ & $0.090162$ & $0.41295$ & $34.749$ & $0.60524$ & $0.28342$ & $0.32181$ & -- & -- & -- & -- \\
$29.800$ & $0.43411$ & $0.052238$ & $0.38188$ & $39.716$ & $0.56625$ & $0.23426$ & $0.33199$ & -- & -- & -- & -- \\
$39.721$ & $0.32779$ & $0.017799$ & $0.30999$ & $44.688$ & $0.53578$ & $0.19579$ & $0.33999$ & -- & -- & -- & -- \\
$49.657$ & $0.25451$ & $0.0062285$ & $0.24828$ & $49.662$ & $0.50802$ & $0.16398$ & $0.34404$ & -- & -- & -- & --   
\\  
\hline
\end{tabular}
}
\end{scriptsize}
\caption{Computation of zero coupon bond prices for the USD, EUR and JPY markets as of May 22${}^{nd}$, 2015. $T$ refers to the maturity of the bonds. RNA and BA stand for risk neutral and benchmark approach respectively.}\label{ZCBUSD201505}
\end{table}

\begin{table}[htp]
\centering
\begin{scriptsize}
\rotatebox{90}{
\begin{tabular}{rrrr|rrrr|rrrr}
\multicolumn{4}{c}{USD}&\multicolumn{4}{c}{EUR}&\multicolumn{4}{c}{JPY}\\
$T$	& RNA & BA & Difference&$T$	& RNA & BA & Difference&$T$	& RNA & BA & Difference\\
\hline
$0.25554$ & $0.99928$ & $0.99928$ & $1.0829\times 10^{-10}$ & $0.50840$ & $0.99976$ & $0.99976$ & $2.2031\times 10^{-10}$ & $0.50820$ & $0.99931$ & $0.99931$ & $1.6587\times 10^{-7}$ \\
$0.32460$ & $0.99903$ & $0.99903$ & $6.2311\times 10^{-10}$ & $0.58660$ & $0.99968$ & $0.99968$ & $3.1197\times 10^{-9}$ & $0.58865$ & $0.99920$ & $0.99920$ & $1.0811\times 10^{-7}$ \\
$0.39953$ & $0.99872$ & $0.99872$ & $3.4141\times 10^{-9}$ & $0.66657$ & $0.99961$ & $0.99961$ & $2.4997\times 10^{-8}$ & $0.67119$ & $0.99908$ & $0.99908$ & $2.6756\times 10^{-8}$ \\
$0.47725$ & $0.99838$ & $0.99838$ & $1.3855\times 10^{-7}$ & $0.74954$ & $0.99953$ & $0.99953$ & $1.3705\times 10^{-7}$ & $0.75075$ & $0.99897$ & $0.99897$ & $6.3637\times 10^{-9}$ \\
$0.57165$ & $0.99791$ & $0.99791$ & $1.4954\times 10^{-6}$ & $0.83597$ & $0.99945$ & $0.99945$ & $5.6901\times 10^{-7}$ & $0.83837$ & $0.99885$ & $0.99885$ & $1.4196\times 10^{-11}$ \\
$0.64663$ & $0.99747$ & $0.99747$ & $5.9978\times 10^{-6}$ & $0.91688$ & $0.99937$ & $0.99937$ & $1.7041\times 10^{-6}$ & $0.91767$ & $0.99874$ & $0.99874$ & $1.5913\times 10^{-12}$ \\
$0.72707$ & $0.99698$ & $0.99696$ & $1.9628\times 10^{-5}$ & $0.99976$ & $0.99931$ & $0.99931$ & $4.3878\times 10^{-6}$ & $1.0022$ & $0.99862$ & $0.99862$ & $7.8938\times 10^{-12}$ \\
$0.97378$ & $0.99511$ & $0.99489$ & $0.00022707$ & $1.0851$ & $0.99922$ & $0.99921$ & $1.0046\times 10^{-5}$ & $1.5147$ & $0.99786$ & $0.99786$ & $2.7632\times 10^{-10}$ \\
$1.2397$ & $0.99255$ & $0.99143$ & $0.0011145$ & $1.1662$ & $0.99911$ & $0.99909$ & $1.9817\times 10^{-5}$ & $2.0077$ & $0.99702$ & $0.99702$ & $4.8335\times 10^{-11}$ \\
$1.4890$ & $0.98962$ & $0.98659$ & $0.0030258$ & $1.2545$ & $0.99898$ & $0.99894$ & $3.7722\times 10^{-5}$ & $3.0041$ & $0.99476$ & $0.99476$ & $2.4972\times 10^{-11}$ \\
$1.7213$ & $0.98641$ & $0.98041$ & $0.0060010$ & $1.3331$ & $0.99887$ & $0.99880$ & $6.2478\times 10^{-5}$ & $4.0003$ & $0.99123$ & $0.99123$ & $6.3061\times 10^{-14}$ \\
$1.9870$ & $0.98220$ & $0.97128$ & $0.010919$ & $1.4162$ & $0.99872$ & $0.99862$ & $0.00010050$ & $5.0021$ & $0.98608$ & $0.98608$ & $1.9075\times 10^{-11}$ \\
$2.2331$ & $0.97788$ & $0.96100$ & $0.016881$ & $1.5073$ & $0.99856$ & $0.99840$ & $0.00015985$ & $6.0002$ & $0.97940$ & $0.97940$ & $4.3776\times 10^{-13}$ \\
$2.4819$ & $0.97320$ & $0.94904$ & $0.024159$ & $2.0046$ & $0.99737$ & $0.99638$ & $0.00098958$ & $6.9981$ & $0.97137$ & $0.97137$ & $2.8516\times 10^{-11}$ \\
$2.7332$ & $0.96819$ & $0.93556$ & $0.032629$ & $2.9990$ & $0.99278$ & $0.98607$ & $0.0067067$ & $8.0013$ & $0.96216$ & $0.96216$ & $1.9938\times 10^{-11}$ \\
$2.9790$ & $0.96292$ & $0.92110$ & $0.041819$ & $3.9921$ & $0.98453$ & $0.96612$ & $0.018405$ & $8.9960$ & $0.95177$ & $0.95177$ & $2.2599\times 10^{-10}$ \\
$3.2249$ & $0.95739$ & $0.90564$ & $0.051755$ & $4.9861$ & $0.97255$ & $0.93798$ & $0.034571$ & $9.9931$ & $0.94151$ & $0.94151$ & $8.1510\times 10^{-10}$ \\
$3.9844$ & $0.93952$ & $0.85396$ & $0.085560$ & $5.9785$ & $0.95766$ & $0.90445$ & $0.053210$ & $11.986$ & $0.91443$ & $0.91443$ & $5.0668\times 10^{-9}$ \\
$4.9774$ & $0.91356$ & $0.78121$ & $0.13235$ & $6.9693$ & $0.94047$ & $0.86769$ & $0.072777$ & $14.975$ & $0.86546$ & $0.86546$ & $2.5412\times 10^{-8}$ \\
$5.9699$ & $0.88650$ & $0.70875$ & $0.17775$ & $7.9648$ & $0.92159$ & $0.82921$ & $0.092371$ & $19.954$ & $0.77896$ & $0.77896$ & $1.2903\times 10^{-7}$ \\
$6.9621$ & $0.85913$ & $0.63974$ & $0.21939$ & $8.9487$ & $0.90216$ & $0.79108$ & $0.11109$ & $24.934$ & $0.70246$ & $0.70246$ & $3.0852\times 10^{-7}$ \\
$7.9596$ & $0.83194$ & $0.57542$ & $0.25652$ & $9.9375$ & $0.88237$ & $0.75344$ & $0.12893$ & $29.915$ & $0.63652$ & $0.63652$ & $5.2877\times 10^{-7}$ \\
$8.9460$ & $0.80541$ & $0.51711$ & $0.28831$ & $10.926$ & $0.86257$ & $0.71694$ & $0.14562$ & $34.887$ & $0.57692$ & $0.57692$ & $9.1025\times 10^{-7}$ \\
$9.9378$ & $0.77937$ & $0.46388$ & $0.31549$ & $11.915$ & $0.84301$ & $0.68188$ & $0.16113$ & $39.864$ & $0.52026$ & $0.52026$ & $1.2892\times 10^{-6}$ \\
$10.929$ & $0.75401$ & $0.41581$ & $0.33820$ & $14.881$ & $0.78758$ & $0.58665$ & $0.20093$ & -- & -- & -- & -- \\
$11.921$ & $0.72944$ & $0.37257$ & $0.35687$ & $19.832$ & $0.70937$ & $0.46013$ & $0.24924$ & -- & -- & -- & -- \\
$14.896$ & $0.66049$ & $0.26774$ & $0.39275$ & $24.786$ & $0.64764$ & $0.36572$ & $0.28192$ & -- & -- & -- & -- \\
$19.858$ & $0.56111$ & $0.15453$ & $0.40657$ & $29.744$ & $0.59288$ & $0.29143$ & $0.30145$ & -- & -- & -- & -- \\
$24.819$ & $0.47879$ & $0.089580$ & $0.38921$ & $34.694$ & $0.54186$ & $0.23191$ & $0.30994$ & -- & -- & -- & -- \\
$29.782$ & $0.40948$ & $0.052033$ & $0.35745$ & $39.652$ & $0.49796$ & $0.18552$ & $0.31244$ & -- & -- & -- & -- \\
$39.701$ & $0.30345$ & $0.017796$ & $0.28565$ & $44.617$ & $0.46386$ & $0.15040$ & $0.31346$ & -- & -- & -- & -- \\
$49.632$ & $0.23032$ & $0.0062287$ & $0.22409$ & $49.582$ & $0.43275$ & $0.12212$ & $0.31064$ & -- & -- & -- & --   
\\  
\hline
\end{tabular}
}
\end{scriptsize}
\caption{Computation of zero coupon bond prices for the USD, EUR and JPY markets as of June 22${}^{nd}$, 2015. $T$ refers to the maturity of the bonds. RNA and BA stand for risk neutral and benchmark approach respectively.}\label{ZCBUSD201506}
\end{table}

\begin{table}[h]
\centering
\begin{tabular}{lccl}
$T$	& Risk Neutral Approach & Benchmark approach & Difference\\
\hline
$1$ week  & $0.99995$ & $0.99995$ & $5.2625\times 10^{-14}$ \\
$2$ weeks  & $0.99989$ & $0.99989$ & $7.1054\times 10^{-14}$ \\
$3$ weeks & $0.99984$ & $0.99984$ & $3.4195\times 10^{-14}$ \\
$1$ month & $0.99977$ & $0.99977$ & $3.6859\times 10^{-14}$ \\
$2$ months & $0.99953$ & $0.99953$ & $4.8850\times 10^{-14}$ \\
$3$ months & $0.99930$ & $0.99930$ & $4.5419\times 10^{-13}$ \\
$6$ months & $0.99832$ & $0.99832$ & $2.7218\times 10^{-7}$ \\
$1$ year & $0.99502$ & $0.99474$ & $0.00027569$ \\
$1.5$ years & $0.98966$ & $0.98652$ & $0.0031411$ \\
$2$ years& $0.98224$ & $0.97103$ & $0.011205$ \\
$3$ years& $0.96295$ & $0.92028$ & $0.042670$ \\
$5$ years& $0.91380$ & $0.78027$ & $0.13353$ \\
$7$ years& $0.85910$ & $0.63795$ & $0.22115$ \\
$10$ years& $0.77870$ & $0.46128$ & $0.31742$ \\  
\hline
\end{tabular}
\vspace{0.3cm}
\caption{Computation of zero coupon bond prices for the USD market as of June 22${}^{nd}$, 2015, where risk neutral valuation is not possible. The associated price difference increases as the maturity becomes larger.}\label{ZCBUSD2015}
\end{table}
%

\clearpage
\bibliography{biblio}

\begin{thebibliography}{}

\bibitem[\protect\astroncite{Abramowitz and Stegun}{1965}]{book_AbramStegun}
Abramowitz, M. and Stegun, I.~A. (1965).
\newblock {\em {Handbook of Mathematical Functions with Formulas, Graphs, and
  Mathematical Tables}}.
\newblock Dover Books on Mathematics. Dover Publications, New York, first
  edition.

\bibitem[\protect\astroncite{Baldeaux et~al.}{2015a}]{bfip14}
Baldeaux, J., Fung, M.~C., Ignatieva, K., and Platen, E. (2015a).
\newblock A hybrid model for pricing and hedging of long dated bonds.
\newblock {\em Applied Mathematical Finance, to appear}.

\bibitem[\protect\astroncite{Baldeaux et~al.}{2015b}]{martinoplaten13}
Baldeaux, J., Grasselli, M., and Platen, E. (2015b).
\newblock {Pricing currency derivatives under the benchmark approach}.
\newblock {\em Journal of Banking and Finance}, 53(0):34--48.

\bibitem[\protect\astroncite{Baldeaux and Platen}{2013}]{bookjanplaten13}
Baldeaux, J. and Platen (2013).
\newblock {\em {Functionals of Multidimensional Diffusions with Applications to
  Finance}}.
\newblock Bocconi \& Springer Series, Vol. 5.

\bibitem[\protect\astroncite{Biagini et~al.}{2014}]{bcp14}
Biagini, F., Cretarola, A., and Platen, E. (2014).
\newblock Local risk-minimization via the benchmark approach.
\newblock {\em Mathematics and Financial Economics}, 8(2):109--134.

\bibitem[\protect\astroncite{Bouleau and Lamberton}{1989}]{bola89}
Bouleau, N. and Lamberton, D. (1989).
\newblock {Residual risks and hedging strategies in Markovian markets}.
\newblock {\em Stochastic Processes and their Applications}, 33(1):131--150.

\bibitem[\protect\astroncite{{Carr} and {Madan}}{1999}]{article_Carr99}
{Carr}, P. and {Madan}, D.~B. (1999).
\newblock Option valuation using the fast {Fourier} transform.
\newblock {\em Journal of Computational Finance}, 2:61--73.

\bibitem[\protect\astroncite{Craddock and Lennox}{2009}]{cradlennox09}
Craddock, M. and Lennox, K. (2009).
\newblock {The calculation of expectations for classes of diffusion processes
  by Lie symmetry methods}.
\newblock {\em Annals of Applied Probability}, 19(1):127--157.

\bibitem[\protect\astroncite{De~Col et~al.}{2013}]{gnoatto11}
De~Col, A., Gnoatto, A., and Grasselli, M. (2013).
\newblock {Smiles all around: FX joint calibration in a Multi-Heston model}.
\newblock {\em Journal of Banking and Finance}, 37(10):3799--3818.

\bibitem[\protect\astroncite{Du and Platen}{2016}]{duplaten14}
Du, K. and Platen, E. (2016).
\newblock {Benchmarked risk minimization}.
\newblock {\em Mathematical Finance}, 26(3):617--637.

\bibitem[\protect\astroncite{Duffee}{2002}]{duf02}
Duffee, G.~R. (2002).
\newblock Term premia and interest rate forecasts in affine models.
\newblock {\em Journal of Finance}, 57:405--443.

\bibitem[\protect\astroncite{Duffie et~al.}{2003}]{article_DFS}
Duffie, D., Filipovic, D., and Schachermayer, W. (2003).
\newblock Affine processes and applications in finance.
\newblock {\em Annals of Applied Probability}, 13:984--1053.

\bibitem[\protect\astroncite{Duffie and Richardson}{1991}]{duri91}
Duffie, D. and Richardson, H. (1991).
\newblock Mean-variance hedging in continuous time.
\newblock {\em The Annals of Applied Probability}, 1(1):1--15.

\bibitem[\protect\astroncite{F\"ollmer and Sondermann}{1986}]{foso86}
F\"ollmer, H. and Sondermann, D. (1986).
\newblock Hedging of non-redundant contingent claims.
\newblock In Hildenbrand, W. and Mas-Colell, A., editors, {\em Contributions to
  Financial Economics}, pages 205--223. North-Holland.

\bibitem[\protect\astroncite{Gnoatto and Grasselli}{2014}]{gg13}
Gnoatto, A. and Grasselli, M. (2014).
\newblock An affine multi-currency model with stochastic volatility and
  stochastic interest rates.
\newblock {\em SIAM Journal on Financial Mathematics}, 5(1):493--531.

\bibitem[\protect\astroncite{Gourinchas and Parker}{2002}]{Gourinchas02}
Gourinchas, P. and Parker, J.~A. (2002).
\newblock {Consumption over the life cycle}.
\newblock {\em Econometrica}, 70(1):47--89.

\bibitem[\protect\astroncite{Grasselli}{2017}]{grasselli13}
Grasselli, M. (2017).
\newblock {The 4/2 stochastic volatility model}.
\newblock {\em Mathematical Finance}, 27(4):1013--1034.

\bibitem[\protect\astroncite{Heath and Platen}{2006}]{currencyplaten06}
Heath, D. and Platen, E. (2006).
\newblock Currency derivatives under a minimal market model.
\newblock {\em The ICFAI Journal of Derivatives Markets}, 3:68--86.

\bibitem[\protect\astroncite{Heston}{1997}]{heston97}
Heston, S. (1997).
\newblock {A simple new formula for options with stochastic volatility}.
\newblock {\em Technical report, Washington University of St. Louis}.

\bibitem[\protect\astroncite{Heston et~al.}{2007}]{Loewenstein07}
Heston, S., Loewenstein, M., and Willard, G.~A. (2007).
\newblock Options and bubbles.
\newblock {\em The Review of Financial Studies}, 20(2):359--390.

\bibitem[\protect\astroncite{Heston}{1993}]{Heston93}
Heston, S.~L. (1993).
\newblock A closed-form solution for options with stochastic volatility with
  applications to bond and currency options.
\newblock {\em Review of Financial Studies}, 6:327--343.

\bibitem[\protect\astroncite{Hulley and Ruf}{2015}]{ruf15}
Hulley, H. and Ruf, J. (2015).
\newblock {Weak tail conditions for local martingales}.
\newblock {\em Working Paper, available at http://arxiv.org/abs/1508.07564}.

\bibitem[\protect\astroncite{Karatzas and Shreve}{1991}]{kar91}
Karatzas, I. and Shreve, S.~E. (1991).
\newblock {\em {Brownian Motion and Stochastic Calculus}}.
\newblock Springer-Verlag.

\bibitem[\protect\astroncite{Lee}{2004}]{lee2004}
Lee, R. (2004).
\newblock Option pricing by transform methods: Extensions, unification and
  error control.
\newblock {\em Journal of Computational Finance}, 7(3):51--86.

\bibitem[\protect\astroncite{Lewis}{2001}]{lewis2001}
Lewis, A. (2001).
\newblock {A simple option formula for general jump-diffusion and other
  exponential L\'{e}vy processes}.
\newblock {\em Envision Financial Systems and {OptionCity.net}}.

\bibitem[\protect\astroncite{Loewenstein and Willard}{2000}]{Loewenstein00}
Loewenstein, M. and Willard, G.~A. (2000).
\newblock Local martingales, arbitrage, and viability: Free snacks and cheap
  thrills.
\newblock {\em Economic Theory}, 2(16):135--161.

\bibitem[\protect\astroncite{Long}{1990}]{long90}
Long, J. (1990).
\newblock {The num\' eraire portfolio}.
\newblock {\em Journal of Financial Economics}, 26(1):29--69.

\bibitem[\protect\astroncite{Lukacs}{1970}]{lukacs}
Lukacs, B. (1970).
\newblock {\em Characteristic Functions}.
\newblock Griffin, London, 2nd edition.

\bibitem[\protect\astroncite{M{\o}ller}{2001}]{mol01}
M{\o}ller, T. (2001).
\newblock Risk-minimizing hedging strategies for insurance payment processes.
\newblock {\em Finance and Stochastics}, 5(4):419--446.

\bibitem[\protect\astroncite{Platen}{1997}]{platen97}
Platen, E. (1997).
\newblock {A non-linear stochastic volatility model}.
\newblock {\em Financial Mathematics Research Report No. FMRR 005-97, Center
  for Financial Mathematics, Australian National University, Canberra.}

\bibitem[\protect\astroncite{Platen and Heath}{2010}]{bookplaten10}
Platen, E. and Heath, D. (2010).
\newblock {\em {A Benchmark Approach to Quantitative Finance}}.
\newblock Springer-Verlag.

\bibitem[\protect\astroncite{Schweizer}{1991}]{schweizer91}
Schweizer, M. (1991).
\newblock Option hedging for semimartingales.
\newblock {\em Stochastic Processes and Their Applications}, 37(2):339--363.

\bibitem[\protect\astroncite{Schweizer}{1994}]{schweizer94a}
Schweizer, M. (1994).
\newblock Approximating random variables by stochastic integrals.
\newblock {\em The Annals of Probability}, 22(3):1536--1575.

\bibitem[\protect\astroncite{Schweizer}{2001}]{schweizer01}
Schweizer, M. (2001).
\newblock A guided tour through quadratic hedging approaches.
\newblock In Jouini, E., Cvitani\'c, J., and Musiela, M., editors, {\em Option
  Pricing, Interest Rates and Risk Management}, pages 538--574. Cambridge
  University Press.

\end{thebibliography}
\bibliographystyle{apa}

\end{document}